\documentclass{article}
\usepackage{kr}

\usepackage{times}
\usepackage{soul}
\usepackage{url}
\usepackage[hidelinks]{hyperref}
\usepackage[utf8]{inputenc}
\usepackage[small]{caption}
\usepackage{graphicx}
\usepackage{amsmath}
\usepackage{amsthm}
\usepackage{booktabs}
\usepackage{algorithm}
\usepackage{algorithmic}
\usepackage{amsmath, latexsym,amsthm}
\usepackage{amsfonts, amssymb, stmaryrd}
\usepackage{xspace}
\usepackage{multirow}
\usepackage{xcolor}

\usepackage{thm-restate} 

\renewcommand\vec{\mathbf}  

\usepackage{fancyhdr}
\fancypagestyle{firstpage}
{
    \fancyhead[L]{Long version of the KR'23 paper with the same title. This version contains an appendix with full proofs and also corrects an error in Table 1 of the published paper, which misreported existing results for the case of denial constraints under brave and intersection semantics, as well as  glitches in point 3 of Proposition 2 (see Remark 2) and the running example of Section 3.2.}    
    \fancyhead[R]{}
}

\usepackage{tikz}
\tikzset{>=latex}

\newtheorem{example}{Example}
\newtheorem{lemma}{Lemma}
\newtheorem{claim}{Claim}
\newtheorem{proposition}{Proposition}

\newtheorem{definition}{Definition}
\newtheorem{remark}{Remark}

\allowdisplaybreaks

\def\aicbody{\mathsf{body}}
\def\aiclits{\mathsf{lits}}
\def\aicup{\mathsf{upd}}
\def\fix{\mi{fix}}

\newcommand{\ans}{\vec{a}}

\newcommand{\inds}{{\bf{C}}}
\newcommand{\vars}{{\bf{V}}}
\newcommand{\preds}{{\bf{S}}}
\newcommand{\domain}[1]{\mn{dom}(#1)}
\newcommand{\constraints}{\Cmc}

\newcommand{\factset}{\mi{Facts}^\preds_\Dmc}

\newcommand{\litset}{\mi{Lits}^\preds_\Dmc}

\newcommand{\conf}{\Emc}
\newcommand{\confgraph}[2]{\Gmc^{#2}_{#1}}
\newcommand{\tofacts}{\mi{facts}}
\newcommand{\mingr}{\mi{min}_{g}}
\newcommand{\minnongr}{\mi{min}}

\newcommand{\bravemodels}[1]{\models_{\text{brave}}^{#1}}
\newcommand{\armodels}[1]{\models_{\text{CQA}}^{#1}}
\newcommand{\iarmodels}[1]{\models_{\cap}^{#1}}

\newcommand{\conflicts}[1]{\mi{Conf}(#1)}
\newcommand{\deltaconflicts}[1]{\mi{Conf}(#1)}
\newcommand{\mhs}[1]{\mi{MHS}(#1)}

\newcommand{\deltareps}[1]{\mi{SRep}(#1)}
\newcommand{\deltagreps}[1]{\mi{GRep}(#1)}
\newcommand{\deltapreps}[1]{\mi{PRep}(#1)}
\newcommand{\deltacreps}[1]{\mi{CRep}(#1)}
\newcommand{\deltaxreps}[1]{\mi{XRep}(#1)}
\newcommand{\deltascorereps}[1]{\mi{LRep}(#1)}

\newcommand{\comp}[2]{\mi{Int}_{#1}(#2)}
\newcommand{\restr}[2]{\mi{Dat}_{#1}(#2)}

\newcommand{\ups}[1]{\mi{Up}(#1)}
\newcommand{\xups}[1]{\mi{XUp}(#1)}
\newcommand{\foundups}[1]{\mi{FoundUp}(#1)}
\newcommand{\justifups}[1]{\mi{JustUp}(#1)}
\newcommand{\foundreps}[1]{\mi{FoundRep}(#1)}
\newcommand{\justifreps}[1]{\mi{JustRep}(#1)}
\newcommand{\wellfoundups}[1]{\mi{WellFoundUp}(#1)}
\newcommand{\wellfoundreps}[1]{\mi{WellFoundRep}(#1)}
\newcommand{\groundups}[1]{\mi{GroundUp}(#1)}
\newcommand{\groundreps}[1]{\mi{GroundRep}(#1)}

\def\ptime{\textsc{P}\xspace}
\def\np{\textsc{NP}\xspace}
\def\conp{co\textsc{NP}\xspace}
\def\piptwo{\ensuremath{\Pi^{p}_{2}}\xspace}
\def\sigmaptwo{\ensuremath{\Sigma^{p}_{2}}\xspace}
\def\deltaptwo{\ensuremath{\Delta^{p}_{2}}\xspace}

\def\no{\neg}

\newcommand{\mn}[1]{\ensuremath{\mathsf{#1}}}
\newcommand{\mi}[1]{\ensuremath{\mathit{#1}}}

\newcommand{\Bmc}{\ensuremath{\mathcal{B}}}
\newcommand{\Cmc}{\ensuremath{\mathcal{C}}}
\newcommand{\Dmc}{\ensuremath{\mathcal{D}}}
\newcommand{\Emc}{\ensuremath{\mathcal{E}}}

\newcommand{\Gmc}{\ensuremath{\mathcal{G}}}
\newcommand{\Hmc}{\ensuremath{\mathcal{H}}}
\newcommand{\Imc}{\ensuremath{\mathcal{I}}}

\newcommand{\Mmc}{\ensuremath{\mathcal{M}}}

\newcommand{\Rmc}{\ensuremath{\mathcal{R}}}
\newcommand{\Smc}{\ensuremath{\mathcal{S}}}

\newcommand{\Umc}{\ensuremath{\mathcal{U}}}
\newcommand{\Vmc}{\ensuremath{\mathcal{V}}}

\newcommand{\eg}{e.g.,~}

\newcommand{\ie}{i.e.,~}
\newcommand{\wrt}{w.r.t.~}
\newcommand{\cf}{cf.~}

\title{Inconsistency Handling in Prioritized Databases with Universal Constraints:\\ Complexity Analysis and Links with Active Integrity Constraints}

\author{%
Meghyn Bienvenu$^1$\and
Camille Bourgaux$^2$
\affiliations
$^1$ CNRS, University of Bordeaux, Bordeaux INP, LaBRI, Talence, France\\
$^2$ DI ENS, ENS, CNRS, PSL University \& Inria, Paris, France\\  
\emails
meghyn.bienvenu@labri.fr,
camille.bourgaux@ens.fr
}

\begin{document}
\maketitle
\thispagestyle{firstpage}
\begin{abstract}
This paper revisits the problem of repairing and querying inconsistent databases equipped with universal constraints. 
We adopt symmetric difference repairs, in which both deletions and additions of facts can be used to restore consistency, and suppose that preferred repair actions are specified via a binary priority relation over (negated) facts. 
Our first contribution is to show how existing notions of optimal repairs, defined for simpler denial constraints and repairs solely based on fact deletion, can be suitably extended to our richer setting. We next study the computational properties of the resulting repair notions, in particular, the data complexity of repair checking and inconsistency-tolerant query answering. 
Finally, we clarify the relationship between optimal repairs of prioritized databases and repair notions introduced in the framework of active integrity constraints. 
In particular, we show that Pareto-optimal repairs in our setting correspond to founded, grounded and justified repairs \wrt the active integrity constraints obtained by translating the prioritized database. Our study also yields useful insights into the behavior of active integrity constraints. 
\end{abstract}

\section{Introduction}\label{sec:intro}
When a database is inconsistent w.r.t.\ the 
integrity constraints, 
it is possible to obtain meaningful query answers by adopting the \emph{consistent query answering} (CQA) approach \cite{ArenasBC99}.
In a nutshell, the idea is to consider a set of 
\emph{repairs}, which correspond to those databases that satisfy the constraints and are as close as possible to the original database. 
An answer is then considered true w.r.t.\ \emph{CQA semantics} if it holds no matter which repair is chosen, thus embodying 
the cautious mode of reasoning employed in many KR contexts. 
The CQA approach was subsequently extended to the setting of ontology-mediated query answering, 
which led to the proposal of other natural repair-based semantics, such as 
the \emph{brave semantics}, 
which considers as true those answers that hold in at least one repair \cite{Bienvenu_TractableApproximation_long}, 
and the \emph{intersection (or IAR) semantics},
which evaluates queries w.r.t.\ the intersection of the repairs \cite{LemboLRRS10}. 
There is now an extensive literature on CQA and other forms of inconsistency-tolerant query answering, 
\cite{DBLP:conf/pods/Bertossi19} and \cite{DBLP:journals/ki/Bienvenu20} provide recent surveys for the database and ontology settings. 

Several different notions of 
repair have been considered, depending on the considered class of constraints and the allowed repair actions.
For denial constraints (such as functional dependencies, FDs)
and constraints given by ontologies, consistency can be restored only by removing information, 
so subset repairs based upon fact deletions are the most common choice. 
For richer 
classes of constraints, however, 
it makes sense to consider 
\emph{symmetric difference repairs} obtained using both 
fact additions and deletions. 
This is the case for the \emph{universal constraints} considered in the present paper, which 
can be used to express data completeness assumptions and other kinds of domain knowledge. 
For example, in a hospital setting, a universal constraint can be used to capture expert knowledge 
that a patient cannot receive a certain treatment without a positive test for a given mutation, 
with violations indicating either an erroneous treatment record or missing test result. 
Universal constraints 
are one of the most expressive classes of first-order constraints 
for which CQA with symmetric difference repairs is decidable, albeit intractable: \piptwo-complete \wrt data complexity 
\cite{DBLP:journals/is/StaworkoC10,DBLP:conf/icdt/ArmingPS16}. 
Despite this high complexity, there have been some prototype implementations using 
 logic programming \cite{DBLP:journals/tods/EiterFGL08,DBLP:journals/dke/MarileoB10}.

Repairs can be further refined by taking into account information about the relative reliability of the database facts. 
In the framework of \emph{prioritized databases} \cite{DBLP:journals/amai/StaworkoCM12}, a binary 
\emph{priority relation} 
indicates preferences between pair of facts involved in some violation of a denial constraint. 
Three kinds of \emph{optimal repair} (Pareto-, globally-, and completion-optimal) are then 
defined to select the most preferred subset repairs according to the priority relation. 
The complexity of reasoning with these three kinds of optimal repair has been investigated, primarily focusing on databases with FDs 
 \cite{DBLP:conf/pods/FaginKK15,DBLP:conf/icdt/KimelfeldLP17,DBLP:conf/pods/LivshitsK17}, but also in the context of description logic knowledge bases \cite{DBLP:conf/kr/BienvenuB20}. 
A recent system implements SAT-based algorithms for optimal repair-based semantics having (co)\np-complete data complexity
\cite{DBLP:conf/kr/BienvenuB22}. 
 
To the best of our knowledge, there has been no work addressing how to define fact-level preferences for databases with universal constraints and how to exploit such preferences to single out the optimal symmetric difference repairs. 
Our first contribution is thus an extension of the framework of prioritized databases to the case of universal constraints and symmetric difference repairs. 
By carefully defining a suitable notion of conflict (which may involve negative facts), 
we are able to faithfully lift existing notions of optimal repairs and optimal repair-based semantics, 
while retaining many properties of the original framework. 

We next study the computational properties of optimal repairs of prioritized databases with universal constraints. 
We provide an almost-complete picture of the data complexity of repair checking and inconsistency-tolerant query answering 
for each of the three notions of optimal repair (Pareto, global, and completion) and 
three repair-based semantics (CQA, brave, and intersection). 
Our results show that adopting optimal repairs does not increase the complexity of inconsistency-tolerant query answering.

Our third contribution is to establish connections with active integrity constraints (AICs), a framework in which universal constraints are enriched with information on what are the allowed update actions (fact deletions or additions) to resolve a given constraint violation \cite{DBLP:conf/ppdp/FlescaGZ04,DBLP:conf/iclp/CaropreseGSZ06,DBLP:journals/tkde/CaropreseGZ09}. 
More precisely, we provide a natural translation from prioritized databases to AICs and observe that Pareto-optimal repairs 
coincide with three kinds of repairs (founded, grounded and justified) that have been defined for AICs. 
This leads us to explore more general conditions under which AIC repair notions coincide, which we subsequently exploit to exhibit
a translation of certain `well-behaved' sets of AICs into prioritized databases.

Proofs can be found
in the appendix.

\section{Preliminaries}
We assume familiarity with propositional and first-order logic (FOL) and provide here terminology and notation for databases, conjunctive queries, constraints, and repairs.

\subsubsection*{Relational databases} 
Let $\inds$ and $\vars$ be two disjoint countably infinite sets of constants and variables respectively. A (relational) \emph{schema} $\preds$ is a finite set of relation names (or \emph{predicates}), each with an associated arity $n > 0$. 
A \emph{fact} over $\preds$ is an expression of the form $P(c_1,\dots, c_n)$ where $P\in\preds$ has arity $n$ and $c_1,\dots, c_n\in\inds$. A \emph{database (instance)} over $\preds$ is a finite set $\Dmc$ of facts over $\preds$. The \emph{active domain} of $\Dmc$, denoted $\domain{\Dmc}$, is the set of constants occurring in $\Dmc$.  

A database $\Dmc$ can also be viewed as a finite relational structure whose domain is $\domain{\Dmc}$ and which interprets each predicate $P\in\preds$ as the set $\{\vec{c} \mid P(\vec{c}) \in \Dmc\}$. We shall use the standard notation $\Dmc \models \Phi$ to indicate that a (set of) FOL 
sentence(s) $\Phi$
is satisfied in this structure. 

\subsubsection*{Conjunctive queries} 
A \emph{conjunctive query} (CQ) is a conjunction of \emph{relational atoms} $P(t_1, \ldots, t_n)$ (with each 
$t_i\in\vars\cup\inds$
), where some variables may be existentially quantified. A \emph{Boolean} CQ (BCQ) has no free variables. 
Given a query $q(\vec{x})$, with free variables~$\vec{x}=(x_1, \ldots, x_k)$,
and a tuple of constants $\vec{a}=(a_1, \ldots, a_k)$,
$q(\vec{a})$ denotes the BCQ obtained by replacing each variable in~$\vec{x}$
by the corresponding constant in~$\vec{a}$. An \emph{answer} to $q(\vec{x})$ 
over a database  
 $\Dmc$ is a tuple of constants~$\vec{a}$ from $\domain{\Dmc}$ 
 such that $\Dmc \models q(\vec{a})$. 

\subsubsection*{Constraints} A \emph{universal constraint} over a schema $\preds$ is a 
FOL sentence of the form 
$\forall\vec{x}(R_1(\vec{t_1})\wedge\dots\wedge R_n(\vec{t_n})\wedge\neg P_1(\vec{u_1})\wedge\dots\wedge\neg P_m(\vec{u_m})\wedge \varepsilon\rightarrow \bot)$,
where each $R_i(\vec{t_i})$ (resp.\ $P_i(\vec{u_i})$) is a relational atom over $\preds$, 
$\varepsilon$ is a (possibly empty) conjunction of inequality atoms, 
and $\vec{u_1}\cup\dots\cup\vec{u_m}\subseteq \vec{t_1}\cup\dots\cup\vec{t_n}$ (safety condition). 
Universal constraints can also be written  in the form 
$\forall\vec{x}(R_1(\vec{t_1})\wedge\dots\wedge R_n(\vec{t_n})\wedge \varepsilon\rightarrow P_1(\vec{u_1})\vee\dots\vee P_m(\vec{u_m}))$. For simplicity, we shall often omit the universal quantification 
and will sometimes use the generic term \emph{constraint} to mean universal constraint. 

\emph{Denial constraints} are universal constraints of the form $\forall\vec{x}(R_1(\vec{t_1})\wedge\dots\wedge R_n(\vec{t_n})\wedge \varepsilon\rightarrow \bot)$, which capture the well-known class of functional dependencies. 

We say that a database $\Dmc$ is \emph{consistent} \wrt a set of constraints $\constraints$ 
if $\Dmc\models \constraints$. 
Otherwise, 
$\Dmc$ is \emph{inconsistent} (\wrt\ $\constraints$).

A constraint is \emph{ground} if it contains no variables. 
Given a 
constraint $\tau$ and database $\Dmc$, we use
$gr_\Dmc(\tau)$ for the set of
all ground constraints obtained by (i) replacing variables with constants from $\domain{\Dmc}$, (ii) removing all true $c \neq d$ atoms, and (iii)
removing all constraints that contain an atom $c \neq c$.  We let $gr_\Dmc(\constraints):=\bigcup_{\tau\in\constraints}gr_\Dmc(\tau)$, and 
note that $\Dmc\models \tau$ iff $\Dmc\models \tau_g$ for every $\tau_g\in gr_\Dmc(\tau)$. 

\subsubsection*{Repairs} 
A \emph{symmetric difference repair}, or $\Delta$-repair, of $\Dmc$ \wrt $\constraints$ is a database $\Rmc$ such that (i) $\Rmc\models\constraints$ and (ii) there is no $\Rmc'$ such that $\Rmc'\models\constraints$ and $\Rmc'\Delta\Dmc\subsetneq \Rmc\Delta\Dmc$, where $\Delta$ is the symmetric difference operator: $S_1\Delta S_2=(S_1\setminus S_2)\cup (S_2\setminus S_1)$. 
If only fact deletions are permitted, we obtain \emph{subset repairs} ($\subseteq$-repairs), 
and if only fact additions are permitted, \emph{superset repairs} ($\supseteq$-repairs). 
We denote the set of $\Delta$-repairs 
of $\Dmc$ \wrt $\constraints$ by $\deltareps{\Dmc,\constraints}$. 

Because of the safety condition, an empty database satisfies any set of universal constraints, so every database has at least one $\subseteq$-repair (which is also a $\Delta$-repair), while  
it may be the case that no $\supseteq$-repair exists. 
Moreover, for the subclass of denial constraints, $\Delta$-repairs and $\subseteq$-repairs coincide since adding facts cannot resolve a violation of a denial constraint. 

\section{Optimal Repairs for Universal Constraints}\label{sec:opti-repairs}
In this section, we show how existing notions of optimal repairs, defined for $\subseteq$-repairs w.r.t.\ denial constraints,
can be 
 lifted to the broader setting of $\Delta$-repairs w.r.t.\ universal constraints. We then use the resulting repair notions
to define inconsistency-tolerant semantics for query answering. 

\subsection{Conflicts for Universal Constraints}\label{subsec:conflicts}
In the setting of denial constraints, a conflict is a minimal subset of the database 
that is inconsistent w.r.t.\ the constraints. 
Conflicts and the associated notion of conflict (hyper)graph
underpin many results and algorithms for consistent query answering, and in particular, 
they appear in the definition of prioritized databases \cite{DBLP:journals/amai/StaworkoCM12}. 
Our first task will thus be to 
define a suitable
notion of conflict for universal constraints. 

An important observation is that the absence of a fact may contribute
to the violation of a universal constraint. For this reason, 
conflicts will contain
both facts and negated facts, where $\no P(\vec{c})$ indicates 
that $P(\vec{c})$ is absent. 
We use $\factset$ for the set of facts over $\preds$ with constants from $\domain{\Dmc}$,
and let $\litset=\Dmc\cup\{\no\alpha\mid \alpha\in \factset\setminus\Dmc\}$ be the set of \emph{literals} of $\Dmc$. 
Conflicts can then be defined as minimal sets of literals that necessarily lead to a constraint violation.

\begin{definition}
Given a database $\Dmc$ and set of (universal) constraints $\constraints$,
the set
$\conflicts{\Dmc,\constraints}$ of \emph{conflicts of $\Dmc$ w.r.t.\ $\constraints$} 
contains all $\subseteq$-minimal sets $\conf\subseteq\litset$ such that for every database $\Imc$, 
if $\Imc\models \conf$, then $\Imc\not\models\constraints$.

\end{definition}

\begin{example}\label{ex:conflicts}
Let $\Dmc= \{A(a), B(a)\}$ and $\constraints=\{\tau_1,\tau_2, \tau_3\}$, where 
$\tau_1:= A(x)\rightarrow C(x)$, $\tau_2:=B(x)\rightarrow D(x)$, and $\tau_3:=C(x)\wedge D(x)\rightarrow \bot$.  
It can be verified that 
\begin{align*}
\deltareps{\Dmc,\constraints}=  \{&\emptyset, \{A(a), C(a)\} , \, \{B(a),D(a)\} \}
\end{align*}
and that the set $\conflicts{\Dmc,\constraints}$ is as follows:
$$ \{\{A(a), \no C(a)\}, \{B(a), \no D(a)\},   \{A(a), B(a)\}\} $$
The first (resp.\ second) conflict directly violates $\tau_1$ (resp.\ $\tau_2$). 
To see why $\{A(a), B(a)\}$ is also a conflict, 
consider any database $\Imc$ such that $\{A(a), B(a)\}\subseteq\Imc$. 
Then either $C(a)\notin\Imc$ or $D(a)\notin\Imc$, in which case $\Imc$ violates $\tau_1$ or $\tau_2$, or $\Imc$ contains both $C(a)$ and $D(a)$, 
hence 
violates $\tau_3$. 
\end{example}

We also provide two alternative characterizations of conflicts,
in terms of the hitting sets of literals removed from $\Delta$-repairs and 
the prime implicants\footnote{We recall that a prime implicant of a propositional formula $\psi$
is a minimal conjunction of propositional literals $\kappa$ that entails $\psi$.  } 
of the propositional formula 
stating that there is a constraint violation (treating the elements of $\litset$ as propositional literals):
\begin{restatable}{proposition}{DefConflicts}\label{prop:defconflicts} 
For every database $\Dmc$ and constraint set $\constraints$: 
\begin{enumerate}
\item $\conflicts{\Dmc,\constraints}=
\{ \Hmc\cap\Dmc\cup\{\no\alpha \mid \alpha\in\Hmc\setminus\Dmc\} \mid \Hmc\in \mhs{\Dmc,\constraints} \}$ where $\mhs{\Dmc,\constraints}$ is the set of all minimal hitting sets of $\{\Rmc\Delta\Dmc\mid \Rmc \in\deltareps{\Dmc,\constraints}\}$. 
\item $\conflicts{\Dmc,\constraints}=
\{ \{\lambda_1, \ldots, \lambda_k\} \subseteq \litset \mid  
\lambda_1 \wedge \ldots \wedge \lambda_k \text{ is a prime implicant of } \bigvee_{\varphi\rightarrow \bot\in gr_\Dmc(\constraints)} \varphi \}$. 
\end{enumerate}
\end{restatable}

We can show that our notion of conflicts enjoy similar properties to conflicts w.r.t.\ denial constraints, 
but to formulate them, we must 
first introduce some useful terminology and notation for moving between databases and sets 
of literals.

Given a database $\Dmc$ over schema $\preds$, 
a \emph{candidate repair} for $\Dmc$ is a database $\Bmc$ with 
$\Bmc\subseteq \factset$. 
For every candidate repair $\Bmc$ for $\Dmc$, we define its corresponding set of literals $\mi{Lits}_\Dmc(\Bmc)=\Bmc\cup\{\no\alpha\mid \alpha\in \factset\setminus\Bmc\}$ and the set of literals $\comp{\Dmc}{\Bmc}=\mi{Lits}_\Dmc(\Bmc)\cap \litset=(\Bmc\cap\Dmc)\cup\{\no\alpha\mid \alpha\in\factset\setminus(\Bmc\cup\Dmc)\}$ upon which $\Bmc$ and $\Dmc$ agree. 
Furthermore, with every subset $\Bmc \subseteq\litset$ we can associate a candidate repair $\restr{\Dmc}{\Bmc}=\Bmc\cap\Dmc\cup\{\alpha\mid\no\alpha\in \litset\setminus\Bmc\}$. 
Note that if $\Bmc$ is a candidate repair, $\restr{\Dmc}{\comp{\Dmc}{\Bmc}}=\Bmc$.

\begin{restatable}{proposition}{CharacterizationsRepairsComp}\label{prop:characterizations-repairs-comp}
Let $\Dmc$ be a database, $\constraints$ a set of universal constraints, and $\Rmc$ a candidate repair for $\Dmc$. 
\begin{enumerate}
\item $\Rmc\in\deltareps{\Dmc,\constraints}$ iff $\comp{\Dmc}{\Rmc}$ is a maximal subset of $\litset$ such that $\restr{\Dmc}{\comp{\Dmc}{\Rmc}}\models\constraints$, \ie $\Rmc\models\constraints$.
\item $\Rmc\in\deltareps{\Dmc,\constraints}$ iff $\comp{\Dmc}{\Rmc}$ is a maximal subset of $\litset$ such that 
$\conf \not \subseteq\comp{\Dmc}{\Rmc}$ for every $\conf\in\conflicts{\Dmc,\constraints}$. 
\item  $\Rmc\in\deltareps{\Dmc,\constraints}$ iff $\comp{\Dmc}{\Rmc}$ is a maximal independent set (MIS) of the \emph{conflict hypergraph} $\confgraph{\Dmc}{\constraints}$, 
whose vertices are the literals from $\litset$ and whose hyperedges are the conflicts of $\Dmc$ \wrt $\constraints$. 
\end{enumerate}
\end{restatable}

The first property states that $\Delta$-repairs correspond to the consistent databases 
that preserve a maximal set of the original literals, while the second rephrases consistency in terms of conflicts. 
The third generalizes a well-known hypergraph-based characterization of $\subseteq$-repairs. As the next remark explains,
an earlier attempt at defining conflicts for universal constraints failed to obtain such a property.

\begin{remark}
\citeauthor{DBLP:journals/is/StaworkoC10} \shortcite{DBLP:journals/is/StaworkoC10} define a conflict as a set of literals obtained by grounding a universal constraint, and the hyperedges of their extended conflict hypergraph $ECG(\Dmc,\constraints)$  are either conflicts or `relevant' pairs of literals $\{\alpha,\neg\alpha\}$. 
For instance, if we take 
$\Dmc$ and 
$\Cmc$ as in Example 1, 
then $ECG(\Dmc,\constraints)$ 
has hyperedges 
$ \{A(a), \no C(a)\}, \{B(a), \no D(a)\}, \{C(a), D(a)\}$, $\{C(a),\neg C(a)\}$ and $\{D(a), \neg D(a)\}$. 

Every repair gives rise to a MIS of $ECG(\Dmc,\constraints)$, but a MIS need not correspond to any repair. 
Proposition 4 in \cite{DBLP:journals/is/StaworkoC10} claims a weaker converse: for every MIS $M$ of $ECG(\Dmc,\constraints)$, either its positive projection $M^+=M\cap\factset$ is a $\Delta$-repair of $\Dmc$ \wrt $\constraints$, or there exists 
a MIS 
$N$ of $ECG(\Dmc,\constraints)$ such that 
$N^+\Delta\Dmc\subsetneq M^+\Delta\Dmc$. 
However, our example disproves this claim, 
as $M=\{A(a), B(a), C(a)\}$ is a MIS of $ECG(\Dmc,\constraints)$, but $M^+=M$ is not a $\Delta$-repair (it violates $\tau_2$), and there is no MIS $N$
with $N^+\Delta\Dmc\subsetneq M^+\Delta\Dmc$.  
Essentially, the problem is that 
their notion of conflicts does not take into account implicit constraints ($A(x) \wedge B(x) \rightarrow \bot$ in this example).  
\end{remark}

To clarify the relationship between the universal and denial constraint settings,
we translate the former into the latter. 
Take a database $\Dmc$ and set of universal constraints $\Cmc$ over schema $\preds$.
To represent negative literals, we introduce an extended schema $\preds'=\preds \cup \{\widetilde{P} \mid P \in \preds\}$
and a function $\tofacts$ that maps sets of literals over $\preds$ into sets of facts over $\preds'$ 
by replacing each negative literal $\neg P(\vec{c})$ by $\widetilde{P}(\vec{c})$.
We then consider the database $\Dmc_d=\tofacts(\litset)=\Dmc\cup\{\widetilde{P}(\vec{c})\mid P(\vec{c})\in\factset\setminus\Dmc\}$,
and the set of ground denial constraints 
$\constraints_{d,\Dmc}=\{ (\bigwedge_{\alpha \in \tofacts(\conf)} \alpha) \rightarrow \bot \mid \conf\in\conflicts{\Dmc,\constraints}\}$. 

\begin{restatable}{proposition}{ReductionUCDenials}\label{prop:reduction-UC-denials}For every database $\Dmc$ and constraint set $\constraints$: 
$\conflicts{\Dmc_d,\constraints_{d,\Dmc}}=\{ \tofacts(\conf)\mid \conf\in\conflicts{\Dmc,\constraints}\}$ and $\deltareps{\Dmc_d,\constraints_{d,\Dmc}}=\{\tofacts(\comp{\Dmc}{\Rmc})\mid \Rmc\in\deltareps{\Dmc,\constraints}\}$.
\end{restatable}

One may naturally wonder whether 
a set of denial constraints $\constraints_d$ which does not depend on $\Dmc$ could be used in place of $\constraints_{d,\Dmc}$ in Proposition \ref{prop:reduction-UC-denials}. The answer is no: the existence of such a set 
$\constraints_d$ would imply 
a 
data-independent bound on the size of conflicts that may appear in any set $\conflicts{\Dmc_d,\constraints_{d}}$, and hence in 
$\conflicts{\Dmc,\constraints}$. However, as the next example illustrates, 
universal constraints differ from denial constraints in that 
the size of the conflicts cannot be bounded independently from the database. 
\begin{example}
Let $\constraints$ consist of $R(x,y)\wedge A(x)\rightarrow A(y)$ and $A(x)\wedge B(x) \rightarrow \bot$. Then for every $n\geq1$, we can build a database
$\{A(a_0), R(a_0,a_1),\dots,R(a_{n-1}, a_n), B(a_n)\}$ of size $n+2$ which is a conflict (of itself) \wrt $\constraints$. 
\end{example}

\subsection{Prioritized Databases \& Optimal Repairs}\label{subsec:optimalrepairs}
With the definition of conflicts in place, we can extend the notion of prioritized database \cite{DBLP:journals/amai/StaworkoCM12} to the setting 
of universal constraints. 

\begin{definition}
A \emph{priority relation} $\succ$ for a database $\Dmc$ \wrt a set of universal constraints $\constraints$ is an acyclic binary relation over the literals of $\conflicts{\Dmc, \constraints}$ such that if $\lambda\succ\mu$, then there exists $\conf\in\conflicts{\Dmc, \constraints}$ such that $\{\lambda,\mu\}\subseteq \conf$. 
We say that~$\succ$ is \emph{total} if for every pair $\lambda\neq\mu$ such that $\{\lambda,\mu\}\subseteq \conf$ for some $\conf\in\conflicts{\Dmc, \constraints}$, either $\lambda\succ\mu$ or $\mu\succ\lambda$. 
A \emph{completion} of $\succ$ is a total priority relation $\succ'\ \supseteq \ \,\succ$.

A priority relation $\succ$ is \emph{score-structured} if there is 
a scoring function $s:\bigcup_{\conf\in\conflicts{\Dmc, \constraints}} \conf\rightarrow \mathbb{N}$ such that 
for every 
$\{\lambda,\mu\}\subseteq \conf$ with $\conf\in\conflicts{\Dmc, \constraints}$, 
$\lambda\succ\mu$ iff $s(\lambda)>s(\mu)$. 
\end{definition}

\begin{definition}
A \emph{prioritized database} $\Dmc^\constraints_\succ=(\Dmc,\constraints,\succ)$ consists of a database $\Dmc$, a set of universal constraints $\constraints$, and a priority relation $\succ$ for $\Dmc$ \wrt $\constraints$.  
\end{definition}

We now extend the definitions of optimal repairs to the case of universal constraints.

\begin{definition} Consider a prioritized database $\Dmc^\constraints_\succ=(\Dmc,\constraints,\succ)$, and let $\Rmc \in \deltareps{\Dmc,\constraints}$.
\begin{itemize}
\item A \emph{Pareto improvement} of $\Rmc$ 
 is a database $\Bmc$ consistent \wrt $\constraints$ such that 
 there is $\mu\in \comp{\Dmc}{\Bmc}\setminus \comp{\Dmc}{\Rmc}$ with $\mu\succ\lambda$ for every $\lambda\in \comp{\Dmc}{\Rmc}\setminus \comp{\Dmc}{\Bmc}$. 
 \item A \emph{global improvement} of $\Rmc$ 
 is a database $\Bmc$ consistent \wrt $\constraints$ such that $\comp{\Dmc}{\Bmc}\neq \comp{\Dmc}{\Rmc}$ and for every $\lambda\in \comp{\Dmc}{\Rmc}\setminus \comp{\Dmc}{\Bmc}$, there exists $\mu\in \comp{\Dmc}{\Bmc}\setminus \comp{\Dmc}{\Rmc}$ such that $\mu\succ\lambda$.
\end{itemize}
We say that 
$\Rmc$ is:
\begin{itemize}
\item \emph{Pareto-optimal} if there is no Pareto improvement of $\Rmc$. 
\item \emph{globally-optimal} if there is no global improvement of $\Rmc$. 
\item \emph{completion-optimal}  
if $\Rmc$ is a globally-optimal  $\Delta$-repair of $\Dmc^\constraints_{\succ'}$, for some completion $\succ'$ of $\succ$.
\end{itemize}
We denote by $\deltagreps{\Dmc^\constraints_\succ}$, 
$\deltapreps{\Dmc^\constraints_\succ}$ and $\deltacreps{\Dmc^\constraints_\succ}$ the sets of globally-, 
Pareto- and completion-optimal $\Delta$-repairs.
\end{definition}

A Pareto improvement is also a global improvement, so $\deltagreps{\Dmc^\constraints_\succ}\subseteq \deltapreps{\Dmc^\constraints_\succ}$, and a global improvement \wrt $\succ$ is a global improvement \wrt any completion $\succ'$ of $\succ$, 
so $\deltacreps{\Dmc^\constraints_\succ}\subseteq \deltagreps{\Dmc^\constraints_\succ}$. 
Hence, as in the denial constraints case, $\deltacreps{\Dmc^\constraints_\succ}\subseteq \deltagreps{\Dmc^\constraints_\succ}\subseteq \deltapreps{\Dmc^\constraints_\succ}$. 
Moreover, there always exists at least one completion-(hence Pareto- and globally-)optimal $\Delta$-repair, which can be obtained from $\confgraph{\Dmc}{\constraints}$ by the following greedy procedure: while some literal from $\litset$ has not been considered, pick a literal 
that is maximal \wrt $\succ$ among those not yet considered, and add it to the current set if it does not introduce a conflict from $\conflicts{\Dmc,\constraints}$. If $\Bmc$ is a subset of $\litset$ obtained by this procedure, we show that $\restr{\Dmc}{\Bmc}\in\deltacreps{\Dmc^\constraints_\succ}$. 
This procedure requires us to compute $\conflicts{\Dmc,\constraints}$, hence does not run in polynomial time (unlike the denial constraint case).
However, as for denial constraints, we have: 
\begin{restatable}{proposition}{Categoricity}\label{prop:categoricity}
If $\succ$ is total, then $| \deltapreps{\Dmc^\constraints_\succ}|=1$. 
\end{restatable}
\noindent In particular, this means $\deltagreps{\Dmc^\constraints_\succ}= \deltapreps{\Dmc^\constraints_\succ}$ when $\succ$ is total,
so we may replace globally-optimal by Pareto-optimal in the definition of completion-optimal $\Delta$-repairs.

\begin{example}\label{ex:optimal-reps}
Let $\Dmc=\{S(a,b), S(a,c), R(d,b), R(d,c)\}$, where 
$R(d,b)\succ S(a,b)$, $S(a,b)\succ\no A(a)$, $S(a,c)\succ R(d,c)$, $S(a,c)\succ \no B(a)$,
and $\Cmc$ contains the constraints:
\begin{align*}
&S(x,y)\wedge S(x,z)\wedge y\neq z\rightarrow \bot && S(x,y)\rightarrow A(x)  \\
& R(x,y)\wedge R(x,z)\wedge y\neq z\rightarrow \bot & &S(x,y)\rightarrow B(x)  \\
&R(y,x)\wedge S(z,x)\rightarrow\bot&&
\end{align*}
The conflicts are all binary, so the conflict hypergraph is a graph, pictured below.
We use an arrow $\lambda \rightarrow \mu$ when $\lambda\succ\mu$ and dotted lines for conflicting literals with no priority. 
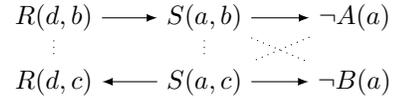
\begin{center}
\begin{tikzpicture} 
\node (1) [above] at (0,0.25) {$R(d,b)$};
\node (2) [above] at (2,0.25) {$S(a,b)$};
\node (3) [above] at (4,0.25) {$\no A(a)$};
\node (4) [below] at (0,0) {$R(d,c)$};
\node (5) [below] at (2,0) {$S(a,c)$};
\node (6)  [below] at (4,0) {$\no B(a)$};
\draw[->] (1) -- (2);
\draw[->] (2) -- (3);
\draw[->] (5) -- (4);
\draw[->] (5) -- (6);
\draw[dotted] (1) -- (4);
\draw[dotted] (2) -- (5);
\draw[dotted] (2) -- (6);
\draw[dotted] (3) -- (5);
\end{tikzpicture}
\end{center}
It can be verified that the optimal repairs are as follows:
\begin{align*}
\deltacreps{\Dmc^\constraints_\succ}=&\{\{R(d,b), S(a,c), A(a), B(a)\}\}
\\
 \deltagreps{\Dmc^\constraints_\succ}=&\deltacreps{\Dmc^\constraints_\succ}\cup\{\{R(d,b)\},\{R(d,c)\}\}
\\
\deltapreps{\Dmc^\constraints_\succ}=&\deltagreps{\Dmc^\constraints_\succ}\cup\\
& \{\{R(d,c), S(a,b), A(a), B(a)\}\}
\end{align*}
and that $\deltareps{\Dmc,\constraints} = \deltapreps{\Dmc^\constraints_\succ}$. 
\end{example}

When $\succ$ is score-structured with scoring function $s$, we define the \emph{prioritization} of $\bigcup_{\conf\in\conflicts{\Dmc,\Cmc}}\conf$ as the partition $\Smc_1,\dots,\Smc_n$ such that for every $1\leq i\leq n$, there exists $m\in \mathbb{N}$ such that $\Smc_i=\{\lambda\mid s(\lambda)=m\}$, and for every $\{\lambda_i,\lambda_j\}\subseteq\conf\in\conflicts{\Dmc,\constraints}$, $\lambda_i\succ \lambda_j$ iff $\lambda_i\in\Smc_i$, $\lambda_j\in\Smc_j$ and $ i< j$. 
Intuitively, 
the more reliable a literal $\lambda$ the smaller the index of $\Smc_i$ that contains $\lambda$. 
 \citeauthor{DBLP:conf/aaai/BienvenuBG14} \shortcite{DBLP:conf/aaai/BienvenuBG14} introduced a notion of 
$\subseteq_P$-repair based upon such prioritizations, which we adapt below to $\Delta$-repairs. 

\begin{definition}
Let $\Dmc^\constraints_\succ$ be a prioritized database such that $\succ$ is score-structured and $\Smc_1,\dots,\Smc_n$ is the prioritization of $\bigcup_{\conf\in\conflicts{\Dmc,\Cmc}}\conf$. 
A \emph{$\Delta_P$-repair} of $\Dmc^\constraints_\succ$ is a candidate repair $\Rmc$ such that (i) $\Rmc\models\constraints$ and (ii) there is no $\Rmc'\models\constraints$ such that 
 there is some $1\leq i\leq n$ such that 
 \begin{itemize}
 \item $\comp{\Dmc}{\Rmc}\cap\Smc_i\subsetneq \comp{\Dmc}{\Rmc'}\cap\Smc_i $ and 
 \item for all $1\leq j<i$, $\comp{\Dmc}{\Rmc}\cap\Smc_j= \comp{\Dmc}{\Rmc'}\cap\Smc_j $. 
\end{itemize}
We denote by $\deltascorereps{\Dmc^\constraints_\succ}$ the set of $\Delta_P$-repairs of $\Dmc^\constraints_\succ$. 
\end{definition}

As in the case of denial constraints, all four notions of optimal $\Delta$-repairs coincide when $\succ$ is score-structured.
\begin{restatable}{proposition}{ScoreStructuredCollapse}\label{prop:score-structured-collapse}
If $\succ$ is score-structured, then 
$\deltacreps{\Dmc^\constraints_\succ} = \deltagreps{\Dmc^\constraints_\succ} = \deltapreps{\Dmc^\constraints_\succ}=\deltascorereps{\Dmc^\constraints_\succ}$.
\end{restatable}

We can now define variants of existing inconsistency-tolerant semantics based upon our optimal repairs. 
\begin{definition}
Fix X $\in \{S,P,G,C\}$ and 
consider a prioritized database $\Dmc^\constraints_\succ$, query $q(\vec{x})$, and tuple of constants $\ans$ 
with $|\vec{x}|=|\ans|$. 
Then $\ans$ is an answer to $q(\vec{x})$ over $\Dmc^\constraints_\succ$ 
\begin{itemize}
\item under \emph{X-brave semantics}, denoted $\Dmc^\constraints_\succ\bravemodels{X} q(\ans)$, if $\Rmc\models q(\ans)$ for some $\Rmc \in \deltaxreps{\Dmc^\constraints_\succ}$;
\item under \emph{X-CQA semantics}, denoted $\Dmc^\constraints_\succ \armodels{X} q(\ans)$, if $\Rmc\models q(\ans)$ for every $\Rmc \in \deltaxreps{\Dmc^\constraints_\succ}$;
\item under \emph{X-intersection semantics}, denoted $\Dmc^\constraints_\succ \iarmodels{X} q(\ans)$, if $\Bmc \models q(\ans)$ where $\Bmc=\bigcap_{\Rmc \in \deltaxreps{\Dmc^\constraints_\succ} }\Rmc$.
\end{itemize}
\end{definition}

Just as in the case of denial constraints, these semantics are related as follows:
$$\Dmc^\constraints_\succ \iarmodels{X} q(\ans) \Rightarrow \Dmc^\constraints_\succ \armodels{X} q(\ans)  \Rightarrow \Dmc^\constraints_\succ \bravemodels{X} q(\ans)$$ 
Unlike the denial constraint case, 
the intersection of the optimal $\Delta$-repairs may be inconsistent \wrt $\constraints$. For example, if $\Dmc=\{A(a)\}$, $\constraints=\{A(x)\rightarrow B(x) \vee C(x)\}$, $A(a)\succ\neg B(a)$ and $A(a)\succ \neg C(a)$, then $\bigcap_{\Rmc \in \deltapreps{\Dmc^\constraints_\succ}}\Rmc=\{A(a)\}$ violates the constraint. This is not a problem since we consider conjunctive queries, which are monotone, meaning that if the intersection of the optimal $\Delta$-repairs yields a query answer, then the tuple is an answer in every optimal $\Delta$-repair. 

\begin{example}[Example \ref{ex:optimal-reps} cont'd]
Considering the different semantics based upon Pareto-optimal repairs:
\begin{itemize}
\item $\Dmc^\constraints_\succ\bravemodels{P} A(a)$ but $\Dmc^\constraints_\succ\not\armodels{P} A(a)$;
\item $\Dmc^\constraints_\succ\armodels{P} \exists y R(d,y)$ but $\Dmc^\constraints_\succ\not\iarmodels{P} \exists y R(d,y)$.
\end{itemize}
If we consider now 
semantics for the different kinds of optimal repairs, we find that, \eg:
\begin{itemize}
\item $\Dmc^\constraints_\succ\armodels{C} A(a)$ but $\Dmc^\constraints_\succ\not\armodels{G} A(a)$;
\item 
$\Dmc^\constraints_\succ\bravemodels{P} S(a,b)$ but $\Dmc^\constraints_\succ\not\bravemodels{G} S(a,b)$.
\end{itemize}
\end{example}

\section{Complexity Analysis}\label{sec:complexity}
In this section, we analyze the data complexity of the central computational tasks related to optimal repairs. 
We consider the following decision problems: 
\begin{itemize}
\item X-repair checking:  
given a prioritized database $\Dmc^\constraints_\succ$ 
and a candidate repair $\Rmc$, decide whether $\Rmc \in \deltaxreps{\Dmc^\constraints_\succ}$;
\item Query answering under X-Sem semantics: 
given a prioritized database $\Dmc^\constraints_\succ$,
a query $q$, and a candidate answer $\vec{a}$, decide whether 
 $\Dmc^\constraints_\succ\models^\text{X}_\text{Sem} q(\ans)$;
\end{itemize}
where X $\in \{S,P,G, C\}$ and $\text{Sem}\in\{$brave, CQA, $\cap\}$. 
We focus on data complexity, which is measured in terms of the size of the database 
$\Dmc$,
treating the 
constraints $\Cmc$ and query $q$ as fixed and of constant size (
under the latter assumption, $\Rmc$ and $\ans$ are of polynomial size w.r.t.\ $\Dmc$).
Table \ref{tab:complexity} summarizes our new results for optimal repairs w.r.t.\ universal constraints
alongside existing results for denial constraints. 

\citeauthor{DBLP:journals/is/StaworkoC10} \shortcite{DBLP:journals/is/StaworkoC10} showed that $S$-repair checking is \conp-complete in data complexity. 
We show that the same holds for Pareto- and globally-optimal repairs: 

\begin{restatable}{theorem}{ThComplexityPGrepairchecking}
X-repair checking is \conp-complete in data complexity for X $\in \{P,G\}$.  
\end{restatable}
\begin{proof}[Proof Sketch]
The lower bound is inherited from $\Delta$-repairs. For the upper bounds,  
we sketch \np\ procedures
for checking whether $\Rmc \not \in \deltaxreps{\Dmc^\constraints_\succ}$ for a given candidate repair~$\Rmc$. 
In a nutshell, we guess either  (i) `inconsistent', (ii) `not maximal' together with another candidate repair $\Rmc'$,
or (iii) `improvement' together with a candidate (Pareto or global) improvement $\Bmc$. 
In case (i), it suffices to verify in \ptime\ that $\Rmc \not \models \Cmc$, returning yes if so. 
In case (ii), we test in \ptime\ whether $\Rmc'\Delta\Dmc\subsetneq \Rmc\Delta\Dmc$ and $\Rmc' \models \Cmc$, returning yes if both conditions hold. 
In case (iii), we check in \ptime whether $\Bmc$ is indeed a (Pareto / global) improvement of $\Rmc$, returning yes if so. 
\end{proof}

Interestingly, we observe that P-repair checking is hard 
even if we already know the input is a $\Delta$-repair:

\begin{restatable}{lemma}{ParetoHard}\label{pareto-hard}
Deciding whether a given $\Delta$-repair is Pareto-optimal is 
\conp-complete in data complexity. 
\end{restatable}

We next turn to C-repair checking. A first idea would be to guess a completion $\succ'$
and check (using an \np\ oracle) that the input database is Pareto-optimal w.r.t.~$\succ'$. However, determining whether the guessed binary relation 
is a completion is not straightforward, as we must make sure that 
we relate all and only those literals that appear together in some conflict. As the following result shows, 
even identifying conflicts is a challenging task for universal contraints:

\begin{table}[t]
\setlength{\tabcolsep}{3pt}
\begin{tabular}{l@{\quad}lcccc}
&& $S$ & $P$ & $C$ & $G$\\
\midrule
\multirow{3}{*}{\rotatebox[origin=c]{90}{Univ.}}& \sc{RC} &\conp   
& \conp & \conp-h, in \sigmaptwo &\conp   \\
&\sc{Brave} &\sigmaptwo  
&\sigmaptwo&\sigmaptwo & \sigmaptwo\\
&CQA, \sc{Int} & \piptwo &\piptwo&\piptwo
& \piptwo \\ \midrule 
\multirow{3}{*}{\rotatebox[origin=c]{90}{Denial}}& \sc{RC} &in \ptime  &in \ptime  &in \ptime  & \conp \\
&\sc{Brave} &in \ptime  
&\np  & \np & \sigmaptwo\\
&\sc{CQA}/\sc{Int}&  \conp/in \ptime & \conp 
&\conp& \piptwo\\
\bottomrule
\end{tabular}
\caption{Data complexity of X-repair checking (\textsc{RC}) and query answering under X-brave (\textsc{Brave}), X-CQA, 
and X-intersection (\textsc{Int}) semantics (X $\in \{S,P,G, C\}$) w.r.t.\ universal or denial constraints. Completeness results except where indicated otherwise. 
} 
\label{tab:complexity}
\end{table}

\begin{restatable}{lemma}{LemConflictChecking}\label{lem:conflict-checking}
Deciding whether a set of literals 
belongs to $\conflicts{\Dmc,\constraints}$ is $BH_2$-complete \wrt data complexity.
\end{restatable}

With a more careful approach, we can 
show that C-repair checking does belong to $\Sigma^p_2$. 
The exact complexity is open. 

\begin{restatable}{theorem}{ThCompletionRepairChecking}\label{th:completion-repair-checking}
C-repair checking is \conp-hard and in $\Sigma^p_2$ w.r.t.\ data complexity. 
\end{restatable}
\begin{proof}[Proof Sketch]
We use a non-deterministic version of the greedy procedure sketched in Section \ref{subsec:optimalrepairs}: to decide if 
$\Rmc \in \deltacreps{\Dmc^\constraints_\succ}$, we guess the order in which literals of $\litset$ will be considered, 
and for each $\lambda\in\litset\setminus\comp{\Dmc}{\Rmc}$, we guess a set of literals $L \subseteq \comp{\Dmc}{\Rmc}$ that
precede $\lambda$ in the order and such that $L \cup \{\lambda\}$ forms a conflict. 
\end{proof}

Leveraging our results for repair checking, 
we can establish the precise data complexity of query answering for all combinations of semantics and optimality notions: 

\begin{restatable}{theorem}{ThComplexityQueryAnswering}\label{th:complexityQueryAnswering}
Query answering under X-brave (resp.\ X-CQA and X-intersection) semantics is 
\sigmaptwo-complete (resp.\ \piptwo-complete) in data complexity, for X $\in \{P,G, C\}$.  
\end{restatable}

The lower bounds that are higher for universal constraints than denial constraints 
involve databases whose conflicts are difficult to compute. 
This is no coincidence, 
as we show that if the set of conflicts are available, the complexity drops: 

\begin{restatable}{theorem}{ComplexityWithConflictsGiven}\label{thm:conflictsgiven}
If $\conflicts{\Dmc,\constraints}$ is 
given and considered as part of the input, then all complexity results for denial constraints listed in Table \ref{tab:complexity}
hold also for universal constraints. 
\end{restatable}

The lower complexities apply in particular to sets of constraints whose conflicts have bounded size, such as universal constraints with
at most two relational atoms. Unfortunately, we show that it is impossible in general to determine whether a given set of constraints has bounded conflicts:

\begin{restatable}{theorem}{BoundedConflictsUndec}\label{th:boundedconflictsUndec}
Given a set of universal constraints $\constraints$, it is undecidable to determine whether there exists $k \in \mathbb{N}$ such that for every database $\Dmc$, $\max_{\conf\in\conflicts{\Dmc,\constraints}}(|\conf|)\leq k$. 
\end{restatable}

\section{Links with Active Integrity Constraints}\label{sec:aics}

Active integrity constraints define which update operations are allowed to solve a constraint violation  \cite{DBLP:conf/ppdp/FlescaGZ04,DBLP:conf/iclp/CaropreseGSZ06,DBLP:journals/tkde/CaropreseGZ09}, 
in the same spririt that prioritized databases express preferred ways of solving conflicts. 
This section investigates how these two frameworks relate.

\subsection{Preliminaries on Active Integrity Constraints}\label{subsec:aics}
We briefly recall the basics of 
active integrity constraints, directing 
readers to \cite{DBLP:journals/ai/BogaertsC18} for a good overview of the area.

\subsubsection*{Update actions} 
An \emph{update atom} is of the form $+P(\vec{x})$ or $-P(\vec{x})$ where $P(\vec{x})$ is a relational atom. 
We use $\fix$ to map relational literals to the corresponding update atoms: $\fix(P(\vec{x}))= -P(\vec{x})$
and $\fix(\neg P(\vec{x}))= +P(\vec{x})$. 
An \emph{update action} is a ground update atom, \ie is of the form $-\alpha$ or $+\alpha$ with $\alpha$ a fact. 
A set of update actions $\Umc$ is \emph{consistent} if $\Umc$ does not contain both $-\alpha$ and $+\alpha$ for some fact $\alpha$. 
The result of applying a consistent set of update actions $\Umc$ on a database  $\Dmc$ is $\Dmc\circ\Umc:=\Dmc\setminus\{\alpha\mid -\alpha\in \Umc\}\cup\{\alpha\mid+\alpha\in\Umc\}$.

\subsubsection*{Active integrity constraints} 
An \emph{active integrity constraint} (AIC) takes the form $r=\ell_1\wedge\dots\wedge \ell_n\rightarrow \{A_1, \dots, A_k\}$, where
$\aicbody(r)=\ell_1\wedge\dots\wedge \ell_n$ is such that $\tau_r := \aicbody(r) \rightarrow \bot$ is a universal constraint,
 $\aicup(r)=\{A_1, \dots, A_k\}$ is non-empty, and 
every $A_j$ is equal to $\fix(\ell_i)$ for some 
$\ell_i$. 
We use 
$\aiclits(r)$ 
for the set of literals appearing in $\aicbody(r)$,
and say that $\ell \in \aiclits(r)$ 
is  \emph{non-updatable} if $\fix(\ell) \not \in \aicup(r)$. 
A database 
$\Dmc$ \emph{satisfies} $r$, denoted $\Dmc\models r$, if it satisfies $\tau_r$. It satisfies a set of AICs $\eta$, denoted $\Dmc\models\eta$, if $\Dmc\models r$ for every $r\in\eta$. 
A set of AICs is \emph{consistent} if there exists a database $\Dmc$ such that $\Dmc\models\eta$. 

A \emph{ground} AIC is an AIC that contains no variables. 
The set $gr_\Dmc(r)$ 
contains all  
ground AICs obtained from $r$ by (i) replacing variables by constants from $\domain{\Dmc}$, 
(ii) removing all true $c \neq d$ atoms, and (iii) removing all ground AICs with an atom $c \neq c$. 
We let $gr_\Dmc(\eta):=\bigcup_{r\in\eta}gr_\Dmc(r)$, and 
observe that $\Dmc\models r$ iff $\Dmc\models r_g$ for every $r_g\in gr_\Dmc(r)$. 

An AIC is called \emph{normal} if $|\aicup(r)|=1$. 
The \emph{normalization} of an AIC 
$r$ is the set of AICs $N(r)= \{\aicbody(r) \rightarrow \{A\} \mid A \in \aicup(r)\}$. 
The normalization of a set of AICs $\eta$ is $N(\eta)=\bigcup_{r\in\eta}N(r)$. Note that $gr_\Dmc(N(\eta))=N(gr_\Dmc(\eta))$.

\subsubsection*{Repair updates} 
A \emph{repair update (r-update)}\footnote{Repair updates are usually called repairs in the AIC literature, we use this term to avoid confusion with the other repair notions.} of a database  $\Dmc$ \wrt a set of AICs $\eta$ is a consistent subset-minimal set of update actions $\Umc$ such that $\Dmc\circ\Umc\models \eta$. We denote the set of r-updates of $\Dmc$ \wrt $\eta$ by $\ups{\Dmc,\eta}$. 
It is easy to check that $\{\Dmc\circ\Umc\mid\Umc\in\ups{\Dmc,\eta}\}=\deltareps{\Dmc,\constraints_\eta}$ where $\constraints_\eta$ is the set of universal constraints that correspond to AICs in $\eta$. 

To take into account the restrictions on the possible update actions expressed by the AICs, several classes of r-updates have been defined. 
The first one, \emph{founded} r-updates \cite{DBLP:conf/iclp/CaropreseGSZ06}, was criticized for exhibiting circularity of support, leading to the introduction of more restrictive \emph{justified} \cite{DBLP:journals/tplp/CaropreseT11}, \emph{well-founded} \cite{DBLP:conf/tase/Cruz-FilipeGEN13}, and \emph{grounded} r-updates \cite{DBLP:conf/iclp/Cruz-Filipe16}. 
The latter were motivated by arguably unexpected behaviors of justified and well-founded r-updates. 
In particular, justified r-updates are criticized for being too complicated and for excluding some r-updates that seem reasonable. 
Moreover, they are sensitive to normalization, unlike founded, well-founded and grounded r-updates.  

\begin{definition}
An r-update $\Umc$ of $\Dmc$ \wrt $\eta$ is:
\begin{itemize}
\item \emph{founded} if for every $A\in\Umc$, there exists $r\in gr_\Dmc(\eta)$ such that $A \in \aicup(r)$ 
and $\Dmc\circ\Umc\setminus\{A\}\not\models r$. 
\item \emph{well-founded} if there exists a sequence of actions $A_1,\dots,A_n$ such that $\Umc=\{A_1,\dots,A_n\}$, and for every $1\leq i\leq n$, there exists $r_i\in gr_\Dmc(\eta)$ such that $A_i \in \aicup(r_i)$  
and $\Dmc\circ\{A_1,\dots, A_{i-1}\}\not\models r_i$. 
\item \emph{grounded} if for every $\Vmc\subsetneq\Umc$, there exists $r\in gr_\Dmc(N(\eta))$ such that $\Dmc\circ\Vmc\not\models r$ and
$\aicup(r) \subseteq \Umc\setminus\Vmc$. 
\item \emph{justified} if $ ne(\Dmc,\Dmc\circ\Umc)\cup\Umc$ is a minimal set of update actions closed under $\eta$ that contains $ne(\Dmc,\Dmc\circ\Umc)$ where 
\begin{itemize}
\item $ne(\Dmc,\Dmc\circ\Umc)=\{+\alpha\mid \alpha\in\Dmc\cap(\Dmc\circ\Umc)\}\cup\{-\alpha \mid \alpha\notin\Dmc\cup (\Dmc\circ\Umc), \alpha\in\factset\}$ 
(set of \emph{no-effect actions}) 
\item 
$\Umc$ is \emph{closed under $\eta$} if for every $r\in gr_\Dmc(\eta)$, if $\Umc$ satisfies all the non-updatable literals of $r$, 
 then $\Umc$ contains an update action from $r$. 
\end{itemize}
\end{itemize}
We denote by $\foundups{\Dmc,\eta}$, $\wellfoundups{\Dmc,\eta}$, $\groundups{\Dmc,\eta}$ and $\justifups{\Dmc,\eta}$ respectively the sets of founded, well-founded, grounded and justified r-updates of $\Dmc$ \wrt $\eta$ and let $\deltaxreps{\Dmc,\eta}=\{\Dmc\circ\Umc\mid\Umc\in\xups{\Dmc,\eta}\}$ be the set of corresponding repairs. 
\end{definition}

\citeauthor{DBLP:journals/eswa/CalauttiCGMTZ21} \shortcite{DBLP:journals/eswa/CalauttiCGMTZ21} recently redefined founded r-updates. In fact, we show that their definition coincides with grounded r-updates, yielding the following characterization.

\begin{restatable}{proposition}{NewFoundedGrounded}\label{prop:new-founded-grounded}
For every $\Umc \in \ups{\Dmc,\eta}$, 
$\Umc$ is grounded iff 
$\Umc \in \ups{\Dmc,\eta[\Umc]}$, 
where $\eta[\Umc]$ is the set of AICs 
obtained from $gr_\Dmc(\eta)$ by deleting update actions not occurring in $\Umc$ and AICs whose update actions have all been deleted. 
\end{restatable}

The relationships between the various kinds of repairs are represented below, where a plain arrow from $X$ to $Y$ means $X\subseteq Y$ and the dotted arrow represents an inclusion that only holds when $\eta$ is a set of normal AICs. All inclusions may be strict 
\cite{DBLP:journals/tplp/CaropreseT11,DBLP:conf/tase/Cruz-FilipeGEN13,DBLP:conf/iclp/Cruz-Filipe16}.

{\centering{
\begin{tikzpicture} 
\node [left] at (0,1.5) {$\groundreps{\Dmc,\eta}$};
\node [right] at (1,1.5) {$\foundreps{\Dmc,\eta}$};
\node [left] at (0,0.75) {$\justifreps{\Dmc,\eta}$};
\node [right] at (1,0.75) {$\wellfoundreps{\Dmc,\eta}$};

\draw[<-] (1,1.5) -- (0,1.5);
\draw[<-] (1,1.5) -- (0,0.75);
\draw[dotted, <-] (0,1.5) -- (0,0.75);
\draw[<-] (1,0.75) -- (0,1.5);
\draw[<-] (1,0.75) -- (0,0.75);
\end{tikzpicture}
}}

\subsection{From Prioritized Databases to AICs}\label{subsec:prio-to-aics}
Given a prioritized database $\Dmc^\constraints_\succ$ we define the following set of ground AICs: $\eta^\constraints_\succ=\{r_\conf \mid \conf\in\conflicts{\Dmc,\constraints}\}$ where $$ r_\conf:=\bigwedge_{\lambda\in\conf}\lambda \rightarrow\{\mi{fix}(\lambda)\mid \lambda\in\conf,\forall\mu\in\conf, \lambda\not\succ\mu\}. $$  
Intuitively, $\eta^\constraints_\succ$ expresses that conflicts of $\Dmc^\constraints_\succ$ should be fixed by modifying the least preferred literals 
according to $\succ$. 

We can prove that Pareto-optimal repairs of $\Dmc^\constraints_\succ$ coincide with several kinds of repairs of $\Dmc$ \wrt $\eta^\constraints_\succ$. 

\begin{restatable}{proposition}{ReductionPrioToAICs}\label{prop:reductionPrioToAICs}
For every prioritized database $\Dmc^\constraints_\succ$, 
$\deltapreps{\Dmc^\constraints_\succ}=\justifreps{\Dmc,\eta^\constraints_\succ}=\groundreps{\Dmc,\eta^\constraints_\succ}=\foundreps{\Dmc,\eta^\constraints_\succ}\subseteq\wellfoundreps{\Dmc,\eta^\constraints_\succ}$.
\end{restatable}

This result is interesting not only because it provides additional evidence for the naturalness of Pareto-optimal repairs, but also because it identifies a  class of AICs for which justified, grounded, and founded r-updates coincide. The proof in fact shows that these three notions coincide for every set of ground AICs $\eta$ that is \emph{monotone}, i.e. does not contain both a fact $\alpha$ and the complementary literal $\neg \alpha$. 

We remark that the final inclusion in Proposition \ref{prop:reductionPrioToAICs} may be strict.
This is demonstrated on the next example, which suggests that 
well-founded repairs may be too permissive: 

\begin{example}\label{ex:AICsfromPrioKB-well-founded-diff}
It is possible to construct a prioritized database $\Dmc^\constraints_\succ$ where 
$\Dmc=\{\alpha,\beta,\gamma,\delta\}$ and $\eta^\constraints_\succ = \{\alpha\wedge\beta\rightarrow\{-\beta\},\  \alpha\wedge\gamma\rightarrow\{-\alpha\},\ \gamma\wedge\delta\rightarrow\{-\gamma\}\}$. 
For the AICs $\eta^\constraints_\succ$, the r-update $\{-\alpha, -\gamma\}$  is well-founded, but not founded,
as the only founded r-update is $\{-\beta,-\gamma\}$.  
We argue that $\{-\beta,-\gamma\}$ should indeed be preferred to $\{-\alpha, -\gamma\}$, since the first AIC expresses that it is better to remove $\beta$ than $\alpha$. 
\end{example}

The reduction used to show Proposition \ref{prop:reductionPrioToAICs} 
is data-dependent and requires us to create potentially exponentially 
many ground AICs, one for every conflict. 
In the case of \emph{denial constraints}, however, we can give 
an alternative data-independent reduction, provided that 
the priority relation $\succ$ is specified in the database. 
We thus assume for the next result that $P_\succ$ is a predicate in $\preds$,  
that the first attribute of each relation in $\preds\setminus\{P_\succ\}$ stores a unique fact identifier, 
and $P_\succ$ stores pairs of such identifiers. 
Then given a set 
of denial constraints $\constraints$ over $\preds\setminus\{P_\succ\}$, we 
build a set $\minnongr(\constraints)$ that is equivalent to $\constraints$ but has 
the property that the conflicts of $\Dmc$ w.r.t.\ $\Cmc$ are precisely 
the images of constraint bodies of $\minnongr(\constraints)$ on $\Dmc$. 
This can be achieved by replacing each $\varphi \rightarrow \bot \in \Cmc$ with all refinements 
obtaining by (dis)equating variables in $\varphi$ with each other, or with constants mentioned in $\Cmc$, 
then removing any subsumed constraints. 
For example, if $\Cmc = \{R(x,x) \rightarrow \bot, R(x,y) \wedge S(y) \rightarrow \bot\}$, 
then $\minnongr(\constraints)$ contains $R(x,x) \rightarrow \bot$ and 
$R(x,y) \wedge S(y) \wedge x\neq y \rightarrow \bot$, so $\{R(a,a),S(a)\}$ 
is no longer an image of a constraint body. 
We then define $\eta^\constraints$ as the set of all AICs
$$\bigl(\ell_1 \wedge \ldots \wedge \ell_n \wedge \varepsilon \wedge   \bigwedge_{\ell_j \neq \ell_i}   \neg P_\succ (id_i,id_j)\bigr)  \rightarrow \{-\ell_i\}.$$
such that $\ell_1 \wedge \ldots \wedge\ell_n \wedge \varepsilon \rightarrow \bot \in \minnongr(\constraints)$, $i \in \{1, \ldots, n\}$, 
and for every $1 \leq k \leq n$, 
$\ell_k=R(id_k,\vec{t})$ for some $R,\vec{t}$. 

\begin{restatable}{proposition}{ReductionDenialPrioToAICs}\label{prop:reductionDenialPrioToAICs}
For every set of denial constraints $\constraints$, database $\Dmc$ and priority relation $\succ$ of $\Dmc$ \wrt $\constraints$, $\deltapreps{\Dmc^\constraints_\succ}=\justifreps{\Dmc,\eta^\constraints}=\groundreps{\Dmc,\eta^\constraints}=\foundreps{\Dmc,\eta^\constraints}\subseteq\wellfoundreps{\Dmc,\eta^\constraints}$.
\end{restatable}

This reduction could be used for example to transfer data complexity lower bounds for prioritized databases with denial constraints to the setting of AICs.

\subsection{Towards Well-Behaved AICs}\label{subsec:well-behaved}

When translating a prioritized database into AICs, we obtained monotone sets of AICs, 
for which most of the different kinds of r-update coincide. Can we generalize this idea to obtain 
larger classes of `well-behaved' sets of AICs which share this desirable behavior? This subsection explores 
this question and provides some first insights. 

We start by defining the following condition, which serves to ensure that 
 all constraints are made explicit:

\begin{definition}
We say that a set $\eta$ of ground AICs is \emph{closed under resolution} if it is consistent, 
and for every pair of AICs $r_1, r_2\in \eta$, if there exists $\alpha \in \aiclits(r_1)$ such that $\neg \alpha \in \aiclits(r_2)$,  
and $\aiclits(r_1) \cup \aiclits(r_2) \setminus \{\alpha, \neg \alpha\}$ is a consistent set of literals, 
then there exists $r_3\in\eta$ with $\aiclits(r_3)=\aiclits(r_1) \cup \aiclits(r_2) \setminus \{\alpha, \neg \alpha\}$. 
 A set of AICs $\eta$ is closed under resolution if so is $gr_\Dmc(\eta)$ for every database  $\Dmc$.
\end{definition}

The name `closure under resolution' comes from considering the clauses that correspond to the negation of the rule bodies: 
if we have AICs whose clauses are $\neg\alpha\vee \varphi$ and $\alpha\vee \psi$, then 
we should also have an AIC for their resolvent 
$\varphi\vee\psi'$, corresponding to the implied constraint $\neg \varphi \wedge \neg \psi \rightarrow \bot$. 
This property ensures that $\eta$ captures all potential conflicts: 
for every $\Dmc$, if 
$\constraints_\eta=\{\tau_r \mid r \in \eta\}$, then 
$\conflicts{\Dmc,\constraints_\eta}=\{\aiclits(r)\mid r\in gr_\Dmc(\eta), \Dmc\not\models r,
\text { and there is no }  r'\in gr_\Dmc(\eta) \text{ with } \aiclits(r') \subsetneq \aiclits(r)\}$. 

The following example, given by \citeauthor{DBLP:journals/ai/BogaertsC18} \shortcite{DBLP:journals/ai/BogaertsC18} to show that grounded r-updates do not coincide with the intersection of founded and well-founded r-updates, illustrates that sets of AICs not closed under resolution may exhibit undesirable behaviors. 

\begin{example}
Consider $\Dmc=\emptyset$ and $\eta$ that contains the AICs:
\begin{align*}
r_1&: \neg\alpha\wedge\neg\beta\rightarrow \{+\alpha\} & r_4&: \alpha\wedge\beta\wedge\neg\gamma\rightarrow \{+\gamma\}\\
r_2&: \alpha\wedge\neg\beta\rightarrow \{+\beta\}& r_5&: \alpha\wedge\neg\beta\wedge\gamma\rightarrow \{+\beta\}\\
r_3&:  \neg\alpha\wedge\beta\rightarrow \{-\beta\} &r_6&: \neg\alpha\wedge\beta\wedge\gamma\rightarrow \{+\alpha\}
\end{align*}
$\Umc=\{+\alpha,+\beta,+\gamma\}$ is founded and well-founded but is not grounded: taking $\Vmc=\{+\beta\}$, we have $\Vmc \subsetneq\Umc$ but there is no $r\in\eta$ such that $\{\beta\}\not\models r$ and 
$\aicup(r) \cap \{+\alpha,+\gamma \} \neq \emptyset$. 

However, it can be verified that $\Umc$ is in fact the only r-update of $\Dmc$ \wrt $\eta$. Indeed, the conflicts of $\Dmc$ \wrt the constraints expressed by $\eta$ are $\{\no\alpha\}$, $\{\no\beta\}$ and $\{\no\gamma\}$. 

If $\eta$ were closed under resolution, it would contain $\neg\alpha\rightarrow \{+\alpha\}$, $\neg\beta\rightarrow \{+\beta\}$, and $\neg\gamma\rightarrow \{+\gamma\}$, in which case $\Umc$ would be grounded, as expected for the unique r-update.
\end{example}

It is always possible to transform a set of ground AICs into one that is closed under resolution by adding the required AICs. However this may result in an exponential blowup. 
Moreover, we need to choose the update actions of the added AICs. We advocate for this to be done by propagating the relevant update actions of the rules on which the resolution is done. 
A set of ground AICs obtained in this way will be closed under resolution and will \emph{preserve actions under resolution} according to the following definition.

\begin{definition}
We say that a set $\eta$ of ground AICs \emph{preserves actions under resolution} if for every triple of AICs $r_1$, $r_2$, $r_3$ $\in \eta$, 
if there exists $\alpha$ such that  $\alpha \in \aiclits(r_1)$, $\neg \alpha \in \aiclits(r_2)$, and $\aiclits(r_3)=\aiclits(r_1) \cup \aiclits(r_2) \setminus \{\alpha, \neg \alpha\}$, 
then $\aicup(r_1) \cup \aicup(r_2) \setminus \{+\alpha,-\alpha\} \subseteq \aicup(r_3)$. 
A set of AICs $\eta$ preserves actions under resolution if so does $gr_\Dmc(\eta)$ for every database  $\Dmc$.
\end{definition}

The next example shows that a set of AICs 
which does not preserve actions under resolution may be ambiguous. 
\begin{example}\label{ex:faithful-implicants}
Let $\Dmc=\{\alpha,\beta,\gamma\}$, and 
$\eta$ that contains:
\begin{align*}
r_1&: \alpha\wedge \beta \rightarrow \{-\alpha\}  &
r_3&: \alpha\wedge\neg\delta\rightarrow \{+\delta\} \\
r_2&: \beta\wedge \gamma \rightarrow \{-\gamma\} 
&
r_4&: \beta\wedge\delta\rightarrow \{-\beta\}
\end{align*}
This set of AICs is closed under resolution but does not preserve actions under resolution: due to $r_3$ and $r_4$, $-\beta$ should be an update action of $r_1$. 
Indeed, $r_3$ and $r_4$ together indicate that if $\alpha$ and $\beta$ are present, $\beta$ should be removed (since if $\delta$ is absent, it should be added, due to $r_3$, and $\beta$ should be removed when $\delta$ is present, by $r_4$). 

To make $\eta$ preserve actions under resolution, there are three possibilities: (1) change $r_1$ to $\alpha\wedge \beta \rightarrow \{-\beta\}$ (if $\alpha$ is preferred to $\beta$), or (2) change $r_4$ to $\beta\wedge\delta\rightarrow \{-\delta\}$ (if $\beta$ is preferred to $\alpha$), or (3) change $r_1$  to $\alpha\wedge \beta \rightarrow \{-\alpha,-\beta\}$  (if neither $\alpha$ nor $\beta$ is preferred to the other). 
\end{example}

Sets of AICs that are closed under resolution and preserve actions under resolution are well behaved in the sense that they make most of the r-update notions coincide. The monotone sets of AICs mentioned in relation to Proposition \ref{prop:reductionPrioToAICs} trivially satisfy these two conditions.

\begin{restatable}{proposition}{ClosedFaithfulCollapse}\label{prop:founded-grounded-closed-min-faith}
If $\eta$ is closed under resolution and preserves actions under resolution, then for every database $\Dmc$, $\justifreps{\Dmc,\eta}=\groundreps{\Dmc,\eta}=\foundreps{\Dmc,\eta}\subseteq\wellfoundreps{\Dmc,\eta}$. 
\end{restatable}

The next example 
shows that 
both conditions 
are necessary for obtaining Proposition~\ref{prop:founded-grounded-closed-min-faith}.

\begin{example}\label{ex:difference-founded-grounded}
Consider $\Dmc=\{\alpha,\beta,\gamma\}$ and the two sets 
\begin{align*}
\eta_1=&\{
\alpha\wedge\neg\beta\rightarrow \{-\alpha\}, \quad
\neg\alpha\wedge\beta\rightarrow\{-\beta\}, \\&\quad
\alpha\wedge\beta\wedge\gamma\rightarrow \{-\gamma\}
\}\\
\eta_2=&\eta_1\cup\{
\alpha\wedge\gamma\rightarrow\{-\gamma\}, \quad
\beta\wedge\gamma\rightarrow\{-\gamma\}\}.
\end{align*} 
$\eta_1$ is not closed under resolution but (trivially) preserves actions under resolution, while $\eta_2$ is closed under resolution but does not preserve actions under resolution. 

In both cases, there are two founded r-updates: $\{-\gamma\}$ and $\{-\alpha,-\beta\}$. However, $\{-\alpha,-\beta\}$ is not well-founded, hence not grounded nor justified. 
Indeed, $\Dmc$ violates only AICs whose only update action is $-\gamma$.
\end{example}

Even if a set of AICs is such that justified, grounded and founded repairs are guaranteed to exist and coincide, its behavior may still be puzzling, as illustrated next. 

\begin{example}\label{ex:faithfulstronger}
Let $\Dmc=\{\alpha,\beta,\gamma,\delta\}$ and $\eta$ be the monotone set of AICs comprising the following AICs:
\begin{align*}
r_1&:\alpha\wedge\delta \rightarrow \{-\delta\}
&
r_3&:\alpha\wedge\beta\wedge\gamma\wedge\delta \rightarrow \{-\beta\}\\
r_2&: \alpha\wedge\beta\wedge\delta \rightarrow \{-\alpha\}
&
r_4&:\beta\wedge\gamma \rightarrow \{-\gamma\}
\end{align*}
There are four r-updates:
\begin{align*}
\Umc_1=&\{-\alpha,-\gamma\}\text{ and }\Umc_2=\{-\delta,-\gamma\}\text{ are founded}\\
\Umc_3=&\{-\delta,-\beta\}\text{ is not founded but is well-founded}\\
\Umc_4=&\{-\alpha,-\beta\} \text{ is not founded nor well-founded} 
\end{align*}
There are two conflicts: $\{\alpha,\delta\}$ and $\{\beta,\gamma\}$. 
It is natural to prefer removing $\gamma$ rather than $\beta$ to resolve the latter conflict (due to $r_4$), which would justify to preferring $\Umc_1$ and $\Umc_2$ over $\Umc_4$ and $\Umc_3$ respectively. However, the exact same argument applied to $r_1$ should lead us to prefer removing $\delta$ to solve the first conflict, thus to prefer $\Umc_2$ over $\Umc_1$. It is therefore not clear why both $\Umc_1$ and $\Umc_2$ should be the preferred r-updates. 
The intention of a user specifying the preceding AICs is probably quite far from their actual behavior. 
\end{example}

We thus believe that a reasonable property for sets of AICs is to respect the principle that adding atoms to a rule body can only restrict the possible update actions. We call the \emph{anti-normalization} of a set $\eta$ of AICs the set $AN(\eta)$ of AICs that replace all the AICs $r_1,\dots,r_n\in\eta$ that share the same body by a single AIC whose update actions are the union of the update actions of $r_1,\dots,r_n$. 

\begin{definition}
We say that a set $\eta$ of ground AICs \emph{preserves actions under strengthening} if for every pair of AICs $r_1$, $r_2$ in $AN(\eta)$, if 
$\aiclits(r_1) \subseteq \aiclits(r_2)$, then $\aicup(r_2) \subseteq \aicup(r_1)$. 
A set of AICs $\eta$ preserves actions under strengthening if so does $gr_\Dmc(\eta)$ for every database  $\Dmc$.
\end{definition}

The following proposition shows that if $\eta$ preserves actions under strengthening, then constraints that have non-minimal bodies have no influence on the r-updates.

\begin{restatable}{proposition}{StrongerMinBod}\label{prop:stronger-min-bod}
Let $\eta$ be a set of ground AICs and $\minnongr(\eta)$ be the set of AICs from $AN(\eta)$ that have (subset-)minimal bodies. 
If $\eta$ 
preserves actions under strengthening, then for every $\Dmc$, for $X\in\{\mi{Found},\mi{WellFound},\mi{Ground}\}$ 
$\xups{\Dmc,\eta}=\xups{\Dmc,AN(\eta)}=\xups{\Dmc,\minnongr(\eta)}$, and $\justifups{\Dmc,AN(\eta)}=\justifups{\Dmc,\minnongr(\eta)}$. 
\end{restatable}

\subsection{From AICs to Prioritized Databases}\label{subsec:aics-to-prio}

We next study the possibility of reducing well-behaved sets of AICs to prioritized databases and discuss the differences between the two settings.

\subsubsection*{Binary conflicts case} 
We first consider the case where the size of the conflicts is at most two (this covers, for example, AIC bodies corresponding to functional dependencies or class disjointness). 
In this case, given a set $\eta$ of AICs \emph{closed under resolution that preserves actions under resolution and under strengthening} and a database $\Dmc$, we build a set of constraints $\constraints_\eta$ and a binary relation $\succ_\eta$ such that if $\succ_\eta$ is acyclic, the Pareto-optimal repairs of $\Dmc^{\constraints_\eta}_{\succ_\eta}$ coincide with the founded, grounded and justified repairs of $\Dmc$ \wrt $\eta$. 
We take $\constraints_\eta=\{\tau_r \mid r \in \eta\}$ and define $\succ_\eta$ so that 
$\lambda\succ_\eta\mu$ iff 
\begin{itemize}
\item there exists $r\in \mingr(\eta)$ such that $\Dmc\not\models r$, $\{\lambda,\mu\}\subseteq  \aiclits(r)$, and $\mi{fix}(\mu) \in \aicup(r)$; 
and
\item for every $r\in \mingr(\eta)$ such that $\Dmc\not\models r$ and $\{\lambda,\mu\}\subseteq \aiclits(r)$, $\mi{fix}(\lambda) \not \in \aicup(r)$, 
\end{itemize} 
where $\mingr(\eta) = \{r \in gr_\Dmc(\eta) \mid \text{ there is no }   r'\in gr_\Dmc(\eta) $ $ \text{with } \aiclits(r') \subsetneq \aiclits(r)\}$. 
As $\eta$ is closed under resolution, $\conflicts{\Dmc,\constraints_\eta}=\{\aiclits(r)\mid r\in \mingr(\eta), \Dmc\not\models r\}$.

\begin{restatable}{proposition}{ReductionAICsPrioBinary}\label{prop:reduction-AICs-prio-binary}
If $\eta$ is closed under resolution, preserves actions under resolution and under strengthening, the size of the conflicts of $\Dmc$ \wrt $\eta$ is bounded by 2, and $\succ_\eta$ is acyclic, then 
$\deltapreps{\Dmc^{\constraints_\eta}_{\succ_\eta}}=\justifreps{\Dmc,\eta}=\groundreps{\Dmc,\eta}$ $=\foundreps{\Dmc,\eta}\subseteq\wellfoundreps{\Dmc,\eta}$.
\end{restatable}

The following examples show that the three first conditions are necessary.

\begin{example}
Let $\Dmc=\{\alpha,\beta,\gamma\}$ and $\eta=\{\alpha\wedge\beta\rightarrow\{-\beta\},$ $\neg\beta\wedge\gamma\rightarrow\{-\gamma\}\}$, which preserves actions under resolution and strengthening but is not closed under resolution. We have $\conflicts{\Dmc,\constraints_\eta}=\{\{\alpha,\beta\},\{\alpha,\gamma\}\}$ and $\alpha\succ_\eta\beta$. 
Both $\{\alpha\}$ and $\{\beta,\gamma\}$ are Pareto-optimal, 
but the only founded r-update (which is also grounded and justified) is $\{-\beta,-\gamma\}$. 
\end{example}

\begin{example}[Example \ref{ex:faithful-implicants} cont'd]
In Example~\ref{ex:faithful-implicants}, $\eta$ is closed under resolution and preserves actions under strengthening but not under resolution. We have $\conflicts{\Dmc,\constraints_\eta}=\{\{\alpha,\beta\},\{\beta,\gamma\}, \{\alpha,\no\delta\}\}$ and 
 $\beta\succ_\eta\alpha$, $\beta\succ_\eta\gamma$, $\alpha\succ_\eta\no\delta$. 
The only Pareto-optimal repair is $\{\beta\}$, 
but $\{-\beta,+\delta\}$ is a founded, grounded and justified r-update. 
\end{example}

\begin{example}[Example \ref{ex:faithfulstronger} cont'd]
In Example~\ref{ex:faithfulstronger}, $\eta$ is closed under resolution and 
preserves actions under resolution but not under strengthening. We have $\conflicts{\Dmc,\constraints_\eta}=\{\{\alpha,\delta\},\{\beta,\gamma\}\}$ and 
 $\alpha\succ_\eta\delta$, $\beta\succ_\eta\gamma$. 
The only Pareto-optimal repair is $\{\alpha,\beta\}$,  
but $\{-\alpha,-\gamma\}$ is a founded, grounded and justified r-update. 
\end{example}

Note that $\succ_\eta$ may be cyclic: if $\eta=\{A(x)\wedge B(x)\rightarrow \{-A(x)\}, B(x)\wedge C(x)\rightarrow \{-B(x)\}, C(x)\wedge A(x)\rightarrow \{-C(x)\}\}$ and $\Dmc=\{A(a),B(a),C(a)\}$, we obtain $A(a)\succ_\eta C(a)\succ_\eta B(a)\succ_\eta A(a)$.

\subsubsection*{General case} 
Let us now consider the case where the size of the conflicts is not bounded. If we apply the same reduction, we can only show the following inclusions between repairs of $\Dmc$ \wrt $\eta$ and Pareto-optimal repairs of $\Dmc^{\constraints_\eta}_{\succ_\eta}$:

\begin{restatable}{proposition}{ReductionAICsPrioGeneral}\label{prop:AICs-prio-non-binary}
If $\eta$ is closed under resolution, preserves actions under resolution and under strengthening,  and $\succ_\eta$ is acyclic, then 
$\justifreps{\Dmc,\eta}=\groundreps{\Dmc,\eta}=\foundreps{\Dmc,\eta}\subseteq\deltapreps{\Dmc^{\constraints_\eta}_{\succ_\eta}}$.
\end{restatable}

The next example shows that the inclusion may be strict. 

\begin{example}
Let $\Dmc=\{\alpha,\beta,\gamma,\delta,\epsilon\}$ and $\eta$ consist of: 
\begin{align*}
r_1&:  \alpha\wedge\beta\wedge\gamma \rightarrow \{-\beta\}  
&
r_3&:  \delta\wedge\epsilon\rightarrow \{-\delta\}\\
r_2&:  \alpha\wedge\beta\wedge\delta\rightarrow \{-\alpha,-\beta\} 
\end{align*}
We obtain $\conflicts{\Dmc,\constraints_\eta}=\{\{\alpha,\beta,\gamma\},\{\alpha,\beta,\delta\}, \{\delta,\epsilon\}\}$ and 
 $\gamma\succ_\eta\beta$, $\delta\succ_\eta\alpha$, $\delta\succ_\eta\beta$ and $\epsilon\succ_\eta\delta$ (note that $\alpha\not\succ_\eta\beta$ because $-\alpha$ is an update action of $r_2$). 
The 
repair $\{\beta,\gamma,\epsilon\}$ is Pareto-optimal, but the corresponding r-update $\{-\alpha,-\delta\}$ is not founded, as 
$-\alpha$ appears only in $r_2$ 
and $\Dmc\circ\{-\delta\}\models r_2$. 

One might try to 
modify the definition of $\succ_\eta$ by 
dropping the second condition and adding $\mi{fix}(\lambda) \not \in \aicup(r)$ to the first. In this case, 
$\{\beta,\gamma,\epsilon\}$ is no longer Pareto-optimal. However, now if we take 
$\eta' = \eta\setminus\{r_3\}$, then $\{-\alpha\}$ would be founded, but the corresponding repair 
  $\{\beta,\gamma,\delta,\epsilon\}$ would not be Pareto-optimal, violating the 
 inclusion of Proposition \ref{prop:AICs-prio-non-binary}.
\end{example}

This example shows that even for AICs corresponding to denial constraints, 
there is no clear way to define a priority relation 
that captures the preferences expressed  by the AICs. 

\section{Conclusion and Future Work}\label{sec:conclusion}
We studied how to incorporate preferences into repair-based query answering
for an expressive setting in which databases are equipped with universal constraints, 
and both fact additions and deletions are used to restore consistency. 
We showed that the existing framework of prioritized databases 
could be faithfully adapted to 
this richer setting, 
although the proofs are more involved and crucially rely upon finding
the right definition of what constitutes a conflict. 
While these results focus on databases, we expect that they will also 
prove useful for exploring  
symmetric difference repairs in related KR settings, 
e.g.\ ontologies with closed predicates. 

Our complexity analysis showed that adopting optimal repairs in place of 
symmetric difference repairs does not increase the complexity of repair-based
query answering. A major difference between denial and universal constraints 
is that the latter may lead to conflicts of unbounded size. We showed that 
it is intractable to recognize a conflict and that several problems 
drop in complexity if we assume that the conflicts are available. This suggests
the interest of developing structural conditions on constraint sets that ensure
easy-to-compute conflicts, as well as 
practical algorithms for computing 
and updating the set of conflicts, which could enable an integration with
existing SAT-based approaches. 

Intrigued by the high-level similarities between prioritized databases and active integrity constraints,
we explored 
how the two formalisms relate. We exhibited a natural translation 
of prioritized databases into AICs whereby Pareto-optimal repairs coincide with 
founded, grounded and justified repairs \wrt the generated set of AICs.
We take this as further evidence that Pareto-optimal repairs are an especially natural 
notion 
(we previously showed 
that Pareto-optimal (subset) repairs 
correspond to stable extensions in argumentation  \cite{DBLP:conf/kr/BienvenuB20}). 
It would be of interest to extend our comparison 
to other more recent notions of repair updates 
for AICs 
(\citeauthor{DBLP:journals/fuin/FeuilladeHR19} \citeyear{DBLP:journals/fuin/FeuilladeHR19},
\citeauthor{DBLP:journals/ai/BogaertsC18} \citeyear{DBLP:journals/ai/BogaertsC18,DBLP:journals/tocl/BogaertsC21}).

Our work also provided new insights into AICs. Existing examples used to 
distinguish different notions of r-update often 
seem unnatural in some respect. This led us to devise a set of criteria for `well-behaved' AICs, 
which provide sufficient conditions for founded, grounded and justified repairs to coincide 
(Example \ref{ex:AICsfromPrioKB-well-founded-diff} suggests that well-founded repairs 
are too permissive). 
Even for such restricted AICs, it is not always clear what user intentions are being captured. 
We thus believe that there is still work to be done to develop user-friendly 
formalisms for expressing constraints and preferences on how to handle
constraint violations. 

\section*{Acknowledgements}
This work was supported by the ANR AI Chair INTENDED (ANR-19-CHIA-0014).

\bibliographystyle{kr}
\bibliography{ms}

\newpage
\onecolumn
\appendix

\section{Proofs for Section \ref{sec:opti-repairs}}

\subsection{Proofs for Section \ref{subsec:conflicts}}

Before proceeding with the proof of Proposition~\ref{prop:defconflicts}, let us illustrate the alternative characterizations of conflicts on Example~\ref{ex:conflicts}.

\begin{example}[Example \ref{ex:conflicts} cont'd]
Recall that $\Dmc= \{A(a), B(a)\}$, $\constraints=\{\tau_1,\tau_2, \tau_3\}$, where 
$\tau_1:= A(x)\rightarrow C(x)$, $\tau_2:=B(x)\rightarrow D(x)$, and $\tau_3:=C(x)\wedge D(x)\rightarrow \bot$, and 
\begin{align*}
\deltareps{\Dmc,\constraints}=  \{&\emptyset, \{A(a), C(a)\} , \, \{B(a),D(a)\} \}\\
\conflicts{\Dmc,\constraints}=  \{& \{A(a), \no C(a)\}, \{B(a), \no D(a)\},   \{A(a), B(a)\}\}.
\end{align*}
\begin{enumerate}
\item $\{\Rmc\Delta\Dmc\mid \Rmc \in\deltareps{\Dmc,\constraints}\}=\{ \{A(a), B(a)\},$ $\{B(a), C(a)\} ,$ $\{A(a),D(a)\} \}$ so  
the minimal hitting sets of the symmetric differences between $\Dmc$ and each $\Rmc \in\deltareps{\Dmc,\constraints}$ are  
$\mhs{\Dmc,\constraints}=  \{\{A(a), C(a)\}, \, \{B(a), D(a)\},\, \{A(a), B(a)\} \}$. Transforming these sets of facts into sets of literals by negating $C(a)$ and $D(a)$ which do not belong to $\Dmc$ give the conflicts  of $\Dmc$ \wrt $\constraints$. 

\item The prime implicants of $c_1\vee c_2\vee c_3$ where $c_1:=A(a)\wedge\neg C(a)$, $c_2:=B(a)\wedge\neg D(a)$, and $c_3:=C(a)\wedge D(a)$, are $c_1$, $c_2$, $c_3$, $c_4:=A(a)\wedge D(a)$, $c_5:=B(a)\wedge C(a)$, and $c_6:=A(a)\wedge B(a)$. 
Since $\litset=\{A(a), B(a),\neg C(a),\neg D(a)\}$, it follows that the sets of literals of $c_1$, $c_2$, and $c_6$ are the conflicts of $\Dmc$ \wrt $\constraints$. 
\end{enumerate}
\end{example}

\DefConflicts*
\begin{proof}
\noindent $(1)$ 
Let $\conf\in\conflicts{\Dmc,\constraints}$ and assume for a contradiction that $\conf\notin \{ \Hmc\cap\Dmc\cup\{\no\alpha \mid \alpha\in\Hmc\setminus\Dmc\} \mid \Hmc\in \mhs{\Dmc,\constraints} \}$. This means that either (i) $\Hmc=\{\alpha\mid \alpha\text{ or }\neg\alpha\in\conf\}$ is not a hitting set of $\{\Rmc\Delta\Dmc\mid \Rmc \in\deltareps{\Dmc,\constraints}\}$, or (ii) there exists $\Hmc'\subsetneq\Hmc$ which is a hitting set of $\{\Rmc\Delta\Dmc\mid \Rmc \in\deltareps{\Dmc,\constraints}\}$.
\begin{itemize}
\item In case (i), there exists $\Rmc\in\deltareps{\Dmc,\constraints}$ such that $\Hmc\cap(\Rmc\Delta\Dmc)=\emptyset$. 
Let $\lambda\in\conf$. If $\lambda=\alpha$, $\alpha\in\Dmc$ so $\alpha\notin\Rmc\Delta\Dmc$ implies $\alpha\in\Rmc$. 
If $\lambda=\neg\alpha$, $\alpha\notin\Dmc$ so $\alpha\notin\Rmc\Delta\Dmc$ implies $\alpha\notin\Rmc$. 
Since in both cases $\alpha\in\Hmc$, then $\alpha\notin\Rmc\Delta\Dmc$ so $\Rmc\models\lambda$. 
It follows that $\Rmc\models\conf$, so $\Rmc\not\models\constraints$: a contradiction.
\item In case (ii), let $\conf'=\Hmc'\cap\Dmc\cup\{\no\alpha\mid\alpha\in\Hmc'\setminus\Dmc\}$. Since $\Hmc'\subsetneq\Hmc$, then $\conf'\subsetneq\conf$. By definition of $\conf$ there exists a database $\Imc$ such that $\Imc\models \conf'$ and $\Imc\models\constraints$. Hence there exists $\Rmc\in\deltareps{\Dmc,\constraints}$ such that $\Rmc\Delta\Dmc\subseteq\Imc\Delta\Dmc$. 
Since $\Hmc'$ is a hitting set of $\{\Rmc\Delta\Dmc\mid \Rmc \in\deltareps{\Dmc,\constraints}\}$, there exists $\alpha\in\Hmc'\cap\Rmc\Delta\Dmc$, thus $\alpha\in\Hmc'\cap\Imc\Delta\Dmc$. 
Hence there exists $\lambda\in\conf'$ such that either $\lambda=\alpha$ and $\alpha\in\Dmc\setminus\Imc$ or $\lambda=\no\alpha$ and $\alpha\in\Imc\setminus\Dmc$. In both cases, $\Imc\not\models \lambda$: a contradiction.
\end{itemize}
In the other direction,  let $\conf\in \{ \Hmc\cap\Dmc\cup\{\no\alpha \mid \alpha\in\Hmc\setminus\Dmc\} \mid \Hmc\in \mhs{\Dmc,\constraints} \}$ 
and assume for a contradiction that $\conf\notin\conflicts{\Dmc,\constraints}$. This means that either (i) there exists $\Imc$ such that $\Imc\models \conf$ and $\Imc\models\constraints$ or (ii) there exists $\conf'\subsetneq\conf$ such that $\Imc\models \conf'$ implies $\Imc\not\models\constraints$.
\begin{itemize}
\item In case (i), since $\Imc\models\constraints$, there exists $\Rmc\in\deltareps{\Dmc,\constraints}$ such that $\Rmc\Delta\Dmc\subseteq\Imc\Delta\Dmc$. Hence, by definition of $\conf$, there exists $\alpha\in\Rmc\Delta\Dmc$ such that $\lambda\in\conf$ where $\lambda=\alpha$ if $\alpha\in\Dmc$ and $\lambda=\neg\alpha$ if $\alpha\notin\Dmc$. 
Since $\Rmc\Delta\Dmc\subseteq\Imc\Delta\Dmc$, then $\alpha\in\Imc\Delta\Dmc$, so $\lambda=\alpha$ if $\alpha\in\Dmc\setminus\Imc$ and $\lambda=\neg\alpha$ if $\alpha\in\Imc\setminus\Dmc$. 
It follows that $\Imc\not\models \lambda$, thus $\Imc\not\models\conf$: a contradiction. 
\item In case (ii), let $\Hmc=\{\alpha\mid \alpha\text{ or }\neg\alpha\in\conf\}$. It is easy to check that $\Hmc$ is the minimal hitting set of $\{\Rmc\Delta\Dmc\mid \Rmc \in\deltareps{\Dmc,\constraints}\}$ that corresponds to $\conf$. Let $\Hmc'=\{\alpha\mid \alpha\text{ or }\neg\alpha\in\conf'\}$. Since $\conf'\subsetneq\conf$, then $\Hmc'\subsetneq\Hmc$.  
For every $\Rmc\in\deltareps{\Dmc,\constraints}$, since $\Rmc\models \constraints$, then $\Rmc\not\models\conf'$ so there exists $\lambda\in\conf'$ such that $\Rmc\not\models\lambda$. If $\lambda=\alpha$, since $\alpha\in\Dmc$ by construction of $\conf$, it follows that $\alpha\in\Dmc\setminus\Rmc$. If $\lambda=\neg\alpha$, since $\alpha\notin\Dmc$ by construction of $\conf$, it follows that $\alpha\in\Rmc\setminus\Dmc$. Hence in both cases there exists $\alpha\in\Hmc'$ such that $\alpha\in\Rmc\Delta\Dmc$. It follows that $\Hmc'$ is a hitting set of $\{\Rmc\Delta\Dmc\mid \Rmc \in\deltareps{\Dmc,\constraints}\}$: a contradiction.
\end{itemize}

\noindent $(2)$ Recall that $\bigwedge_{\lambda\in\conf}\lambda$ is a prime implicant of a propositional formula $\psi$ if it is a minimal conjunction of propositional literals that entails $\psi$, \ie for every database $\Imc$, $\Imc\models\conf$ implies $\Imc\models\psi$ and for every $\conf'\subsetneq\conf$, there exists a database $\Imc$ such that $\Imc\models\conf'$ and $\Imc\not\models\psi$. 
Moreover, note that $\Imc\models\constraints$ iff $\Imc\models gr_\Imc(\constraints)$ iff $\Imc\not\models \bigvee_{\varphi\rightarrow \bot\in gr_\Imc(\constraints)} \varphi$ and that for every $\Imc\subseteq\factset$, $gr_\Imc(\constraints)\subseteq gr_\Dmc(\constraints)$ since the domain of $\Imc$ is included in the domain of $\Dmc$. 

Let $\conf\in\conflicts{\Dmc,\constraints}$. By definition of $\conf$, it holds that $\conf\subseteq\litset$. 
\begin{itemize}
\item Let $\Imc$ be a database such that $\Imc\models\conf$ and let $\Imc'=\Imc\cap\factset$. $\Imc'\models\conf$ so $\Imc'\not\models\constraints$. Since $gr_{\Imc'}(\constraints)\subseteq gr_\Dmc(\constraints)$ it follows that $\Imc'\models \bigvee_{\varphi\rightarrow \bot\in gr_\Dmc(\constraints)} \varphi$. Since facts that are not in $\factset$ are irrelevant to the satisfaction of $\bigvee_{\varphi\rightarrow \bot\in gr_\Dmc(\constraints)} \varphi$, we obtain that $\Imc\models \bigvee_{\varphi\rightarrow \bot\in gr_\Dmc(\constraints)} \varphi$.
\item For every $\conf'\subsetneq\conf$, there exists $\Imc$ such that $\Imc\models\conf'$ and $\Imc\models\constraints$, so that $\Imc\not\models \bigvee_{\varphi\rightarrow \bot\in gr_\Dmc(\constraints)} \varphi $.
\end{itemize}
Hence  $\bigwedge_{\lambda\in\conf}\lambda$ is a prime implicant of $\bigvee_{\varphi\rightarrow \bot\in gr_\Dmc(\constraints)} \varphi $.

\noindent In the other direction, let $\conf\subseteq \litset$ be such that $\bigwedge_{\lambda\in\conf}\lambda$ is a prime implicant of $\bigvee_{\varphi\rightarrow \bot\in gr_\Dmc(\constraints)} \varphi $. 
\begin{itemize}
\item For every database $\Imc$, if $\Imc\models\conf$, then $\Imc\models \bigvee_{\varphi\rightarrow \bot\in gr_\Dmc(\constraints)} \varphi $ so $\Imc\not\models\constraints$. 

\item For every $\conf'\subsetneq\conf$, there exists $\Imc$ such that $\Imc\models\conf'$ and $\Imc\not\models \bigvee_{\varphi\rightarrow \bot\in gr_\Dmc(\constraints)} \varphi$. Let $\Imc'=\Imc\cap\factset$. Since $\conf'\subseteq\litset$ and facts that are not in $\factset$ are irrelevant to the satisfaction of $\bigvee_{\varphi\rightarrow \bot\in gr_\Dmc(\constraints)} \varphi$, $\Imc'\models \conf'$ and $\Imc'\not\models \bigvee_{\varphi\rightarrow \bot\in gr_\Dmc(\constraints)} \varphi$. Since $gr_{\Imc'}(\constraints)\subseteq gr_\Dmc(\constraints)$, it follows that $\Imc'\models\constraints$. 
\end{itemize}
Hence $\conf\in\conflicts{\Dmc,\constraints}$. 
\end{proof}

The next few basic lemmas allow us to move between databases and sets of literals using functions $\comp{\Dmc}{}$ and $\restr{\Dmc}{}$. 
\begin{align*}
\text{Recall that  }\comp{\Dmc}{\Rmc}=&\mi{Lits}_\Dmc(\Rmc)\cap \litset=(\Rmc\cap\Dmc)\cup\{\no\alpha\mid \alpha\in\factset\setminus(\Rmc\cup\Dmc)\}\\
\text{and  }\restr{\Dmc}{\Bmc}=&\Bmc\cap\Dmc\cup\{\alpha\mid\no\alpha\in \litset\setminus\Bmc\}.
\end{align*}

\begin{lemma}\label{lem:compl-restr}
The following assertions hold:
\begin{enumerate}
\item If $\Rmc$ is a candidate repair for $\Dmc$, \ie $\Rmc\subseteq \factset$, then $\restr{\Dmc}{\comp{\Dmc}{\Rmc}}=\Rmc$. 

\item If $\Bmc\subseteq\litset$, then $\comp{\Dmc}{\restr{\Dmc}{\Bmc}}=\Bmc$.
\end{enumerate}
\end{lemma}
\begin{proof}
\noindent$(1)$ Let $\Rmc\subseteq \factset$. 
 \begin{align*}
 \restr{\Dmc}{\comp{\Dmc}{\Rmc}}=&(\comp{\Dmc}{\Rmc}\cap\Dmc)\cup\{\alpha\mid\no\alpha\in\litset\setminus\comp{\Dmc}{\Rmc}\}
 \\
 =&(\Rmc\cap\Dmc)\cup\{\alpha\mid\no\alpha\in(\litset\setminus\{\no\alpha\mid \alpha\in\factset\setminus(\Rmc\cup\Dmc)\})\}
 \\
  =&(\Rmc\cap\Dmc)\cup\{\alpha\mid\no\alpha\in\litset,\alpha\in\Rmc\cup\Dmc\}\\
    =&(\Rmc\cap\Dmc)\cup\{\alpha\mid\no\alpha\in\litset,\alpha\in\Rmc\setminus\Dmc\}\text{ since }\alpha\in\Dmc\text{ implies }\no\alpha\notin\litset\\
        =&(\Rmc\cap\Dmc)\cup(\Rmc\setminus\Dmc)\text{ since }\alpha\in\Rmc\setminus\Dmc\text{ implies }\no\alpha\in\litset\text{ because }\Rmc\subseteq \factset\\
 =&\Rmc
 \end{align*}
 
 \noindent$(2)$ Let $\Bmc\subseteq\litset$.
  \begin{align*}
 \comp{\Dmc}{\restr{\Dmc}{\Bmc}}
 =&(\restr{\Dmc}{\Bmc}\cap\Dmc)\cup\{\no\alpha\mid \alpha\in\factset\setminus(\restr{\Dmc}{\Bmc}\cup\Dmc)\}\\
 =&(\Bmc\cap\Dmc)\cup\{\no\alpha\mid \no\alpha\in\litset,\alpha\notin \restr{\Dmc}{\Bmc})\}
 \\
  =&(\Bmc\cap\Dmc)\cup\{\no\alpha\mid \no\alpha\in\Bmc\}
  \\
 =&\Bmc\text{ since all positive literals of }\Bmc\text{ are also in }\Dmc\text{ as }\Bmc\subseteq\litset\qedhere
 \end{align*}
 
\end{proof}

\begin{lemma}\label{lem:compl-delta-1}
If $\Rmc_2$ is a candidate repair and $\Rmc_1\Delta\Dmc\subseteq\Rmc_2\Delta\Dmc$ then $\comp{\Dmc}{\Rmc_2}\subseteq\comp{\Dmc}{\Rmc_1}$. If the first inclusion is strict, so is the second. 
\end{lemma}
\begin{proof}
Let $\lambda\in\comp{\Dmc}{\Rmc_2}$. If $\lambda=\alpha\in\Dmc$, then $\alpha\in\Rmc_2$. Thus since $\Rmc_1\Delta\Dmc\subseteq\Rmc_2\Delta\Dmc$, then $\alpha\in\Rmc_1$. If $\lambda=\no\alpha$ for $\alpha\notin\Dmc$, then $\alpha\notin\Rmc_2$. Thus since $\Rmc_1\Delta\Dmc\subseteq\Rmc_2\Delta\Dmc$, then $\alpha\notin\Rmc_1$. In both cases, $\lambda\in\comp{\Dmc}{\Rmc_1}$. Hence $\comp{\Dmc}{\Rmc_2}\subseteq\comp{\Dmc}{\Rmc_1}$. 

If the inclusion is strict, $\Rmc_1\Delta\Dmc\subsetneq\Rmc_2\Delta\Dmc$, there is $\alpha\in (\Rmc_2\Delta\Dmc)\setminus(\Rmc_1\Delta\Dmc)$. If $\alpha\in\Dmc$ it is clear that $\alpha\in \comp{\Dmc}{\Rmc_1}$ while $\alpha\notin \comp{\Dmc}{\Rmc_2}$. If $\alpha\notin\Dmc$, since $\alpha\in\Rmc_2$ and $\Rmc_2$ is a candidate repair, then $\alpha\in\factset$ so $\no\alpha\in\litset$ so $\no\alpha\in \comp{\Dmc}{\Rmc_1}$ while $\no\alpha\notin \comp{\Dmc}{\Rmc_2}$ and $\comp{\Dmc}{\Rmc_2}\subsetneq\comp{\Dmc}{\Rmc_1}$.  
\end{proof}

\begin{lemma}\label{lem:compl-delta-2}
If $\Bmc_1\subseteq\Bmc_2\subseteq\litset$ then $\restr{\Dmc}{\Bmc_2}\Delta\Dmc\subseteq\restr{\Dmc}{\Bmc_1}\Delta\Dmc$. If the first inclusion is strict, so is the second. 
\end{lemma}
\begin{proof}
Let $\alpha\in\restr{\Dmc}{\Bmc_2}\Delta\Dmc$. If $\alpha\in\Dmc$, then $\alpha\notin\restr{\Dmc}{\Bmc_2}$ so $\alpha\notin\Bmc_2$ which implies that $\alpha\notin\Bmc_1$ and $\alpha\notin\restr{\Dmc}{\Bmc_1}$. Hence $\alpha\in\Dmc\setminus\restr{\Dmc}{\Bmc_1}$. 
If $\alpha\notin\Dmc$, then $\alpha\in\restr{\Dmc}{\Bmc_2}$ so $\no\alpha\notin\Bmc_2$ which implies that $\no\alpha\notin\Bmc_1$ and $\alpha\in\restr{\Dmc}{\Bmc_1}$. Hence $\alpha\in\restr{\Dmc}{\Bmc_1}\setminus\Dmc$. 
In both cases $\alpha\in\restr{\Dmc}{\Bmc_1}\Delta\Dmc$ so $\restr{\Dmc}{\Bmc_2}\Delta\Dmc\subseteq\restr{\Dmc}{\Bmc_1}\Delta\Dmc$. 

If inclusion is strict, $\Bmc_1\subsetneq\Bmc_2$, there exists $\lambda\in\Bmc_2\setminus\Bmc_1$: If $\lambda=\alpha\in\Dmc$, then $\alpha\in\restr{\Dmc}{\Bmc_2}$ while $\alpha\notin \restr{\Dmc}{\Bmc_1}$; and if $\lambda=\no\alpha$ with $\alpha\notin\Dmc$, then $\alpha\notin\restr{\Dmc}{\Bmc_2}$ while $\alpha\in \restr{\Dmc}{\Bmc_1}$. In both cases, $\alpha\in \restr{\Dmc}{\Bmc_1}\Delta\Dmc$ while $\alpha\notin\restr{\Dmc}{\Bmc_2}\Delta\Dmc$ so $\restr{\Dmc}{\Bmc_2}\Delta\Dmc\subsetneq\restr{\Dmc}{\Bmc_1}\Delta\Dmc$. 
\end{proof}

\begin{lemma}\label{lem:max-no-conf}
If $\Bmc$ is a maximal subset of $\litset$ such that there is no conflict $\conf\in\deltaconflicts{\Dmc,\constraints}$ such that $\conf\subseteq\Bmc$, then $\restr{\Dmc}{\Bmc}\in\deltareps{\Dmc,\constraints}$.
 \end{lemma}
 \begin{proof}
Assume that $\Bmc$ is a maximal subset of $\litset$ such that there is no conflict $\conf\in\deltaconflicts{\Dmc,\constraints}$ such that $\conf\subseteq\Bmc$. 
\begin{itemize}
\item For every $\conf\in\deltaconflicts{\Dmc,\constraints}$, $\conf\subseteq\litset$ and  $\conf\not\subseteq\Bmc$ so $\conf\cap(\litset\setminus\Bmc)\neq \emptyset$. Thus $\litset\setminus\Bmc$ is a hitting set of $\deltaconflicts{\Dmc,\constraints}$. 

Let $\Hmc\in \mhs{\Dmc,\constraints}$ and let $\conf=\Hmc\cap\Dmc\cup\{\no\alpha\mid\alpha\in\Hmc\setminus\Dmc\}$ be the corresponding conflict. Since $\litset\setminus\Bmc$ is a hitting set of $\deltaconflicts{\Dmc,\constraints}$, there is $\lambda \in(\litset\setminus\Bmc)\cap\conf$. 
If $\lambda=\alpha\in\Dmc$, $\alpha\in\Hmc$ and $\alpha\in\Dmc\setminus\restr{\Dmc}{\Bmc}$. 
If $\lambda=\no\alpha$ for some $\alpha\notin\Dmc$, $\alpha\in\Hmc$ and $\alpha\in\restr{\Dmc}{\Bmc}\setminus\Dmc$. 
In both cases, $\alpha\in\Hmc\cap(\restr{\Dmc}{\Bmc}\Delta\Dmc)$. 

It follows that $\restr{\Dmc}{\Bmc}\Delta\Dmc$ is a hitting set of $\mhs{\Dmc,\constraints}$. 

\item Assume for a contradiction that there is no $\Rmc\in\deltareps{\Dmc}$ such that $\Rmc\Delta\Dmc\subseteq \restr{\Dmc}{\Bmc}\Delta\Dmc$. Let $\Mmc=\bigcup_{\Rmc\in \deltareps{\Dmc}} (\Rmc\Delta\Dmc)\setminus (\restr{\Dmc}{\Bmc}\Delta\Dmc)$. 

Since for every $\Rmc\in \deltareps{\Dmc}$, $\Rmc\Delta\Dmc\not\subseteq \restr{\Dmc}{\Bmc}\Delta\Dmc$, so that $(\Rmc\Delta\Dmc)\setminus (\restr{\Dmc}{\Bmc}\Delta\Dmc)\neq\emptyset$, it follows that $\Mmc$ is a hitting set of $\{\Rmc\Delta\Dmc\mid\Rmc\in \deltareps{\Dmc}\}$. Hence there is some $\Mmc'\in \mhs{\Dmc,\constraints}$ such that $\Mmc'\subseteq\Mmc$. 

Since $\restr{\Dmc}{\Bmc}\Delta\Dmc$ is a hitting set of $\mhs{\Dmc,\constraints}$, then $(\restr{\Dmc}{\Bmc}\Delta\Dmc)\cap\Mmc'\neq\emptyset$. Hence $(\restr{\Dmc}{\Bmc}\Delta\Dmc)\cap\Mmc\neq\emptyset$. However, by definition of $\Mmc$, $(\restr{\Dmc}{\Bmc}\Delta\Dmc)\cap\Mmc=\emptyset$: contradiction. 

It follows that there exists  $\Rmc\in\deltareps{\Dmc}$ such that $\Rmc\Delta\Dmc\subseteq \restr{\Dmc}{\Bmc}\Delta\Dmc$. 

\item Assume for a contradiction that $\restr{\Dmc}{\Bmc}\neq\Rmc$, \ie $\Rmc\Delta\Dmc\subsetneq \restr{\Dmc}{\Bmc}\Delta\Dmc$. By Lemma \ref{lem:compl-delta-1}, $\comp{\Dmc}{\restr{\Dmc}{\Bmc}}\subsetneq \comp{\Dmc}{\Rmc}$, so by Lemma \ref{lem:compl-restr}, $\Bmc\subsetneq \comp{\Dmc}{\Rmc}$. Hence, by assumption on $\Bmc$, there must be some $\conf\in\deltaconflicts{\Dmc,\constraints}$ such that $\conf\subseteq \comp{\Dmc}{\Rmc}$, hence $\Rmc\models\conf$. It follows that $\Rmc\not\models\constraints$: contradiction. 
Hence $\restr{\Dmc}{\Bmc}=\Rmc$ so $\restr{\Dmc}{\Bmc}\in\deltareps{\Dmc,\constraints}$. \qedhere
\end{itemize}
 \end{proof}

\CharacterizationsRepairsComp*
\begin{proof}
Let $\Rmc$ be a candidate repair for $\Dmc$. By definition, $\comp{\Dmc}{\Rmc}=\mi{Lits}_\Dmc(\Rmc)\cap \litset$ is a subset of $\litset$. 
\smallskip

\noindent$(1)$ 
First note that by Lemma \ref{lem:compl-restr}, $\restr{\Dmc}{\comp{\Dmc}{\Rmc}}\models \constraints$ is indeed equivalent to $\Rmc\models \constraints$. 
 
Assume that $\Rmc\in\deltareps{\Dmc,\constraints}$. By definition $\Rmc\models\constraints$. 
Assume for a contradiction that there exists $\Bmc$ such that $\comp{\Dmc}{\Rmc}\subsetneq\Bmc\subseteq\litset$ and $\restr{\Dmc}{\Bmc}\models\constraints$. 
By Lemma \ref{lem:compl-delta-2}, $\restr{\Dmc}{\Bmc}\Delta\Dmc\subsetneq\restr{\Dmc}{\comp{\Dmc}{\Rmc}}\Delta\Dmc$, \ie 
by Lemma \ref{lem:compl-restr}, $\restr{\Dmc}{\Bmc}\Delta\Dmc\subsetneq\Rmc\Delta\Dmc$. Since $\restr{\Dmc}{\Bmc}\models\constraints$, this contradicts $\Rmc\in\deltareps{\Dmc,\constraints}$. 

In the other direction, assume that $\comp{\Dmc}{\Rmc}$ is a maximal subset of $\litset$ such that $\restr{\Dmc}{\comp{\Dmc}{\Rmc}}\models \constraints$ and assume for a contradiction that $\Rmc\notin \deltareps{\Dmc,\constraints}$. Since $\Rmc\models\constraints$, this means that there exists $\Rmc'$ such that $\Rmc'\Delta\Dmc\subsetneq\Rmc\Delta\Dmc$ and $\Rmc'\models\constraints$. 
Since $\Rmc$ is a candidate repair for $\Dmc$, by Lemma \ref{lem:compl-delta-1}, 
 $\comp{\Dmc}{\Rmc}\subsetneq\comp{\Dmc}{\Rmc'}\subseteq\litset$. 
Since by Lemma \ref{lem:compl-restr} $\restr{\Dmc}{\comp{\Dmc}{\Rmc'}}=\Rmc'$ and $\Rmc'\models \constraints$, this contradicts our assumption on $\comp{\Dmc}{\Rmc}$.
\smallskip

\noindent$(2)$ 
Assume that $\Rmc\in\deltareps{\Dmc,\constraints}$. Since $\Rmc\models\constraints$, then for every conflict $\conf\in\conflicts{\Dmc,\constraints}$, $\Rmc\not\models\conf$, so that $\conf\not\subseteq\comp{\Dmc}{\Rmc}$. 
Let $\Bmc$ be a maximal subset of $\litset$ such that $\comp{\Dmc}{\Rmc}\subseteq\Bmc$ and for every conflict $\conf\in\conflicts{\Dmc,\constraints}$, $\conf\not\subseteq\Bmc$. By Lemma \ref{lem:max-no-conf}, $\restr{\Dmc}{\Bmc}\in\deltareps{\Dmc,\constraints}$ and 
by Lemmas \ref{lem:compl-delta-2} and \ref{lem:compl-restr}, $\restr{\Dmc}{\Bmc}\Delta\Dmc\subseteq\Rmc\Delta\Dmc$. Since both $\restr{\Dmc}{\Bmc}$ and $\Rmc$ are in $\deltareps{\Dmc,\constraints}$, it follows that $\Rmc=\restr{\Dmc}{\Bmc}$. Hence $\comp{\Dmc}{\Rmc}=\Bmc$ (by Lemma \ref{lem:compl-restr}) is a maximal subset of $\litset$ such that for every conflict $\conf\in\conflicts{\Dmc,\constraints}$, $\conf\not\subseteq\Bmc$. 

In the other direction, assume that $\comp{\Dmc}{\Rmc}$ is a maximal subset of $\litset$ such that for every conflict $\conf\in\conflicts{\Dmc,\constraints}$, $\conf\not\subseteq\Bmc$. 
By Lemma \ref{lem:max-no-conf}, $\restr{\Dmc}{\comp{\Dmc}{\Rmc}}\in\deltareps{\Dmc,\constraints}$ so by Lemma \ref{lem:compl-restr}, $\Rmc\in\deltareps{\Dmc,\constraints}$.
\smallskip

\noindent$(3)$ 
We know from point (2) that $\Rmc\in\deltareps{\Dmc,\constraints}$ iff $\comp{\Dmc}{\Rmc}$ is a maximal subset of $\litset$ such that there is no conflict $\conf\in\conflicts{\Dmc,\constraints}$ such that $\conf\subseteq\comp{\Dmc}{\Rmc}$. It follows that $\Rmc\in\deltareps{\Dmc,\constraints}$ iff $\comp{\Dmc}{\Rmc}$ is a maximal independent set of the hypergraph $\confgraph{\Dmc}{\constraints}$ with vertices $\litset$ and hyperedges conflicts in $\conflicts{\Dmc,\constraints}$. 
\end{proof}

\begin{remark}
Each maximal independent set $\Mmc$ of $\confgraph{\Dmc}{\constraints}$ corresponds to a maximal independent set $\Mmc'$ of the hypergraph $\Gmc\subseteq \confgraph{\Dmc}{\constraints}$ whose vertices are $\bigcup_{\conf\in \conflicts{\Dmc,\constraints}} \conf$ and edges are the same as $\confgraph{\Dmc}{\constraints}$, as follows: $\Mmc=\Mmc'\cup (\litset\setminus \bigcup_{\conf\in \conflicts{\Dmc,\constraints}} \conf)$. Hence, we can use the maximal independent sets of the ``compact version'' $\Gmc$ of the conflict hypergraph to obtain repairs. Note that the original formulation of point (3) of Proposition 2 in the published KR'23 paper used this compact graph instead, but the formulation did not properly account for the literals in $\litset\setminus \bigcup_{\conf\in \conflicts{\Dmc,\constraints}} \conf$. The updated formulation, given in the present long version, corrects this omission and provides  an arguably simpler and more intuitive characterization by defining the hypergraph over $\litset$ rather than over $\bigcup_{\conf\in \conflicts{\Dmc,\constraints}} \conf$. \end{remark}

\paragraph{Reduction to denial constraints} 
Recall that $\preds'=\preds \cup \{\widetilde{P} \mid P \in \preds\}$, $\tofacts$ maps sets of literals over $\preds$ into sets of facts over $\preds'$ by replacing each negative literal $\neg P(\vec{c})$ by $\widetilde{P}(\vec{c})$, and that given a database $\Dmc$ and set of universal constraints $\constraints$ over schema $\preds$, we define $\Dmc_d=\tofacts(\litset)=\Dmc\cup\{\widetilde{P}(\vec{c})\mid P(\vec{c})\in\factset\setminus\Dmc\}$,
and 
$\constraints_{d,\Dmc}=\{ (\bigwedge_{\alpha \in \tofacts(\conf)} \alpha) \rightarrow \bot \mid \conf\in\conflicts{\Dmc,\constraints}\}$.

\ReductionUCDenials*
\begin{proof}
\begin{itemize}
\item Let $\conf\in\conflicts{\Dmc,\constraints}$. 
Clearly, $\tofacts(\conf)\not\models (\bigwedge_{\alpha \in \tofacts(\conf)} \alpha) \rightarrow \bot$ so $\tofacts(\conf)\not\models\constraints_{d,\Dmc}$. 
Moreover, there is no proper subset $\Bmc$ of $\tofacts(\conf)$ such that $\Bmc\not\models \constraints_{d,\Dmc}$: Otherwise, $\Bmc$ would violate some $((\bigwedge_{\alpha \in \tofacts(\conf')} \alpha) \rightarrow \bot)\in \constraints_{d,\Dmc}$ with $\tofacts(\conf')\subseteq\Bmc\subsetneq\tofacts(\conf)$, which implies that there is $\conf'\in\conflicts{\Dmc,\constraints}$ with $\conf'\subsetneq\conf$ (which is not possible since $\conf$ and $\conf'$ should both be minimal subsets of $\litset$ such that for every $\Imc$, $\Imc\models\conf^{(')}$ implies $\Imc\not\models\constraints$). 
Hence $\tofacts(\conf)\in\conflicts{\Dmc_d,\constraints_{d,\Dmc}}$. 

In the other direction, let $\Bmc\in \conflicts{\Dmc_d,\constraints_{d,\Dmc}}$: $\Bmc$ is a minimal subset of $\Dmc_d$ inconsistent with $\constraints_{d,\Dmc}$. Since $\Bmc\not\models \constraints_{d,\Dmc}$, there exists $\tau_\conf:=((\bigwedge_{\alpha \in \tofacts(\conf)} \alpha) \rightarrow \bot)$ in $\constraints_{d,\Dmc}$ (corresponding to $\conf\in\conflicts{\Dmc,\constraints}$) such that $\Bmc\not\models\tau_{\conf}$, \ie $\tofacts(\conf)\subseteq\Bmc$. 
Since $\tofacts(\conf)\not\models\tau_{\conf}$ so that $\tofacts(\conf)\not\models\constraints_{d,\Dmc}$, by minimality of $\Bmc$, it must be the case that $\tofacts(\conf)=\Bmc$.

Hence $\conflicts{\Dmc_d,\constraints_{d,\Dmc}}=\{ \tofacts(\conf)\mid \conf\in\conflicts{\Dmc,\constraints}\}$.

\item Let $\Rmc\in\deltareps{\Dmc,\constraints}$. By Proposition~\ref{prop:characterizations-repairs-comp}, $\comp{\Dmc}{\Rmc}$ is a maximal subset of $\litset$ such that there is no conflict $\conf\in\conflicts{\Dmc,\constraints}$ such that $\conf\subseteq\comp{\Dmc}{\Rmc}$. 
Hence, $\tofacts(\comp{\Dmc}{\Rmc})$ is a maximal subset of $\tofacts(\litset)$ such that there is no conflict $\Bmc\in \conflicts{\Dmc_d,\constraints_{d,\Dmc}}$ such that $\Bmc\subseteq \tofacts(\comp{\Dmc}{\Rmc})$. Since $\tofacts(\litset)=\Dmc_d$, it follows that $\tofacts(\comp{\Dmc}{\Rmc})\in \deltareps{\Dmc_d,\constraints_{d,\Dmc}}$.

In the other direction, let $\Bmc\in \deltareps{\Dmc_d,\constraints_{d,\Dmc}}$: $\Bmc$ is a maximal subset of $\Dmc_d$ that does not contain any conflict from $\conflicts{\Dmc_d,\constraints_{d,\Dmc}}$. Thus $\Bmc$ is a maximal subset of $\tofacts(\litset)$
such that there is no conflict $\conf\in \conflicts{\Dmc,\constraints}$ such that $\tofacts(\conf)\subseteq \Bmc$. 
Let $\Rmc=\Bmc\cap\Dmc\cup\{\neg P(\vec{a})\mid\widetilde{P}(\vec{a})\in\Bmc\}$ be the subset of $\litset$ such that $\Bmc=\tofacts(\comp{\Dmc}{\Rmc})$: $\comp{\Dmc}{\Rmc}$ is a maximal subset of $\litset$
such that there is no conflict $\conf\in \conflicts{\Dmc,\constraints}$ such that $\conf\subseteq \comp{\Dmc}{\Rmc}$. 
By Proposition~\ref{prop:characterizations-repairs-comp}, it follows that $\Rmc\in \deltareps{\Dmc,\constraints}$.

Hence $\deltareps{\Dmc_d,\constraints_{d,\Dmc}}=\{\tofacts(\comp{\Dmc}{\Rmc})\mid \Rmc\in\deltareps{\Dmc,\constraints}\}$. 
\qedhere
\end{itemize}
\end{proof}

\subsection{Proofs for Section \ref{subsec:optimalrepairs}}

The following lemma will sometimes be used to show that a $\Delta$-repair is not Pareto-optimal.
\begin{lemma}\label{lem:pareto-improvement-if-one-beat-all}
Let $\Dmc^\constraints_\succ$ be a prioritized database and $\Rmc$ be a database. 
If there exists $\lambda$ such that for every $\conf\in\conflicts{\Dmc,\constraints}$, 
$\conf\not\subseteq\comp{\Dmc}{\Rmc}\cup\{\lambda\}\setminus\{\mu\mid\lambda\succ\mu\}$, then 
$\Rmc\notin\deltapreps{\Dmc^\constraints_\succ}$. 
\end{lemma}
\begin{proof}
Let  $\Rmc_\lambda=\restr{\Dmc}{\comp{\Dmc}{\Rmc}\cup\{\lambda\}\setminus\{\mu\mid\lambda\succ\mu\}}$, so that $\comp{\Dmc}{\Rmc_\lambda}=\comp{\Dmc}{\Rmc}\cup\{\lambda\}\setminus\{\mu\mid\lambda\succ\mu\}$ by Lemma~\ref{lem:compl-restr}. 
Let $\Bmc$ be a maximal subset of $\litset$ such that for every $\conf\in\conflicts{\Dmc,\constraints}$, 
$\conf\not\subseteq\comp{\Dmc}{\Rmc_\lambda}\cup\Bmc$, and let $\Rmc'=\restr{\Dmc}{\comp{\Dmc}{\Rmc_\lambda}\cup\Bmc}$. 

By Lemma~\ref{lem:compl-restr}, $\comp{\Dmc}{\Rmc'}=\comp{\Dmc}{\Rmc_\lambda}\cup\Bmc$. Hence $\Rmc'$ is such that $\comp{\Dmc}{\Rmc'}$ is a maximal subset of $\litset$ such that there is no $\conf\in\deltaconflicts{\Dmc,\constraints}$ such that $\conf\subseteq\comp{\Dmc}{\Rmc'}$. By Proposition~\ref{prop:characterizations-repairs-comp}, $\Rmc'\in\deltareps{\Dmc,\constraints}$ hence $\Rmc'\models\constraints$. 

Moreover, for every $\mu\in\comp{\Dmc}{\Rmc}\setminus\comp{\Dmc}{\Rmc'}$, $\mu\in\comp{\Dmc}{\Rmc}\setminus\comp{\Dmc}{\Rmc_\lambda}\subseteq\{\mu\mid\lambda\succ\mu\}$ so $\lambda\succ\mu$. Thus $\Rmc'$ is a Pareto-improvement of $\Rmc$ and $\Rmc\notin\deltapreps{\Dmc^\constraints_\succ}$.
\end{proof}

The following lemma gives the result we mentioned about the existence of completion-(hence also globally- and Pareto-) optimal $\Delta$-repairs.
\begin{lemma}\label{existenceCompletionOpti}
For every set of universal constraints $\constraints$, database $\Dmc$ and priority relation $\succ$, $\deltacreps{\Dmc^\constraints_\succ}\neq\emptyset$.  
\end{lemma}
\begin{proof}
Let $\Bmc$ be a set of literals from $\litset$ obtained from $\confgraph{\Dmc}{\constraints}$ by the following greedy procedure: while some literal from $\litset$ has not been considered, pick a literal 
that is maximal \wrt $\succ$ among those not yet considered, and add it to $\Bmc$ if it does not introduce a conflict from $\conflicts{\Dmc,\constraints}$. 
We show that $\Rmc=\restr{\Dmc}{\Bmc}$ belongs to $\deltacreps{\Dmc^\constraints_\succ}$. 

By Lemma \ref{lem:compl-restr}, $\comp{\Dmc}{\Rmc}=\comp{\Dmc}{\restr{\Dmc}{\Bmc}}=\Bmc$, so since 
$\Bmc$ is a maximal subset of $\litset$ that does not contain any conflict from $\conflicts{\Dmc,\constraints}$, by Proposition \ref{prop:characterizations-repairs-comp}, $\Rmc\in\deltareps{\Dmc,\constraints}$. 
Assume for a contradiction that $\Rmc\notin\deltacreps{\Dmc^\constraints_\succ}$. 
\begin{itemize}
\item Let $\succ'$ be the binary relation defined as follows: for every $\lambda,\mu$ such that $\{\lambda,\mu\}\subseteq\conf\in\conflicts{\Dmc,\constraints}$, $\lambda\succ'\mu$ iff $\lambda$ has been picked before $\mu$ by the greedy procedure. 
Since the procedure examines all literals of $\litset$ exactly once, $\succ'$ is a total priority relation. 
Moreover, if $\lambda\succ\mu$, $\lambda$ is picked before $\mu$ so $\lambda\succ'\mu$. Hence $\succ'$ is a completion of $\succ$.

\item Since $\Rmc\notin\deltacreps{\Dmc^\constraints_\succ}$, then $\Rmc\notin\deltagreps{\Dmc^\constraints_{\succ'}}$: there exists a database $\Rmc'$ consistent \wrt $\constraints$ such that $\comp{\Dmc}{\Rmc}\neq \comp{\Dmc}{\Rmc'}$ and for every $\lambda\in \comp{\Dmc}{\Rmc}\setminus \comp{\Dmc}{\Rmc'}$, there exists $\mu\in \comp{\Dmc}{\Rmc'}\setminus \comp{\Dmc}{\Rmc}$ such that $\mu\succ'\lambda$.

\item Let $\Bmc'=\comp{\Dmc}{\Rmc'}$. It holds that
\begin{itemize}
\item[(i)] for every $\lambda\in \Bmc\setminus \Bmc'$, there exists $\mu\in \Bmc'\setminus \Bmc$ such that $\mu\succ'\lambda$;

\item[(ii)]  there does not exist any $\conf\in \conflicts{\Dmc,\constraints}$ such that $\conf\subseteq\Bmc'$ (otherwise, $\Rmc'\models\conf$ and $\Rmc'\not\models\constraints$). 
\end{itemize}
\item We show by induction that we can build an infinite chain $\mu_1\prec'\lambda_1\prec'\mu_2\prec'\lambda_2\prec'\dots$ such that for every $i$, $\mu_i\in\Bmc'\setminus\Bmc$ and $\lambda_i\in\Bmc\setminus\Bmc'$. Since $\succ'$ is acyclic, all $\mu_i,\lambda_i$ must be distinct, contradicting the fact that $\litset$ is finite.
\begin{itemize}
\item Base case: 
Let $\mu_1\in \Bmc'\setminus \Bmc$. Since $\mu_1$ has not been added to $\Bmc$, there exists $\conf\in \conflicts{\Dmc,\constraints}$ such that $\conf\setminus\{\mu_1\}\subseteq\Bmc$ and for every $\lambda\in \conf\setminus\{\mu_1\}$, $\lambda\succ'\mu_1$. 
 Since $\conf\not\subseteq\Bmc'$ (by (ii)), there exists $\lambda_1\in \conf\setminus\{\mu_1\}$ such that $\lambda_1\notin\Bmc'$. Hence $\lambda_1\in\Bmc\setminus\Bmc'$. 
 
We thus have $\mu_1\in \Bmc'\setminus \Bmc$ and $\lambda_1\in\Bmc\setminus\Bmc'$ such that $\lambda_1\succ'\mu_1$. 

\item Induction step:  Assume that we have built $\mu_1\prec'\lambda_1\prec'\dots\prec'\mu_i\prec'\lambda_i$ as required. Since $\lambda_i\in \Bmc\setminus\Bmc'$, by (i), there exists $\mu_{i+1}\in \Bmc'\setminus\Bmc$ such that $\mu_{i+1}\succ'\lambda_i$. 
Since $\mu_{i+1}\in \Bmc'\setminus \Bmc$, we obtain as in the base case $\lambda_{i+1}\in\Bmc\setminus\Bmc'$ such that $\lambda_{i+1}\succ'\mu_{i+1}$.
We thus have $\mu_1\prec'\lambda_1\prec'\dots\prec'\mu_i\prec'\lambda_i\prec'\mu_{i+1}\prec'\lambda_{i+1}$ as required.
\end{itemize}
\end{itemize}
We obtain a contradiction so $\Rmc\in\deltacreps{\Dmc^\constraints_\succ}$.
\end{proof}

\Categoricity*
\begin{proof}
We will use the following claim:
\begin{claim}\label{lem:pareto-chain}
If $\succ$ is total, $\Rmc\in\deltareps{\Dmc,\constraints}$, $\Rmc'\in\deltapreps{\Dmc^\constraints_\succ}$ and $\lambda\in\comp{\Dmc}{\Rmc}\setminus\comp{\Dmc}{\Rmc'}$, then there exists $\lambda'\in\comp{\Dmc}{\Rmc'}\setminus\comp{\Dmc}{\Rmc}$ such that $\lambda\prec\lambda'$. 
\end{claim}
\noindent\emph{Proof of the claim}
Assume for a contradiction that for every $\lambda'\in\comp{\Dmc}{\Rmc'}\setminus\comp{\Dmc}{\Rmc}$, $\lambda\not\prec\lambda'$. 
Since $\succ$ is total, this means that for every $\lambda'\in\comp{\Dmc}{\Rmc'}\setminus\comp{\Dmc}{\Rmc}$, either $\lambda\succ\lambda'$ or there is no $\conf\in\deltaconflicts{\Dmc,\constraints}$ such that $\{\lambda,\lambda'\}\subseteq\conf$. 
 Let $\conf\in\deltaconflicts{\Dmc,\constraints}$ and assume for a contradiction that $\conf\subseteq \comp{\Dmc}{\Rmc'}\cup\{\lambda\}\setminus\{\mu\mid\lambda\succ\mu\}$. 
\begin{itemize}
\item Since $\Rmc'$ is a $\Delta$-repair, $\conf\not\subseteq\comp{\Dmc}{\Rmc'}$ by Proposition \ref{prop:characterizations-repairs-comp}.  Hence $\lambda\in\conf$ and $\conf\setminus\{\lambda\}\subseteq\comp{\Dmc}{\Rmc'}$. 
\item Since $\Rmc$ is a $\Delta$-repair, $\conf\not\subseteq \comp{\Dmc}{\Rmc}$ so there is some $\lambda'\in \conf\setminus\{\lambda\}$ such that $\lambda'\in \comp{\Dmc}{\Rmc'}\setminus\comp{\Dmc}{\Rmc}$. 
\item Since $\lambda'\in \comp{\Dmc}{\Rmc'}\setminus\comp{\Dmc}{\Rmc}$ and $\{\lambda,\lambda'\}\subseteq\conf$, then $\lambda\succ\lambda'$ by assumption, so $\lambda'\notin \comp{\Dmc}{\Rmc'}\cup\{\lambda\}\setminus\{\mu\mid\lambda\succ\mu\}$. 
\end{itemize}
Hence $\conf\not\subseteq \comp{\Dmc}{\Rmc'}\cup\{\lambda\}\setminus\{\mu\mid\lambda\succ\mu\}$. This holds for every $\conf\in\deltaconflicts{\Dmc,\constraints}$ so by Lemma~\ref{lem:pareto-improvement-if-one-beat-all}, $\Rmc'\notin\deltapreps{\Dmc^\constraints_\succ}$.
\smallskip

We now show the proposition. By Lemma \ref{existenceCompletionOpti} and the fact that $\deltacreps{\Dmc^\constraints_\succ}\subseteq\deltapreps{\Dmc^\constraints_\succ}$, $|\deltapreps{\Dmc^\constraints_\succ}|\geq 1$. Assume that $\succ$ is total. To show that $|\deltapreps{\Dmc^\constraints_\succ}|=1$, assume for a contradiction that there exist $\Rmc_1\neq\Rmc_2$ in $\deltapreps{\Dmc^\constraints_\succ}$. We will build by induction an infinite chain $\lambda_1\prec\mu_1\prec\lambda_2\prec\mu_2\prec\dots$ where the $\lambda_i$ are in $\comp{\Dmc}{\Rmc_1}\setminus\comp{\Dmc}{\Rmc_2}$ and the $\mu_i$ are in $\comp{\Dmc}{\Rmc_2}\setminus\comp{\Dmc}{\Rmc_1}$. Since $\succ$ is acyclic, all $\lambda_i,\mu_i$ must be distinct, contradicting the fact that $\litset$ is finite.
\begin{itemize}
\item Base case: By Proposition \ref{prop:characterizations-repairs-comp}, for $i\in\{1,2\}$, $\comp{\Dmc}{\Rmc_i}$ is a maximal subset of $\litset$ such that there is no conflict $\conf\in\conflicts{\Dmc,\constraints}$ such that $\conf\subseteq\comp{\Dmc}{\Rmc_i}$. This implies that $\comp{\Dmc}{\Rmc_1}\not\subseteq\comp{\Dmc}{\Rmc_2}$ so that there exists $\lambda_1\in \comp{\Dmc}{\Rmc_1}\setminus\comp{\Dmc}{\Rmc_2}$. 
By Claim \ref{lem:pareto-chain}, there exists $\mu_1\in \comp{\Dmc}{\Rmc_2}\setminus\comp{\Dmc}{\Rmc_1}$ such that $\lambda_1\prec\mu_1$. 
\item Induction step: Assume that we have built $\lambda_1\prec\mu_1\prec\dots\prec\lambda_i\prec\mu_i$ as required. By Claim \ref{lem:pareto-chain}, since $\mu_i\in \comp{\Dmc}{\Rmc_2}\setminus\comp{\Dmc}{\Rmc_1}$, there exists $\lambda_{i+1}\in \comp{\Dmc}{\Rmc_1}\setminus\comp{\Dmc}{\Rmc_2}$ such that $\mu_i\prec\lambda_{i+1}$. Then by Claim \ref{lem:pareto-chain} again, since $\lambda_{i+1}\in \comp{\Dmc}{\Rmc_1}\setminus\comp{\Dmc}{\Rmc_2}$, we got $\mu_{i+1}\in \comp{\Dmc}{\Rmc_2}\setminus\comp{\Dmc}{\Rmc_1}$ such that $\lambda_{i+1}\prec\mu_{i+1}$. \qedhere
\end{itemize}

\end{proof}

\ScoreStructuredCollapse*
\begin{proof}
Let $\Smc_1,\dots,\Smc_n$ be the prioritization of $\bigcup_{\conf\in\conflicts{\Dmc,\Cmc}}\conf$. 

We first show that $ \deltapreps{\Dmc^\constraints_\succ}\subseteq \deltacreps{\Dmc^\constraints_\succ} $, which implies that $\deltacreps{\Dmc^\constraints_\succ} = \deltagreps{\Dmc^\constraints_\succ} = \deltapreps{\Dmc^\constraints_\succ}$. 
Let $\Rmc\in  \deltapreps{\Dmc^\constraints_\succ}$ and let $\succ'$ be a completion of $\succ$ such that for every $1\leq i\leq n$, for all $\lambda_1,\lambda_2 \in\Smc_i$ such that $\{\lambda_1,\lambda_2\}\subseteq\conf\in\deltaconflicts{\Dmc,\constraints}$, if $\lambda_1\in \comp{\Dmc}{\Rmc}$ and $\lambda_2\notin \comp{\Dmc}{\Rmc}$ then $\lambda_1\succ'\lambda_2$.  
Assume for a contradiction that $\Rmc$ is not Pareto-optimal \wrt $\succ'$: There exists a database $\Bmc$ consistent \wrt $\constraints$ such that there is $\mu\in \comp{\Dmc}{\Bmc}\setminus \comp{\Dmc}{\Rmc}$ with $ \mu\succ'\lambda$ for every $\lambda\in \comp{\Dmc}{\Rmc}\setminus \comp{\Dmc}{\Bmc}$. 
\begin{itemize}
\item Since $\Rmc\in\deltapreps{\Dmc^\constraints_\succ}$, $\Bmc$ is not a Pareto-improvement of $\Rmc$ \wrt $\succ$ so there exists $\lambda\in \comp{\Dmc}{\Rmc}\setminus \comp{\Dmc}{\Bmc}$ such that $\mu\not\succ\lambda$. 
\item Since $\succ'$ extends $\succ$, it follows that $\lambda$ and $\mu$ belong to the same $\Smc_i$ (otherwise it must be the case that $\mu\succ\lambda$ or $\lambda\succ\mu$ and since $ \mu\succ'\lambda$ the latter is not possible). 
\item Hence, as $\lambda\in \comp{\Dmc}{\Rmc}$ and $ \mu\notin \comp{\Dmc}{\Rmc}$, by construction of $\succ'$, $\lambda\succ'\mu$, which contradicts $ \mu\succ'\lambda$. 
\end{itemize}
It follows that $\Rmc$ is Pareto-optimal \wrt $\succ'$, so $\Rmc\in  \deltacreps{\Dmc^\constraints_\succ}$. 
\smallskip

We now show that  $\deltapreps{\Dmc^\constraints_\succ}=\deltascorereps{\Dmc^\constraints_\succ}$. 

\noindent($\Rightarrow$) Let $\Rmc\in  \deltareps{\Dmc,\constraints}$ be such that $\Rmc\notin \deltascorereps{\Dmc^\constraints_\succ}$: There exists a database $\Bmc$ consistent \wrt $\constraints$ such that 
 there is some $1\leq i\leq n$ such that $\comp{\Dmc}{\Rmc}\cap\Smc_i \subsetneq \comp{\Dmc}{\Bmc}\cap\Smc_i$ and for all $1\leq j<i$, $\comp{\Dmc}{\Bmc}\cap\Smc_j  = \comp{\Dmc}{\Rmc}\cap\Smc_j $. 
Hence there exists $\mu\in \Smc_i$ such that $\mu\in \comp{\Dmc}{\Bmc}\setminus\comp{\Dmc}{\Rmc}$, and for every $\lambda\in \comp{\Dmc}{\Rmc}\setminus \comp{\Dmc}{\Bmc}$, $\lambda\in\Smc_j$ for some $j> i$, so that $\mu\succ\lambda$.  
Thus $\Bmc$ is a Pareto-improvement of $\Rmc$ and $\Rmc\notin  \deltapreps{\Dmc^\constraints_\succ}$.

\noindent($\Leftarrow$) Let $\Rmc\in  \deltareps{\Dmc,\constraints}$ such that $\Rmc\notin  \deltapreps{\Dmc^\constraints_\succ}$: There exists a database $\Bmc$ consistent \wrt $\constraints$ such that 
 there is $ \mu\in \comp{\Dmc}{\Bmc}\setminus \comp{\Dmc}{\Rmc}$ with $\mu\succ\lambda$ for every $\lambda\in \comp{\Dmc}{\Rmc}\setminus \comp{\Dmc}{\Bmc}$. 
Let $\Smc_i$ be the priority level to which $\mu$ belongs. 
Every $\lambda\in \comp{\Dmc}{\Rmc}\setminus \comp{\Dmc}{\Bmc}$ is such that $\mu\succ\lambda$ so belongs to some $\Smc_j$ with $j>i$. 
Hence, for every $j\leq i$, $\comp{\Dmc}{\Rmc}\cap\Smc_j\subseteq\comp{\Dmc}{\Bmc}\cap\Smc_j $. 
Moreover, $\mu\in\comp{\Dmc}{\Bmc}\setminus\comp{\Dmc}{\Rmc}$ so $\comp{\Dmc}{\Rmc}\cap\Smc_i\subsetneq\comp{\Dmc}{\Bmc}\cap\Smc_i $. 
It follows that $\Rmc\notin \deltascorereps{\Dmc^\constraints_\succ}$.
\end{proof}

\section{Proofs for Section \ref{sec:complexity}}

\ThComplexityPGrepairchecking*
\begin{proof}
The lower bound is inherited from $\Delta$-repairs. 

We use the following \np\ procedure to check whether $\Rmc\notin\deltapreps{\Dmc^\constraints_\succ}$: we guess either  (i) `inconsistent', (ii) `not maximal' together with another candidate repair $\Rmc'$,
or (iii) `improvement' together with a candidate Pareto improvement $\Bmc$. 
In case (i), it suffices to verify in \ptime\ that $\Rmc \not \models \Cmc$, returning yes if so. 
In case (ii), we test in \ptime\ whether $\Rmc'\Delta\Dmc\subsetneq \Rmc\Delta\Dmc$ and $\Rmc' \models \Cmc$, returning yes if both conditions hold. 
In case (iii), we check in \ptime whether $\Bmc$ is indeed a Pareto improvement of $\Rmc$, returning yes if so: we 
check that $\Bmc \models \Cmc$, compute the sets 
$\comp{\Dmc}{\Bmc}\setminus \comp{\Dmc}{\Rmc}$ and 
$\comp{\Dmc}{\Rmc}\setminus \comp{\Dmc}{\Bmc}$,
 and consult the priority relation $\succ$ to determine if 
there is some 
$\mu\in \comp{\Dmc}{\Bmc}\setminus \comp{\Dmc}{\Rmc}$ with $\mu\succ\lambda$ for every $\lambda\in \comp{\Dmc}{\Rmc}\setminus \comp{\Dmc}{\Bmc}$.

For globally-optimal repairs, we proceed analogously, except that (iii) guesses a candidate global improvement $\Bmc$, which is verified in \ptime to be a global improvement.
\end{proof}

\ParetoHard*
\begin{proof} 

The proof is by reduction from 3SAT. We consider a schema that contains a unary relation $Var$, binary relations $Init$ and $Val$, and a 6-ary relation $Clause$. The set $\Cmc_{3SAT}$ will consist of the following constraints:
\begin{itemize}
\item $Init(x,y) \wedge Init(x,z) \wedge y\neq z \rightarrow \bot$ 
\item $Init(x,y) \wedge Var(z) \wedge x\neq y   \rightarrow Val(z,x) \vee Val(z,y)$
\item $Val(z,x) \wedge Val(z,y) \wedge x\neq y \rightarrow \bot$
\item $Clause(z_1, w_1, z_2, w_2, z_3, w_3) \wedge Val(z_1,w_1) \wedge Val(z_2,w_2) \wedge Val(z_3,w_3) \rightarrow \bot$
\end{itemize}
Given a 3SAT instance $\varphi= \lambda_1 \wedge \ldots \wedge \lambda_k$ over variables $u_1, \ldots, u_n$, 
where $\lambda_j= \ell_{j,1} \vee \ell_{j,2} \vee \ell_{j,3}$
we build the database $\Dmc_\varphi$ consisting of the following facts:
\begin{itemize}
\item $Init(0,0), Init(0,1)$;
\item $Var(u_i)$, for $1 \leq i \leq n$;
\item $Clause(u_{j,1}, b_{j,1},u_{j,2}, b_{j,2},u_{j,3}, b_{j,3})$, for $1 \leq j \leq k$, where $u_{j,h}$ is the variable of literal $\ell_{j,h}$,
$b_{j,h}=0$ if $\ell_{j,h}=u_{j,h}$, 
and $b_{j,h}=1$ if $\ell_{j,h}= \neg u_{j,h}$; 
\end{itemize}
The priority relation $\succ$ 
contains $Init(0,1) \succ Init(0,0)$ and $Init(0,1) \succ \neg Val(u_j,b)$ for 
$1 \leq j \leq n$ and $b \in \{0,1\}$. 

Let $\Rmc=\Dmc_\varphi \setminus \{Init(0,1)\}$. It is easy to check that $\Rmc$ is a $\Delta$-repair. 
We show that 
$\Rmc$ is
 Pareto-optimal iff $\varphi$ is unsatisfiable, or equivalently, that $\Rmc$ has a Pareto improvement iff $\varphi$ is satisfiable.

First suppose that $\Rmc = \Dmc_\varphi \setminus \{Init(0,1)\}$ has a Pareto improvement $\Bmc$. We know that $\Bmc$ is consistent and that there exists $\mu\in \comp{\Dmc_\varphi }{\Bmc}\setminus \comp{\Dmc_\varphi }{\Rmc}$ with $\mu\succ\mu'$ for every $\mu'\in \comp{\Dmc_\varphi}{\Rmc}\setminus \comp{\Dmc_\varphi}{\Bmc}$. Clearly, we must have $\mu = Init(0,1)$ as it is the only literal from $Lits^{\preds}_{\Dmc_\varphi}$ which does not occur in $\comp{\Dmc_\varphi}{\Rmc}$. Due to the definition of $\succ$, we know that $\comp{\Dmc_\varphi}{\Rmc}\setminus \comp{\Dmc_\varphi}{\Bmc}$ can only contain $Init(0,0)$ and negated $Val$ facts. In particular, this means that 
$\Bmc$ must contain all of the $Var$ and $Clause$ facts from $\Dmc_\varphi$. Since the constraint $Init(x,y) \wedge Var(z) \wedge x \neq y \rightarrow Val(z,x) \vee Val(z,y)$ is satisfied by $\Bmc$, we can further infer that for every $1 \leq i \leq n$, either $Val(u_i,0)$ or $Val(u_i,1)$ belongs to $\Bmc$. The constraint  $Init(x,y) \wedge Val(z,x) \wedge Val(z,y) \wedge x\neq y \rightarrow \bot$ ensures a single truth value is selected for each $u_j$. 
Thus $\Bmc$ defines a valuation of the variables $u_1, \ldots, u_n$. Further observe that the constraint $Clause(z_1, w_1, z_2, w_2, z_3, w_3) \wedge Val(z_1,w_1) \wedge Val(z_2,w_2) \wedge Val(z_3,w_3) \rightarrow \bot$ ensures that this valuation does not violate any of the clauses. Thus, we may conclude that the formula $\varphi$ is satisfiable. 

Conversely, if $\varphi$ is satisfiable, take some satisfying valuation $\nu$. Define $\Bmc_\nu$ as $\Dmc_\varphi \cup \{Init(0,1)\} \setminus \{Init(0,0)\} \cup \{Val(u_i,0) \mid \nu(u_i)=0\} \cup \{Val(u_i,1) \mid \nu(u_i)=1\}$. It is not hard to see that $\Bmc_\nu$ is consistent. To see why it is Pareto improvement of $\Rmc =\Dmc_\varphi \setminus \{Init(0,1)\}$, observe that $\comp{\Dmc_\varphi}{\Rmc}\setminus \comp{\Dmc_\varphi}{\Bmc_\nu} = \{Init(0,0)\} \cup \{\neg Val(j,0) \mid \nu(u_j)=0\} \cup \{\neg Val(j,1) \mid \nu(u_j)=1\}$ and that $Init(0,1) \in \comp{\Dmc_\varphi}{\Bmc_\nu}\setminus \comp{\Dmc_\varphi}{\Rmc}$ is such that $Init(0,1) \succ \mu'$
for every $\mu'\in \comp{\Dmc_\varphi}{\Rmc}\setminus \comp{\Dmc_\varphi}{\Bmc}$.  
\end{proof}

\LemConflictChecking*
\begin{proof}
Let $\conf\subseteq\litset$. By definition, $\conf\in\conflicts{\Dmc,\constraints}$ iff (i) for every database instance $\Imc$, if $\Imc\models \conf$, then $\Imc\not\models\constraints$ and (ii) for every $\conf'\subsetneq\conf$, there exists $\Imc$ such that $\Imc\models \conf'$ and $\Imc\models\constraints$. 
\begin{itemize}
\item Condition (i) can be checked in \conp (to show that it does not hold, guess $\Imc$ and check in \ptime\ that $\Imc\models \conf$ and $\Imc\models\constraints$).
\item Condition (ii) can be checked in \np: Let $\conf=\{\lambda_1,\dots,\lambda_n\}$. If there exists $\conf'\subsetneq\conf$ such that for every database instance $\Imc$,  $\Imc\models \conf'$ implies $\Imc\not\models\constraints$, then this is also the case for every $\conf''$ such that $\conf'\subseteq\conf''\subseteq\conf$. Hence it is sufficient to check the condition for every $\conf'_i=\conf\setminus\{\lambda_i\}$, $1\leq i\leq n$. This can be done by guessing $\Imc_1,\dots,\Imc_n$ such that $\Imc_i\models \conf'_i$ and $\Imc_i\models\constraints$. 
\end{itemize}
As the \np\ calls in (ii) can be grouped into a single call, we obtain membership in $BH_2$. 

We show $BH_2$-hardness by reduction from the problem of testing whether an input 3CNF is a minimal unsatisfiable subset (MUS), i.e. it is unsatisfiable and removing any clause renders the formula satisfiable. This problem was proven $BH_2$-complete in \cite{DBLP:journals/ai/Liberatore05}. 
We shall use the following set of constraints $\Cmc_{MUS}$: 
\begin{itemize}
\item $Val(z,x) \wedge Val(z,y) \wedge x\neq y \rightarrow \bot$
\item $Clause(z_1, w_1, z_2, w_2, z_3, w_3) \rightarrow Val(z_1,w_1) \vee Val(z_2,w_2) \vee Val(z_3,w_3)$
\end{itemize}
The input 3CNF $\varphi= \lambda_1 \wedge \ldots \wedge \lambda_k$ over variables $u_1, \ldots, u_n$, 
where $\lambda_j= \ell_{j,1} \vee \ell_{j,2} \vee \ell_{j,3}$, is captured by the instance $D_\varphi$ consisting of the facts:
\begin{itemize}
\item $F_j= Clause(u_{j,1}, b_{j,1},u_{j,2}, b_{j,2},u_{j,3}, b_{j,3})$, for $1 \leq j \leq k$, where $u_{j,h}$ is the variable from literal $\ell_{j,h}$,
$b_{j,h}=1$ if $\ell_{j,h}=u_{j,h}$, 
and $b_{j,h}=0$ if $\ell_{j,h}= \neg u_{j,h}$; 
\end{itemize}
Note that the $b_{j,h}$ positions of the $Clause$ relation are defined differently from the reduction in Lemma \ref{pareto-hard}, as here we give the assignments to the variables that can be used to satisfy the clause, rather than those that violate the clause. 

We claim that $\varphi$ is a MUS iff $D_\varphi \in \conflicts{\Dmc_\varphi,\constraints_{MUS}}$:
\begin{itemize}
\item First suppose that $\varphi$ is a MUS. Then $\varphi$ is unsatisfiable, but every formula $\varphi^-$ obtained by removing one or more clauses from $\varphi$ is satisfiable. Suppose for a contradiction that there exists an instance $\Imc$ such that 
$\Imc\models \conf$ and  $\Imc\models\constraints_{MUS}$. Due to the second constraint, for each fact $F_j=Clause(u_{j,1}, b_{j,1},u_{j,2}, b_{j,2},u_{j,3}, b_{j,3}) \in D_\varphi $ there is $h \in \{1,2,3\}$ such that $Val(u_{j,h}, b_{j,h}) \in \Imc$. Moreover, by the first constraint, there does not exist any $u_i$ such that $\Imc$ contains both $Val(u_i,1)$ and $Val(u_i,0)$. We can thus define a valuation $\nu$ such that $Val(u_i, v_i) \in \Imc$ implies $\nu(u_i)=v_i$. By construction, this valuation will satisfy all clauses of $\varphi$, a contradiction. It follows that every $\Imc$ with $\Imc \models D_\varphi$ is such that  $\Imc \not \models\constraints_{MUS}$. It remains to show that $D_\varphi$ is minimal with this property. Take any proper subset $D^-_\varphi \subsetneq D_\varphi$, and let $\varphi^-$ be the corresponding formula.  As $\varphi^-$ is satisfiable,  we can find a satisfying valuation $\nu$. Let $\Imc_\nu= D^-_\varphi \cup \{Val(u,\nu(u)) \mid u \in \mathsf{vars}(\varphi^-)\}$. By definition, we have $\Imc_\nu \models D^-_\varphi$, and it is easily verified that 
$\Imc_\nu \models \constraints_{MUS}$, completing the argument. 
\item For the other direction, suppose that $D_\varphi \in \conflicts{\Dmc_\varphi,\constraints_{MUS}}$. It follows that (i) every $\Imc$ with $\Imc \models D_\varphi$ is such that  $\Imc \not \models\constraints_{MUS}$, and (ii) for every proper subset $D_\varphi^- \subsetneq D_\varphi$, there is some $\Imc $
such that $\Imc \models D^-_\varphi$ and $\Imc \models\constraints_{MUS}$. Let $\nu$ be a valuation of $u_1, \ldots, u_n$, and let $\Imc_\nu = D_\varphi  \cup \{Val(u_i,\nu(u_i)) \mid 1 \leq i \leq n\}$. Since $\Imc_\nu \models D_\varphi$, it follows from (i) that the second constraint is violated, and hence that $\nu$ does not satisfy all of the clauses of $\varphi$. As this is true of any valuation, $\varphi$ is unsatisfiable. Now consider any $\varphi^-$ obtained by removing one or more clauses from $\varphi$, and let $D^-_\varphi \subsetneq D_\varphi$ be the corresponding instance. By (ii), there exists $\Imc $
such that $\Imc \models D^-_\varphi$ and $\Imc \models\constraints_{MUS}$. Following the same argument as in the first item, we can infer that 
$\varphi^-$ is satisfiable. We conclude that $\varphi$ is a MUS. \qedhere
\end{itemize}
\end{proof}

We will use the following lemma in the proof of Theorem \ref{th:completion-repair-checking}.
\begin{lemma}\label{lem:order-from-prio}
For every priority relation $\succ$ for $\Dmc$ \wrt $\constraints$, there exists a total order $>$ over $\litset$ such that if $\lambda\succ\mu$ then $\lambda>\mu$. 
\end{lemma}
\begin{proof}
We can build $>$ as follows: 
Let $>=\emptyset$ , $L=\litset$, and repeat the following step until $L=\emptyset$.
\begin{itemize}
\item Let $NonDom=\{\lambda\mid \lambda\in L, \forall \mu\in L, \mu\not\succ\lambda\}$. 
\item Extend $>$ by (1) setting $\mu > \lambda$ for every $\lambda\in NonDom$ and $\mu\in \litset\setminus L$ and (2) arbitrarily ordering $NonDom$. 
\item Let $L\leftarrow L\setminus NonDom$.
\end{itemize} 

Since $\succ$ is acyclic, the procedure terminates: for every $\lambda\in\litset$, there is a step in which $\lambda\in NonDom$ (otherwise, we can build an infinite chain $\lambda\prec \mu_1\prec\mu_2\prec\dots$ and $\litset$ is finite). 
For every $\lambda,\mu\in\litset$, either $\lambda>\mu$ or $\mu>\lambda$ and $>$ is acyclic. 
Moreover, if $\lambda\succ\mu$, since $\mu$ cannot belong to $NonDom$ while $\lambda\in L$, it follows that $\lambda>\mu$ 
\end{proof}

\ThCompletionRepairChecking*

\begin{proof}
The lower bound is inherited from S-repair checking. 
For the upper bound, we rely on the following \sigmaptwo decision procedure to decide whether a given database $\Rmc$ belongs to $\deltacreps{\Dmc^\Cmc_\succ}$:
\begin{enumerate}
\item Check in polynomial time that $\Rmc$ is a candidate repair (\ie $\Rmc\subseteq \factset$) and that $\Rmc\models\constraints$.
\item Let $\litset\setminus\comp{\Dmc}{\Rmc}=\{\lambda_1,\dots,\lambda_n\}$. Guess:
\begin{itemize}
\item a total order $>$ over $\litset$ such that if $\lambda\succ\mu$ then $\lambda>\mu$ (such order exists by Lemma \ref{lem:order-from-prio}), and 
\item for each $1\leq i\leq n$, a set of literals $\conf_i$ such that $\lambda_i\in\conf_i$, $\conf_i\setminus\{\lambda_i\}\subseteq\comp{\Dmc}{\Rmc}$, and $\mu>\lambda_i$ for every $\mu\in\conf_i\setminus\{\lambda_i\}$.
\end{itemize}
\item Check that $>$ is as required in polynomial time, then for each $1\leq i\leq n$, check that $\conf_i\in\conflicts{\Dmc,\constraints}$ in $BH_2$ (by Lemma~\ref{lem:conflict-checking}).
\end{enumerate}
Step (3) makes a polynomial number of calls to a $BH_2$ oracle hence runs in \deltaptwo. Thus the global procedure runs in \sigmaptwo. 
We next show that it indeed decides whether $\Rmc\in\deltacreps{\Dmc^\Cmc_\succ}$. 
\begin{itemize}
\item Assume that $\Rmc\in\deltacreps{\Dmc^\Cmc_\succ}$. 
\begin{itemize}
\item Since $\Rmc\in\deltareps{\Dmc,\constraints}$, then $\Rmc$ is a candidate repair and $\Rmc\models\constraints$ as required in step~(1). 
\item Since $\Rmc\in\deltacreps{\Dmc^\Cmc_\succ}$, there exists a completion $\succ'$ of $\succ$ such that $\Rmc\in\deltapreps{\Dmc^\Cmc_{\succ'}}$. Let $>$ be a total order over $\litset$ such that $\lambda\succ'\mu$ implies $\lambda>\mu$ (such order exists by Lemma \ref{lem:order-from-prio}). Since $\succ'$ is a completion of $\succ$, $\lambda\succ\mu$ implies $\lambda>\mu$ as required in step~(2).
\item 
Assume for a contradiction that there is $\lambda\in\litset\setminus\comp{\Dmc}{\Rmc}$ such that there is no $\conf\in\conflicts{\Dmc,\constraints}$ such that $\lambda\in\conf$, $\conf\setminus\{\lambda\}\subseteq\comp{\Dmc}{\Rmc}$, and for every $\mu\in\conf\setminus\{\lambda\}$, $\mu>\lambda$ (\ie $\mu\succ'\lambda$).
\begin{itemize}  
\item Since $\Rmc\in\deltareps{\Dmc,\constraints}$, by Proposition \ref{prop:characterizations-repairs-comp}, $\comp{\Dmc}{\Rmc}$ is a maximal subset of $\litset$ such that there is no conflict $\conf\in\conflicts{\Dmc,\constraints}$ such that $\conf\subseteq\comp{\Dmc}{\Rmc}$. Hence, since $\lambda\in\litset\setminus\comp{\Dmc}{\Rmc}$, there exists $\conf\in\conflicts{\Dmc,\constraints}$ such that $\lambda\in\conf$ and $\conf\setminus\{\lambda\}\subseteq\comp{\Dmc}{\Rmc}$.

\item By assumption, for every such $\conf$, there exists $\mu\in\conf\setminus\{\lambda\}$ such that $\mu\not\succ'\lambda$, so such that $\lambda\succ'\mu$.
\item Let $\Bmc=\comp{\Dmc}{\Rmc}\cup\{\lambda\}\setminus\{\mu\mid \lambda\succ'\mu\}$. By construction, there is no $\conf\in\conflicts{\Dmc,\constraints}$ such that $\conf\subseteq\Bmc$. 
By Lemma~\ref{lem:pareto-improvement-if-one-beat-all}, it follows that $\Rmc\notin\deltapreps{\Dmc^\Cmc_{\succ'}}$, a contradiction.
\end{itemize}
Hence there exists $\conf_1,\dots,\conf_n$ as required by step (2).
\end{itemize}
\item Assume that $\Rmc\notin\deltacreps{\Dmc^\Cmc_\succ}$. 
\begin{itemize}
\item If $\Rmc\notin\deltareps{\Dmc,\constraints}$, then $\Rmc\notin\deltacreps{\Dmc^\Cmc_\succ}$ is detected by step (1). \item Otherwise, assume for a contradiction that there exists a total order $>$ as required by step (2) such that for every $\lambda\in \litset\setminus\comp{\Dmc}{\Rmc}$, there exists $\conf\in\conflicts{\Dmc,\constraints}$ such that $\lambda\in\conf$, $\conf\setminus\{\lambda\}\subseteq\comp{\Dmc}{\Rmc}$, and for every $\mu\in\conf\setminus\{\lambda\}$, $\mu>\lambda$.
\begin{itemize}
\item Let $\succ'$ be the completion of $\succ$ induced by $>$. Since $\Rmc\notin\deltacreps{\Dmc^\Cmc_\succ}$, then $\Rmc\notin\deltapreps{\Dmc^\Cmc_{\succ'}}$. Thus there exists a database $\Rmc'$ such that $\Rmc'\models\constraints$ and there exists $\lambda\in\comp{\Dmc}{\Rmc'}\setminus\comp{\Dmc}{\Rmc}$ such that for every $\mu\in \comp{\Dmc}{\Rmc}\setminus\comp{\Dmc}{\Rmc'}$, $\lambda\succ'\mu$.
\item Since $\lambda\in\litset\setminus\comp{\Dmc}{\Rmc}$, then by assumption there exists $\conf\in\conflicts{\Dmc,\constraints}$ such that $\lambda\in\conf$, $\conf\setminus\{\lambda\}\subseteq\comp{\Dmc}{\Rmc}$, and for every $\mu\in\conf\setminus\{\lambda\}$, $\mu>\lambda$ so that $\mu\succ'\lambda$. Since $\comp{\Dmc}{\Rmc}\setminus\comp{\Dmc}{\Rmc'}\subseteq\{\mu\mid \lambda\succ'\mu\}$ and $\lambda\in\comp{\Dmc}{\Rmc'}$, it follows that $\conf\subseteq\comp{\Dmc}{\Rmc'}$. 
Hence $\Rmc'\models\conf$ and $\Rmc'\not\models\constraints$, contradicting $\Rmc'\models\constraints$.
\end{itemize}
Hence there does not exist a total order $>$ and $\conf_1,\dots,\conf_n$ as required by step (2).\qedhere
\end{itemize}
\end{itemize}
\end{proof}

\ThComplexityQueryAnswering*
\begin{proof}

The upper bounds for query answering under X-brave (resp.\ X-CQA and X-intersection) follow from the complexity of X-repair checking and standard query answering. 
We use the following procedures:
\begin{itemize}
\item To decide $\Dmc^\constraints_\succ\not\armodels{X} q$ (resp.\ $\Dmc^\constraints_\succ\bravemodels{X} q$), guess $\Rmc\in\deltaxreps{\Dmc^\constraints_\succ}$ such that $\Rmc\not\models q$ (resp. $\Rmc\models q$).  
\item To decide $\Dmc^\constraints_\succ\not\iarmodels{X} q$, 
compute in polynomial time the subsets $\Bmc_1,\dots,\Bmc_n$ of $\factset$ which are images of $q$ by some homomorphism (there are polynomially many such subsets since their size is bounded by the number of relational atoms in $q$) and guess $\Rmc_1,\dots,\Rmc_n\in\deltaxreps{\Dmc^\constraints_\succ}$ such that $\Bmc_i\not\subseteq\Rmc_i$. 
Since $\bigcap_{\Rmc\in\deltaxreps{\Dmc^\constraints_\succ}}\Rmc\subseteq \bigcap_{i=1}^n\Rmc_i$, this implies that $\Dmc^\constraints_\succ\not\iarmodels{X}q$.%
\end{itemize}
For Pareto- and globally-optimal repairs, since X-repair checking is in \conp and query answering is in \ptime, the \sigmaptwo and \piptwo upper bounds follow immediately. 

For completion-optimal repairs, since C-repair checking is in \sigmaptwo, for each guessed repair $\Rmc$ we can guess together with $\Rmc$ a certificate that $\Rmc\in \deltacreps{\Dmc^\constraints_\succ}$ that can be verified in \deltaptwo, and obtain the \sigmaptwo and \piptwo upper bounds. 
\smallskip
 
The lower bounds follows from the proof of \piptwo-hardness for query answering under S-CQA given in Theorem 6 in \cite{DBLP:journals/is/StaworkoC10}.  The query used in this reduction is $q=\bar{r}$ where $\bar{r}\in\Dmc$ is such that $\bar{r}$ does not belong to a $\Delta$-repair $\Rmc$ iff a fact $r$ originally not in the database is in $\Rmc$. 
Using the same database and set of constraints as \citeauthor{DBLP:journals/is/StaworkoC10} and an empty priority relation, since $\Dmc^\constraints_\succ\armodels{X}q $ iff $\Dmc^\constraints_\succ\iarmodels{X}q $ iff $\Dmc^\constraints_\succ\bravemodels{X} r $, we obtained the \piptwo and \sigmaptwo hardness results. 
\end{proof}

\ComplexityWithConflictsGiven*
\begin{proof}
We use the reduction of Proposition~\ref{prop:reduction-UC-denials} from universal constraints to ground denial constraints, which defines $\Dmc_d$ and $\constraints_{d,\Dmc}$ such that $\conflicts{\Dmc_d,\constraints_{d,\Dmc}}=\{ \tofacts(\conf)\mid \conf\in\conflicts{\Dmc,\constraints}\}$ and $\deltareps{\Dmc_d,\constraints_{d,\Dmc}}=\{\tofacts(\comp{\Dmc}{\Rmc})\mid \Rmc\in\deltareps{\Dmc,\constraints}\}$. 
If $\conflicts{\Dmc,\constraints}$ is given, the construction of $\Dmc_d$ and $\constraints_{d,\Dmc}$ can be done in polynomial time \wrt the size of $\Dmc$ and $\conflicts{\Dmc,\constraints}$. 

The priority relation $\succ$ over the literals of $\conflicts{\Dmc, \constraints}$ corresponds straightforwardly to a priority relation   $\succ_d$ over the facts of $\conflicts{\Dmc_d, \constraints_{d,\Dmc}}$: $\alpha\succ_d\beta$ iff $\lambda\succ\mu$ and $\tofacts(\{\lambda\})=\{\alpha\}$, $\tofacts(\{\mu\})=\{\beta\}$. 
Let $\Rmc\in\deltareps{\Dmc,\constraints}$. 
\begin{itemize}
\item There exists a Pareto improvement of $\Rmc$ (\wrt $\Dmc^\constraints_\succ$) iff there is a Pareto improvement of $\tofacts(\comp{\Dmc}{\Rmc})$ (\wrt ${\Dmc_d}^{\constraints_{d,\Dmc}}_{\succ_d}$).
\begin{itemize}
\item If $\Bmc$ is a Pareto improvement of $\Rmc$:
\begin{itemize}
\item $\Bmc\models\constraints$ so for every $\conf\in\conflicts{\Dmc,\constraints}$, $\Bmc\not\models\conf$, so that $\conf\not\subseteq\comp{\Dmc}{\Bmc}$, \ie  $\tofacts(\conf)\not\subseteq\tofacts(\comp{\Dmc}{\Bmc})$. It follows that $\tofacts(\comp{\Dmc}{\Bmc})\models \constraints_{d,\Dmc}$.
\item There is $\mu\in\comp{\Dmc}{\Bmc}\setminus\comp{\Dmc}{\Rmc}$ with $\mu\succ\lambda$ for every $\lambda\in\comp{\Dmc}{\Rmc}\setminus\comp{\Dmc}{\Bmc}$, \ie there is $\alpha\in \tofacts(\comp{\Dmc}{\Bmc})\setminus\tofacts(\comp{\Dmc}{\Rmc})$ such that $\alpha\succ_d\beta$ for every $\beta\in\tofacts(\comp{\Dmc}{\Bmc})\setminus\tofacts(\comp{\Dmc}{\Rmc})$.  
\end{itemize}
Hence $\tofacts(\comp{\Dmc}{\Bmc})$ is a Pareto improvement of $\tofacts(\comp{\Dmc}{\Rmc})$.
\item If there is a Pareto improvement of $\tofacts(\comp{\Dmc}{\Rmc})$, it is of the form $\tofacts(\Bmc)\subseteq \Dmc_d=\tofacts(\litset)$ for some $\Bmc\subseteq\litset$:
\begin{itemize}
\item $\tofacts(\Bmc)\models \constraints_{d,\Dmc}$, so for every $\conf\in\conflicts{\Dmc,\constraints}$, $\tofacts(\conf)\not\subseteq\tofacts(\Bmc)$, \ie $\conf\not\subseteq\Bmc$. Let $\Bmc'$ be a maximal subset of $\litset$ that includes $\Bmc$ and does not include any conflict. By Proposition \ref{prop:characterizations-repairs-comp} and Lemma \ref{lem:compl-restr}, $\restr{\Dmc}{\Bmc'}\in\deltareps{\Dmc,\constraints}$, so $\restr{\Dmc}{\Bmc'}\models\constraints$.
\item There is $\alpha\in \tofacts(\Bmc)\setminus\tofacts(\comp{\Dmc}{\Rmc})$ such that $\alpha\succ_d\beta$ for every $\beta\in\tofacts(\Bmc)\setminus\tofacts(\comp{\Dmc}{\Rmc})$, so there is $\mu\in\Bmc\subseteq\Bmc'=\comp{\Dmc}{\restr{\Dmc}{\Bmc'}}$ such that $\mu\succ\lambda$ for every $\lambda\in\comp{\Dmc}{\Rmc}\setminus\comp{\Dmc}{\restr{\Dmc}{\Bmc'}}$.
\end{itemize}
Hence $\restr{\Dmc}{\Bmc'}$ is a Pareto improvement of $\Rmc$.
\end{itemize}
\item There exists a global improvement of $\Rmc$ (\wrt $\Dmc^\constraints_\succ$) iff there is a global improvement of $\tofacts(\comp{\Dmc}{\Rmc})$ (\wrt ${\Dmc_d}^{\constraints_{d,\Dmc}}_{\succ_d}$). The proof is similar to the Pareto improvement case.
\end{itemize}

It follows that for every $\Rmc\in\deltareps{\Dmc,\constraints}$, $\Rmc\in\deltaxreps{\Dmc^\constraints_\succ}$ iff $\tofacts(\comp{\Dmc}{\Rmc})\in\deltaxreps{{\Dmc_d}^{\constraints_{d,\Dmc}}_{\succ_d}}$ for $X\in\{G,P,C\}$. 
Hence, to decide whether a database $\Rmc$ belongs to $\deltaxreps{\Dmc^\constraints_\succ}$, we can check in polynomial time that $\tofacts(\comp{\Dmc}{\Rmc})\in\deltareps{{\Dmc_d}^{\constraints_{d,\Dmc}}_{\succ_d}}$ then check that $\tofacts(\comp{\Dmc}{\Rmc})\in\deltaxreps{{\Dmc_d}^{\constraints_{d,\Dmc}}_{\succ_d}}$ in \conp for $X=G$ or  in \ptime   for $X\in\{P,C\}$. 

The complexity upper bounds for the query answering problems follow from the complexity of X-repair checking and query answering as in the proof of Theorem \ref{th:complexityQueryAnswering}. 
\smallskip

We obtain the lower bounds by adapting proofs of complexity hardness of repair checking and query answering with denial constraints. First note that with denial constraints, the conflicts can be computed in polynomial time \wrt data complexity. Hence, they can be assumed to be given without changing the data complexity. Second, note that for every fact $\alpha\in\Dmc$, (1) $\Dmc^\constraints_\succ\armodels{X} \alpha$ iff $\Dmc^\constraints_\succ\iarmodels{X} \alpha$ and (2) if we add to $\Dmc$ a fact $\beta$ with a fresh predicate that does not occur in $\constraints$, to $\constraints$ a denial constraint $\alpha\wedge\beta\rightarrow\bot$ and let $\alpha\succ\beta$, then $\Dmc^\constraints_\succ\armodels{X} \alpha$ iff $\Dmc^\constraints_\succ\not\bravemodels{X} \beta$. It follows that a reduction of a hard problem to query answering under X-CQA (resp.\ X-intersection) semantics that uses a ground atomic query can be adapted into a reduction of the same problem to query answering under X-intersection (resp.\ X-CQA) or X-brave semantics. 
\begin{itemize}
\item The proof of Theorem 2 in \cite{DBLP:journals/amai/StaworkoCM12} uses functional dependencies and a ground atomic query to show \conp-hardness of G-repair checking and \piptwo-hardness of query answering under G-CQA semantics. We can thus obtain the lower bounds in the globally-optimal repair case.
\item The proof of Proposition 6.2.8 in \cite{DBLP:phd/hal/Bourgaux16} uses a DL-Lite TBox which actually consists of denial constraints and a ground atomic ground query to show \conp-hardness of query answering under $\subseteq_P$-intersection semantics (with a score-structured priority). Since with score-structured priority relations, Pareto-optimal and completion-optimal repairs coincide with $\subseteq_P$-repairs, we can obtain the lower bounds for the Pareto- and completion-optimal repair cases.\qedhere
\end{itemize}
\end{proof}

\BoundedConflictsUndec*
\begin{proof}
The proof is by reduction from the Datalog boundedness problem \cite{DBLP:journals/jacm/GaifmanMSV93}: Given a Datalog program $\Pi$, \ie a finite set of rules of the form $R_1(\vec{x_1})\wedge\dots\wedge R_n(\vec{x_n})\rightarrow P(\vec{y})$ with $\vec{y}\subseteq \vec{x_1}\cup\dots\cup\vec{x_n}$, decide whether there exists a bound $k$ such that for every database instance $\Dmc$, $\Pi^k(\Dmc)=\Pi^\infty(\Dmc)$, where for a set of facts $\Smc$, $\Pi(\Smc)$ extends $\Smc$ with all facts that can be obtained by applying some rule of $\Pi$ to $\Smc$, $\Pi^0(\Dmc)=\Dmc$ and $\Pi^{i+1}(\Dmc)=\Pi(\Pi^i(\Dmc))$. 

It is well-known that each fact from $\Pi^i(\Dmc)$ can be associated with at least one proof tree whose leaves are facts from $\Dmc$ and whose inner nodes are facts obtained from their children by applying some rule of $\Pi$. If $\alpha\in \Pi^i(\Dmc)\setminus\Pi^{i-1}(\Dmc)$, all its proof trees have height at least $i$ and $\alpha$ has a proof tree of height $i$. 

Let $\Pi$ be a Datalog program and define $\constraints_\Pi=\Pi$. It follows from the form of Datalog rules that $\constraints_\Pi$ is a set of full tuple-generating dependencies. We show that there exists a bound $k$ such that for every database instance $\Dmc$, $\Pi^k(\Dmc)=\Pi^\infty(\Dmc)$, if and only if there exists a bound $M$ such that for every $\Dmc$, $\max_{\conf\in\conflicts{\Dmc,\constraints_\Pi}}(|\conf|)\leq M$. 

\noindent($\Rightarrow$) Assume for a contradiction that there exists a bound $k$ such that for every database instance $\Imc$, $\Pi^k(\Imc)=\Pi^\infty(\Imc)$ while for every integer $M$, there exists a database instance $\Dmc$ and a conflict $\conf_M\in\conflicts{\Dmc,\constraints_\Pi}$ such that $|\conf_M|>M$.  

Let $n_{max}$ be the maximal number of atoms in the bodies of the Datalog rules in $\Pi$, set $M=n_{max}^k+1$ and let $\Dmc$ and $\conf_M\in\conflicts{\Dmc,\constraints_\Pi}$ be such that $|\conf_M|>M$.

\begin{itemize}
\item By Proposition \ref{prop:defconflicts}, $\conf_M$ corresponds to some prime implicant of the disjunction of the ground contraint bodies $\bigvee_{\varphi\rightarrow \bot\in gr_\Dmc(\constraints_\Pi)} \varphi $. Since the constraints in $\constraints_\Pi$ all contain exactly one negative literal, it follows that $\conf_M$ contains exactly one negative literal $\neg \alpha_M$. Let $\Smc_M=\conf_M\setminus\{\neg\alpha_M\}$ be the set of facts from $\conf_M$. 

\item Since $\conf_M\in\conflicts{\Dmc,\constraints_\Pi}$, $\conf_M$ is a minimal subset of $\litset$ such that for every database instance $\Imc$, 
if $\Imc\models \conf_M$, then $\Imc\not\models\constraints_\Pi$. 
Hence every database instance $\Imc$ such that $\Imc\models\constraints_\Pi$ and $\Imc\models \conf_M\setminus\{\neg\alpha_M\}$ is such that $\Imc\models\alpha_M$. It follows that $\Pi,\Smc_M\models \alpha_M$, \ie $\alpha_M\in\Pi^\infty(\Smc_M)=\Pi^k(\Smc_M)$. 

\item Since $\alpha_M\in\Pi^k(\Smc_M)$, $\alpha_M$ has a proof tree $\tau$ of height at most $k$. Moreover, the degree of any proof tree is bounded by $n_{max}$ so the number of leaves of $\tau$ is at most $n_{max}^k<M<|\conf_M|$. Since $|\Smc_M|=|\conf_M|-1$, it follows that $n_{max}^k<|\Smc_M|$. Hence there exists $\beta\in\Smc_M$ such that $\Pi,\Smc_M\setminus\{\beta\}\models \alpha_M$. 

\item It follows that for every database instance $\Imc$ such that $\Imc\models\constraints_\Pi$ and $\Imc\models \Smc_M\setminus\{\beta\}$, $\Imc\models \alpha_M$. Hence for every database instance such that $\Imc\models \Smc_M\setminus\{\beta\}\cup\{\neg\alpha_M\}$, $\Imc\not\models\constraints_\Pi$. Since $\Smc_M\setminus\{\beta\}\cup\{\neg\alpha_M\}=\conf_M\setminus\{\beta\}$, this contradicts the fact that $\conf_M\in \conflicts{\Dmc,\constraints_\Pi}$.
\end{itemize}
It follows that the boundedness of $\Pi$ implies that the size of the conflicts \wrt $\constraints_\Pi$ can be bounded independently from the database.

\noindent($\Leftarrow$) Assume for a contradiction that 
there exists a bound $M$ such that for every database instance $\Dmc$ and conflict $\conf\in\conflicts{\Dmc,\constraints_\Pi}$, $|\conf|\leq M$ while for every integer $k$, there exists a database instance $\Imc_k$ such that $\Pi^k(\Imc_k)\neq\Pi^\infty(\Imc_k)$.  
Let $p$ be the number of predicates that occur in $\Pi$ and  $a_{max}$ be the maximal arity of such predicates, and set $k= p\times (M\times a_{max})^{a_{max}}$. 

\begin{itemize}
\item Since $\Pi^k(\Imc_k)\neq\Pi^\infty(\Imc_k)$, there exists $\alpha_k\in \Pi^{k+1}(\Imc_k)\setminus\Pi^k(\Imc_k)$. Hence all proof trees of $\alpha_k$ have height at least $k+1$. 

\item Let $\tau$ be a proof tree of $\alpha_k$ such that (i) $\tau$ is non-recursive, \ie does not contain two nodes labelled with the same fact such that one node is the descendant of the other, and (ii) its set of leaves $\Smc_k\subseteq \Imc_k$ is set-minimal among the sets of leaves of proof trees for $\alpha_k$. 
Such a proof tree exists because for every proof tree $\tau'$, there exists a non-recursive proof tree whose leaves are a subset of the leaves of $\tau'$ 
(\cf proof of Proposition 12 in \cite{DBLP:conf/kr/BourgauxBPT22}).

\item $\conf=\Smc_k\cup\{\neg\alpha_k\}$ is a conflict of $\Imc_k$ \wrt $\constraints_\Pi$: Indeed, since $\Pi, \Smc_k\models \alpha_k$, every database instance $\Imc$ such that $\Imc\models \conf$ is such that $\Imc\not\models \constraints_\Pi$ and we can check that $\conf$ is minimal, so that $\conf$ is a conflict.
\begin{itemize} 
\item For every $\beta\in\conf$ such that $\beta\neq\neg\alpha_k$, there is no proof tree for $\alpha_k$ whose leaves are a subset of $\Smc_k\setminus\{\beta\}$ so the database instance $\Imc$ obtained from $\Smc_k\setminus\{\beta\}$ by adding all facts entailed by $\Pi$ and $\Smc_k\setminus\{\beta\}$ is such that $\Imc\models\constraints_\Pi$ and $\alpha_k\notin\Imc$, \ie $\Imc\models \conf\setminus\{\beta\}$.

\item The database instance $\Imc$ obtained from $\Smc_k$ by adding all facts entailed by $\Pi$ and $\Smc_k$ is such that $\Imc\models\constraints_\Pi$ and $\Imc\models \conf\setminus\{\neg\alpha_k\}$.
\end{itemize}
It follows that $|\conf|\leq M$, so $|\Smc_k|< M$.

\item The number of facts over $p$ predicates that can be derived from $\Smc_k$ is bounded by $p\times (M\times a_{max})^{a_{max}}=k$ (since there are at most $M\times a_{max}$ constants in $\Smc_k$), so since we assume that $\tau$ is non-recursive, there cannot be repetition of the same fact on a path from root to leaf in $\tau$ and the height of $\tau$ is bounded by $k$, contradicting the fact that all proof trees of $\alpha_k$ have height at least $k+1$.
\end{itemize}
It follows that the size of the conflicts \wrt $\constraints_\Pi$ being bounded implies the boundedness of $\Pi$.
\end{proof}

\section{Proofs for Section \ref{sec:aics}}

\subsection{Proofs for Section \ref{subsec:aics}}

\NewFoundedGrounded*
\begin{proof}
\noindent$(\Rightarrow)$ Let $\Umc\in\groundups{\Dmc,\eta}$ and assume for a contradiction that $\Umc\notin\ups{\Dmc,\eta[\Umc]}$.  
\begin{itemize}
\item Since $\Umc\in\ups{\Dmc,\eta}$, then $\Dmc\circ\Umc\models r$ for every $r\in gr_\Dmc(\eta)$, so $\Dmc\circ\Umc\models r$ for every $r\in \eta[\Umc]$. Hence, $\Umc\notin\ups{\Dmc,\eta[\Umc]}$ means that there is a proper subset $\Umc'\subsetneq\Umc$ which is such that $\Dmc\circ\Umc'\models r$ for every $r\in\eta[\Umc]$. 

\item Since $\Umc\in\groundups{\Dmc,\eta}$, there exists $r_0\in gr_\Dmc(N(\eta))=N(gr_\Dmc(\eta))$ such that $\Dmc\circ\Umc'\not\models r_0$ and the only update action of $r_0$ is in $\Umc\setminus\Umc'$. 

\item Let $r$ be the AIC from $gr_\Dmc(\eta)$ such that $r_0\in N(r)$. Since $r$ has an update action in $\Umc\setminus\Umc'$ (hence in $\Umc$), then $r\in \eta[\Umc]$. Moreover, since $\aicbody(r)=\aicbody(r_0)$, $\Dmc\circ\Umc'\not\models r$. This contradicts the definition of $\Umc'$. 
\end{itemize}
Hence $\Umc\in\ups{\Dmc,\eta[\Umc]}$. 
\smallskip

\noindent$(\Leftarrow)$ 
Let $\Umc\in\ups{\Dmc,\eta}$ be such that $\Umc\in\ups{\Dmc,\eta[\Umc]}$ and assume for a contradiction that $\Umc\notin\groundups{\Dmc,\eta}$. 
\begin{itemize}
\item Since $\Umc$ is an r-update for $\Dmc$ \wrt $\eta$, $\Umc\notin\groundups{\Dmc,\eta}$ means that there exists $\Umc'\subsetneq\Umc$ such that for every $r\in gr_\Dmc(N(\eta))=N(gr_\Dmc(\eta))$, either $\Dmc\circ\Umc'\models r$ or the only update action of $r$ is not in $\Umc\setminus\Umc'$. 

\item Let $r\in\eta[\Umc]$: $r$ has an update action $A$ such that $A\in\Umc$. Hence there exists $r'\in N(r)$ such that $A$ is the only update action of $r'$. 
If $A\in\Umc'$, $\Dmc\circ\Umc'\models r$. 
Otherwise, $A\in\Umc\setminus\Umc'$ so $\Dmc\circ\Umc'\models r'$, which implies $\Dmc\circ\Umc'\models r$ since $\aicbody(r)=\aicbody(r')$. 

\item It follows that $\Dmc\circ\Umc'\models \eta[\Umc]$, contradicting $\Umc\in\ups{\Dmc,\eta[\Umc]}$. 
\end{itemize}
Hence $\Umc\in\groundups{\Dmc,\eta}$.
\end{proof}

\subsection{Proofs for Section \ref{subsec:prio-to-aics}}

We provide here the mentioned result showing that several repair notions coincide for monotone AICs: 

\begin{proposition}\label{prop:monotone-founded-justified}
For every monotone set $\eta$ of ground AICs and database $\Dmc$, $\justifreps{\Dmc,\eta}=$ $\groundreps{\Dmc,\eta}=$ $\foundreps{\Dmc,\eta}$ $\subseteq\wellfoundreps{\Dmc,\eta}$. 
\end{proposition}
\begin{proof}
This proposition is actually a corollary of Proposition \ref{prop:founded-grounded-closed-min-faith}. Indeed, every monotone set $\eta$ of ground AICs is consistent (since $\{\alpha\mid\neg\alpha\text{ occurs in }\eta\}\models\eta$) and such that there is no  pair of AICs $r_1, r_2\in \eta$ with $\alpha \in \aiclits(r_1)$ and $\neg \alpha \in \aiclits(r_2)$, thus is trivially closed under resolution and preserves actions under resolution. 
\end{proof}

\paragraph{Reduction from prioritized databases to ground AICs}
Recall that given a prioritized database $\Dmc^\constraints_\succ$, 
$$\eta^\constraints_\succ=\{r_\conf \mid \conf\in\conflicts{\Dmc,\constraints}\}\text{ where }r_\conf:=\bigwedge_{\lambda\in\conf}\lambda \rightarrow\{\mi{fix}(\lambda)\mid \lambda\in\conf,\forall\mu\in\conf, \lambda\not\succ\mu\}$$ and $\mi{fix}$ is such that $\mi{fix}(\alpha)=-\alpha$ and $\mi{fix}(\no\alpha)=+\alpha$. 

\ReductionPrioToAICs*
\begin{proof}
For every $\alpha\in\factset$, either $\alpha$ or $\no\alpha$ does not belong to $\litset$, hence does not belong to any conflict. It follows that $\eta^\constraints_\succ$ is monotone so by Proposition~\ref{prop:monotone-founded-justified}, $\justifreps{\Dmc,\eta^\constraints_\succ}=\groundreps{\Dmc,\eta^\constraints_\succ}=\foundreps{\Dmc,\eta^\constraints_\succ}\subseteq\wellfoundreps{\Dmc,\eta^\constraints_\succ}$. 
We show that $\deltapreps{\Dmc^\constraints_\succ}=\foundreps{\Dmc,\eta^\constraints_\succ}$. 
\smallskip

\noindent$(\foundreps{\Dmc,\eta^\constraints_\succ}\subseteq\deltapreps{\Dmc^\constraints_\succ})$ Let $\Umc\in\foundups{\Dmc,\eta^\constraints_\succ}$ and $\Rmc=\Dmc\circ\Umc$. We show that $\Rmc\in\deltapreps{\Dmc^\constraints_\succ}$. 

\begin{itemize}
\item Since $\Umc\in\ups{\Dmc,\eta^\constraints_\succ}$, $\Dmc\circ\Umc\models\eta^\constraints_\succ$, \ie $\Dmc\circ\Umc\models r_\conf$ for every $\conf\in\conflicts{\Dmc,\constraints}$. Hence $\Rmc\not\models\conf$, so  
$\conf\not\subseteq\comp{\Dmc}{\Rmc}$. 
Since $\Umc$ is a minimal such set of update actions, $\comp{\Dmc}{\Rmc}$ is a maximal subset of $\litset$ such that there is no conflict $\conf\in\conflicts{\Dmc,\constraints}$ such that $\conf\subseteq\comp{\Dmc}{\Rmc}$. 
By Proposition~\ref{prop:characterizations-repairs-comp}, $\Rmc\in\deltareps{\Dmc,\constraints}$.

\item Assume for a contradiction that $\Rmc\notin\deltapreps{\Dmc^\constraints_\succ}$. There exists $\Rmc'$ such that $\Rmc'\models\constraints$ and there is $\lambda\in\comp{\Dmc}{\Rmc'}\setminus\comp{\Dmc}{\Rmc}$ with $\lambda\succ\mu$ for every $\mu\in \comp{\Dmc}{\Rmc}\setminus\comp{\Dmc}{\Rmc'}$. 
If $\lambda=\alpha\in\Dmc$, $\alpha\notin\comp{\Dmc}{\Rmc}$ means that $\alpha\notin\Rmc$ so $\mi{fix}(\lambda)=-\alpha$ is in $\Umc$. If $\lambda=\no\alpha$ for some $\alpha\notin\Dmc$, $\no\alpha\notin\comp{\Dmc}{\Rmc}$ means that $\alpha\in\Rmc$ so $\mi{fix}(\lambda)=+\alpha$ is in $\Umc$. 

Since $\Umc$ is founded and $\mi{fix}(\lambda)\in\Umc$, there exists $r_\conf\in\eta^\constraints_\succ$ (that corresponds to $\conf\in\deltaconflicts{\Dmc,\constraints}$) such that $\mi{fix}(\lambda)$ is an update action of $r_\conf$ and $\Dmc\circ\Umc\setminus\{\mi{fix}(\lambda)\}\not\models r_\conf$, \ie $\Rmc_\lambda\models  \conf$ where $\Rmc_\lambda=\Dmc\circ\Umc\setminus\{\mi{fix}(\lambda)\}$.  
Hence $\conf\subseteq\comp{\Dmc}{\Rmc_\lambda}$ and for every $\mu\in\conf$, $\lambda\not\succ\mu$ by definition of the update actions of $r_\conf$. 

It follows that $\conf\subseteq\comp{\Dmc}{\Rmc_\lambda}\setminus\{\mu\mid\lambda\succ\mu\}=\comp{\Dmc}{\Rmc}\cup\{\lambda\}\setminus\{\mu\mid\lambda\succ\mu\}$.  

Since $\lambda\in\comp{\Dmc}{\Rmc'}$ and $\comp{\Dmc}{\Rmc}\setminus\comp{\Dmc}{\Rmc'}\subseteq \{\mu\mid\lambda\succ\mu\}$, then $\comp{\Dmc}{\Rmc}\cup\{\lambda\}\setminus\{\mu\mid\lambda\succ\mu\}\subseteq\comp{\Dmc}{\Rmc'}$. 
It follows that $\conf\subseteq\comp{\Dmc}{\Rmc'}$, which contradicts $\Rmc'\models\constraints$. 
Hence $\Rmc\in\deltapreps{\Dmc^\constraints_\succ}$.
\end{itemize}

\noindent$(\deltapreps{\Dmc^\constraints_\succ}\subseteq\foundreps{\Dmc,\eta^\constraints_\succ})$ Let $\Rmc\in\deltapreps{\Dmc^\constraints_\succ}$ and let $\Umc=\{-\alpha\mid \alpha\in\Dmc\setminus\Rmc\}\cup\{+\alpha\mid \alpha\in\Rmc\setminus\Dmc\}$ be the consistent set of update actions such that $\Dmc\circ\Umc=\Rmc$. 
We show that $\Umc\in\foundups{\Dmc,\eta^\constraints_\succ}$. 
\begin{itemize}
\item Since $\Rmc\in\deltareps{\Dmc,\constraints}$, 
by Proposition~\ref{prop:characterizations-repairs-comp}, $\comp{\Dmc}{\Rmc}$ is a maximal subset of $\litset$ such that there is no conflict $\conf\in\conflicts{\Dmc,\constraints}$ such that $\conf\subseteq\comp{\Dmc}{\Rmc}$. 
In particular, for every $\conf\in\conflicts{\Dmc,\constraints}$, $\Rmc\not\models \conf$ so for every $r_\conf\in\eta^\constraints_\succ$, $\Dmc\circ\Umc\models r_\conf$. Thus $\Dmc\circ\Umc\models\eta^\constraints_\succ$. 
The minimality of $\Umc$ follows from the maximality of $\comp{\Dmc}{\Rmc}$. Hence $\Umc\in\ups{\Dmc,\eta^\constraints_\succ}$.

\item Let $A\in\Umc$. If $A=-\alpha$ for some $\alpha\in\Dmc\setminus\Rmc$, let $\lambda=\alpha$, 
and if $A=+\alpha$ for some $\alpha\in\Rmc\setminus\Dmc$, let $\lambda=\no\alpha$. 
Let $\Rmc_\lambda=\restr{\Dmc}{\comp{\Dmc}{\Rmc}\cup\{\lambda\}\setminus\{\mu\mid\lambda\succ\mu\}}$, so that $\comp{\Dmc}{\Rmc_\lambda}=\comp{\Dmc}{\Rmc}\cup\{\lambda\}\setminus\{\mu\mid\lambda\succ\mu\}$ by Lemma~\ref{lem:compl-restr}. 
By Lemma~\ref{lem:pareto-improvement-if-one-beat-all}, there must be some $\conf\in\deltaconflicts{\Dmc,\constraints}$ such that $\conf\subseteq\comp{\Dmc}{\Rmc_\lambda}$ (otherwise $\Rmc\notin\deltapreps{\Dmc^\constraints_\succ}$).

For every $\mu\in\conf$, $\lambda\not\succ\mu$ (otherwise, $\mu\notin\comp{\Dmc}{\Rmc_\lambda}$ so $\conf\not\subseteq\comp{\Dmc}{\Rmc_\lambda}$). 
Hence, $A=\mi{fix}(\lambda)$ is an update action of the AIC $r_\conf\in\eta^\constraints_\succ$ that corresponds to $\conf$. 

Moreover, $\comp{\Dmc}{\Dmc\circ\Umc\setminus\{A\}}=\comp{\Dmc}{\Rmc}\cup\{\lambda\}\supseteq\comp{\Dmc}{\Rmc_\lambda}$, so since $\conf\subseteq\comp{\Dmc}{\Rmc_\lambda}$, then $\conf\subseteq\comp{\Dmc}{\Dmc\circ\Umc\setminus\{A\}}$. Thus $\Dmc\circ\Umc\setminus\{A\}\not\models r_\conf$. 
Hence $\Umc\in\foundups{\Dmc,\eta^\constraints_\succ}$ .\qedhere
\end{itemize}
\end{proof}

\paragraph{Data-independent reduction in denial constraints case}
\def\const{\mathsf{const}}
\def\refine{\textit{refine}}
\def\bodyvars{\mathsf{vars}}
\def\terms{\mathsf{terms}}

Recall that for the next result, we assume that the 
priority relation $\succ$ is specified in the database. 
Concretely, we add a fresh predicate $P_\succ$ to $\preds$,
and increase the arity of the predicates in $\preds\setminus\{P_\succ\}$ by 1, so that the 
first argument of each $R \in \preds\setminus\{P_\succ\}$ now stores a unique fact identifier,
while $P_\succ$ stores pairs of such identifiers. 

We now take a set of denial constraints $\constraints$ over $\preds\setminus\{P_\succ\}$
and explain how to build the new set of constraints $\minnongr(\constraints)$ that allow us to more easily 
identify conflicts. The construction of $\minnongr(\constraints)$ is done in two steps. 

First, we shall transform every $\tau \in \constraints$ into a set of more specific constraints. It will be
convenient here to have some notations for referring to parts of (sets of) denial constraints: 
we will use $\aicbody(\tau)$ and $\aiclits(\tau)$ for the body and set of literals of $\tau$,
$\bodyvars(\tau)$ and $\terms(\tau)$ for the variables and terms occurring in 
$\tau$, and $\const(\constraints)$ for the constants occurring in $\Cmc$. 
Then the set $\refine(\tau)$ of \emph{refinements} of $\tau$ contains 
all denial constraints that can be obtained 
from $\tau$ by applying the following operations in order:
\begin{enumerate}
\item choose some partition $T_1, \ldots, T_p$ of $\bodyvars(\tau) \cup \const(\constraints)$ such that each set $T_i$ 
in the partition contains at most one constant 
\item for each $T_i$ that contains a constant $c$, replace all occurrences of variables in $T_i$ with $c$
\item for each $T_i$ that does not contain any constant, choose some variable $v \in T_i$ and replace all occurrences of variables in $T_i \setminus \{v\}$
by $v$
\item for every pair of distinct variables $v,z$ in the modified $\tau$, add the inequality atom $v \neq z$
\item for every variable $v$ in the modified $\tau$, and every $c \in \const(\constraints)$, add the inequality atom $v \neq c$
\end{enumerate}
We then let $\refine(\constraints)= \bigcup_{\tau \in \Cmc} \refine(\tau)$. 
By construction, for every constraint $\tau' \in \refine(\constraints)$, there is an inequality atom between every 
pair of distinct variables in $\bodyvars(\tau')$ and between every variable in $\bodyvars(\tau')$ and every constant in $\const(\constraints)$.
It follows that if $\Dmc \models \aicbody(\tau')$, then there is an \emph{injective} homomorphism $h: \terms(\tau') \rightarrow \domain{\Dmc}$
such that $h(c) =c$ for constants $c \in \terms(\tau')$ and $P(h(t_1), \ldots, h(t_k)) \in \Dmc$ for every $P(t_1, \ldots, t_k) \in \aiclits(\tau')$. 
The \emph{image} of $\aicbody(\tau')$ under $h$ on $\Dmc$, denoted $h(\tau')$, is the set of facts 
$P(h(t_1), \ldots, h(t_k)) \in \Dmc$ such that $P(t_1, \ldots, t_k) \in \aiclits(\tau')$. Due to the injectivity of $h$, 
$h(\tau')$ is isomorphic to $\aiclits(\tau')$. 

We say that a constraint $\tau_1 \in \refine(\constraints)$ is \emph{subsumed} by another constraint  $\tau_2 \in \refine(\constraints)$ if there is an injective function $h: \terms(\tau_2) \rightarrow \terms(\tau_1)$ such that $h(c)=c$ for all constants in $\terms(\tau_2)$
and $h(\aiclits(\tau_2)) \subsetneq \aiclits(\tau_1)$. For example, $R(x,x) \wedge A(x) \rightarrow \bot$ is subsumed by $R(x,x) \rightarrow \bot$. 
The set $\minnongr(\constraints)$ contains precisely those constraints from $\refine(\constraints)$ which are not subsumed by
any other constraint in $\refine(\constraints)$. The following lemma resumes the key properties of $\minnongr(\constraints)$. 

\begin{lemma}\label{lem:data-indep-reduction}
Let $\constraints$ be a set of denial constraints  over $\preds\setminus\{P_\succ\}$. 
Then for any database $\Dmc$ over $\preds$, the following are equivalent:
\begin{enumerate}
\item  $\Emc \in \conflicts{\Dmc,\constraints}$ 
\item  $\Emc\in \conflicts{\Dmc,\minnongr(\constraints)}$ 
\item there exists an injective homomorphism 
$h: \terms(\tau') \rightarrow \domain{\Dmc}$ such that $\Emc = h(\tau')$ for some $\tau' \in \minnongr(\constraints)$ and $\conf\subseteq\Dmc$. 
\end{enumerate}
\end{lemma}
\begin{proof}
\begin{itemize}
\item We first show that $\conflicts{\Dmc,\constraints}=\conflicts{\Dmc,\minnongr(\constraints)}$, by showing that for every $\conf\subseteq\Dmc$, $\conf\not\models\constraints$ iff $\conf\not\models\minnongr(\constraints)$. 

If $\conf\not\models\minnongr(\constraints)$, there is $\tau'\in\minnongr(\constraints)$ such that $\conf\models\aicbody(\tau')$, \ie there is a homomorphism $h$ from $\aicbody(\tau')$ to $\conf$. Let $\tau\in\constraints$ be such that $\tau'\in\refine(\tau)$ and $g: \terms(\tau) \rightarrow \terms(\tau')$ be such that $g(t)$ is equal to the variable or constant that has been chosen to represent $T_i$ such that $t\in T_i$ in the refinement sequence from $\tau$ to $\tau'$. Then $h\circ g$ is a homomorphism from $\aicbody(\tau)$ to $\conf$ so that $\conf\not\models \tau$ and $\conf\not\models\constraints$.

If $\conf\not\models\constraints$, there is $\tau\in\constraints$ such that $\conf\not\models\tau$, \ie there is a homomorphism $h$ from $\aicbody(\tau)$ to $\conf$. Let $\tau'$ be the constraint obtained from $\tau$ by the  refinement sequence defined by the partition $T_1,\dots, T_p$ such that two terms $t_1,t_2$ of $\tau$ are in the same $T_i$ iff $h( t_1)=h(t_2)$. The restriction of $h$ to $\terms(\tau')$ is a homomorphism  from $\aicbody(\tau')$ to $\conf$ so that $\conf\not\models \tau'$ and $\conf\not\models\minnongr(\constraints)$.

\item Let $\conf\in\conflicts{\Dmc,\minnongr(\constraints)}$. 
There is some $\tau\in\minnongr(\constraints)$ such that $\conf\not\models\tau$, \ie $\conf\models \aicbody(\tau)$. Since there is an inequality atom between every 
pair of distinct variables in $\bodyvars(\tau)$ and between every variable in $\bodyvars(\tau)$ and every constant in $\const(\constraints)$, it follows that there is an injective homomorphism $h: \terms(\tau') \rightarrow \domain{\conf}\subseteq\domain{\Dmc}$. 
Moreover, since for every proper subset $\conf'$ of $\conf$, $\conf'\not\models \aicbody(\tau)$, it follows that $\conf=h(\tau')$.

\item Assume that there exists an injective homomorphism 
$h: \terms(\tau) \rightarrow \domain{\Dmc}$ such that $\Emc = h(\tau)$ for some $\tau \in \minnongr(\constraints)$ and $\conf\subseteq\Dmc$. 
Since $\conf\models h(\tau)$, $\conf\not\models\tau$ so $\conf\not\models\minnongr(\constraints)$. 
Assume for a contradiction that there exists $\conf'\subsetneq\conf$ such that $\conf'\not\models\minnongr(\constraints)$. There is $\tau'\in\minnongr(\constraints)$ such that $\conf'\not\models\tau'$, so there is an injective homomorphism $h'$ from $\tau'$ to $\conf'$. Since $h'(\tau')\subseteq\conf'\subsetneq\conf$ is isomorphic to $\aiclits(\tau')$, and $\conf=h(\tau)$ is isomorphic to $\aiclits(\tau)$, it follows that $\tau\notin\minnongr(\constraints)$. Hence there does not exist $\conf'\subsetneq\conf$ such that $\conf'\not\models\minnongr(\constraints)$ and $\conf\in \conflicts{\Dmc,\minnongr(\constraints)}$.\qedhere
\end{itemize}

\end{proof}

We then define $\eta^\constraints$ as the set of all AICs
$$r_{\tau,i}:=\bigl(\ell_1 \wedge \ldots \wedge \ell_n \wedge \varepsilon \wedge   \bigwedge_{\ell_j \neq \ell_i}   \neg P_\succ (id_i,id_j)\bigr)  \rightarrow \{-\ell_i\}.$$
such that $\tau=(\ell_1 \wedge \ldots \wedge \ell_n \wedge \varepsilon \rightarrow \bot)$ is in $\minnongr(\constraints)$, $i \in \{1, \ldots, n\}$, 
and for every $1 \leq k \leq n$, 
$\ell_k=R(id_k,\vec{t})$ for some $R,\vec{t}$. 

\ReductionDenialPrioToAICs*
\begin{proof}
First note that $\eta^\constraints$ is monotone: All negative literals in $\eta^\constraints$ have predicate $P_\succ$ which does not occur in $\constraints$, hence does not occur in the positive literals in $\eta^\constraints$.  Hence for every database $\Dmc$, $gr_\Dmc(\eta^\constraints)$ is monotone and by Proposition~\ref{prop:monotone-founded-justified}, $\justifreps{\Dmc,gr_\Dmc(\eta^\constraints)}=\groundreps{\Dmc,gr_\Dmc(\eta^\constraints)}=\foundreps{\Dmc,gr_\Dmc(\eta^\constraints)}\subseteq\wellfoundreps{\Dmc,gr_\Dmc(\eta^\constraints)}$. It follows that $\justifreps{\Dmc,\eta^\constraints}=\groundreps{\Dmc,\eta^\constraints}=\foundreps{\Dmc,\eta^\constraints}\subseteq\wellfoundreps{\Dmc,\eta^\constraints}$. 

We show that for every prioritized database $\Dmc^\constraints_\succ$, $\deltapreps{\Dmc^\constraints_\succ}=\foundreps{\Dmc,\eta^\constraints}$. Let $\Bmc_\succ\subseteq\Dmc$ be the set of facts from $\Dmc$ whose predicate is $P_\succ$ and $P_\succ^\Dmc$ the set of all possible facts on predicate $\Dmc$ and domain $\domain{\Dmc}$ (\ie the set of facts from $\factset$ with predicate $P_\succ$). 
Note that for every $\tau\in\constraints$ and $r_{\tau,i}\in\eta^\constraints$ built from $\tau$, $\terms(\tau)=\terms(r_{\tau,i})$. 
\smallskip

\noindent$(\foundreps{\Dmc,\eta^\constraints}\subseteq\deltapreps{\Dmc^\constraints_\succ})$ Let $\Umc\in\foundups{\Dmc,\eta^\constraints}$ and $\Rmc=\Dmc\circ\Umc$. We show that $\Rmc\in\deltapreps{\Dmc^\constraints_\succ}$. 

\begin{itemize}
\item 
Let $\conf\in\conflicts{\Dmc,\constraints}$. 
By Lemma \ref{lem:data-indep-reduction}, there exists an injective homomorphism 
$h: \terms(\tau) \rightarrow \domain{\Dmc}$ such that $\Emc = h(\tau)$ for some $\tau \in \minnongr(\constraints)$.  
Moreover, since $\succ$ is acyclic, there exists $\alpha\in\conf$ such that for every $\beta\in\conf$, $\alpha\not\succ\beta$. 

Let $\ell_i$ be the literal of $\tau$ such that $h(\ell_i)=\alpha$. 
Then $\conf\cup\{\no\beta\mid \beta\in P_\succ^\Dmc, \beta\notin\Bmc_\succ\}\models \aicbody(r_{\tau,i})$.

Since $\Umc\in\ups{\Dmc,\eta^\constraints}$, $\Dmc\circ\Umc\models\eta^\constraints$. In particular, $\Dmc\circ\Umc\models r_{\tau,i}$ so $\Dmc\circ\Umc\not\models \conf\cup\{\no\beta\mid \beta\in P_\succ^\Dmc, \beta\notin\Bmc_\succ\}$. 

Since $\Umc$ is founded, $\Umc$ does not contain any update action with predicate $P_\succ$ since update atoms on predicate $P_\succ$ do not occur in $\eta^\constraints$. It follows that $\conf\not\subseteq\Dmc\circ\Umc$. 

Hence $\Rmc$ does not contain any conflict in $\conflicts{\Dmc,\constraints}$, so $\Rmc\models\constraints$.

\item Assume for a contradiction that $\Rmc\notin\deltapreps{\Dmc^\constraints_\succ}$. There exists $\Rmc'\subseteq\Dmc$ such that $\Rmc'\models\constraints$ and there is $\alpha\in\Rmc'\setminus\Rmc$ with $\alpha\succ\beta$ for every $\beta\in \Rmc\setminus\Rmc'$. 
Since $\alpha\notin\Rmc$, then $-\alpha\in\Umc$. 
Since $\Umc$ is founded, it follows that there exists $r_{\tau,i}=\bigl(\ell_1 \wedge \ldots \wedge \ell_n \wedge \varepsilon \wedge   \bigwedge_{\ell_j \neq \ell_i}   \neg P_\succ (id_i,id_j)\bigr)  \rightarrow \{-\ell_i\}$ in $\eta^\constraints$ such that $-\alpha$ is an update action of some $r_g\in gr_\Dmc(r_{\tau,i})$ and $\Dmc\circ\Umc\setminus\{-\alpha\}\not\models r_g$, \ie $\Rmc\cup\{\alpha\}\not\models r_g$. 

Let $h: \terms(\tau) \rightarrow \domain{\Dmc}$ be such that $\aiclits(r_g)=h(r_{\tau,i})$. 
Since there is an inequality atom between every 
pair of distinct variables in $\bodyvars(\tau)$ and between every variable in $\bodyvars(\tau)$ and every constant in $\const(\constraints)$ in $\tau$, hence in $r_{\tau,i}$, it follows that $h$ is injective. 
Since $\Rmc\cup\{\alpha\}\not\models r_g$, the set of positive literals of $r_g$ is included in $\Rmc\cup\{\alpha\}$, \ie $ h(\tau)\subseteq \Rmc\cup\{\alpha\}\subseteq\Dmc$. 
By Lemma~\ref{lem:data-indep-reduction}, $\conf=h(\tau)$ is in $\deltaconflicts{\Dmc,\constraints}$.

Moreover, by definition of $h$, $h(\ell_i)=\alpha$ so $\Rmc\cup\{\alpha\}\not\models r_g$ implies that for every $\beta\in\conf$, $\alpha\not\succ\beta$.

Hence there exists $\conf\in\deltaconflicts{\Dmc,\constraints}$ such that $\conf\subseteq\Rmc\cup\{\alpha\}$ and for every $\beta\in\conf$, $\alpha\not\succ\beta$. 
It follows that $\conf\subseteq\Rmc'$, which contradicts $\Rmc'\models\constraints$. 
Hence $\Rmc\in\deltapreps{\Dmc^\constraints_\succ}$.
\end{itemize}

\noindent$(\deltapreps{\Dmc^\constraints_\succ}\subseteq\foundreps{\Dmc,\eta^\constraints})$ Let $\Rmc\in\deltapreps{\Dmc^\constraints_\succ}$ and let $\Umc=\{-\alpha\mid \alpha\in\Dmc\setminus\Rmc\}$ 
be the consistent set of update actions such that $\Dmc\circ\Umc=\Rmc$. We show that $\Umc\in\foundups{\Dmc,\eta^\constraints}$. 

\begin{itemize}
\item Since $\Rmc\in\deltareps{\Dmc,\constraints}$, $\Rmc$ is a maximal subset of $\Dmc$ such that $\Rmc\models\constraints$, hence a maximal subset of $\Dmc$ such that $\Rmc\models\tau$ for every $\tau\in\minnongr(\constraints)$ (\cf proof of Lemma~\ref{lem:data-indep-reduction}). 
It follows from the construction of $\eta^\constraints$ that $\Rmc\models r_{\tau,i}$ for every $\tau\in\minnongr(\constraints)$ and $\ell_i$ literal of $\tau$ (since the body of $r_{\tau,i}$ extends that of $\tau$). Thus $\Dmc\circ\Umc\models\eta^\constraints$.

Since $\Rmc\in\deltareps{\Dmc,\constraints}$, $\Rmc$ is a maximal subset of $\Dmc$ that does not contain any $\conf\in\conflicts{\Dmc,\constraints}$. It follows that for every proper subset $\Umc'\subsetneq\Umc$, $\Dmc\circ\Umc'$ contains some $\conf\in\conflicts{\Dmc,\constraints}$.
By Lemma \ref{lem:data-indep-reduction} there exists an injective homomorphism $h: \terms(\tau) \rightarrow \domain{\Dmc}$ such that $\conf = h(\tau)$ for some $\tau\in \minnongr(\constraints)$, 
and by acyclicity of $\succ$, there exists $\alpha\in\conf$, which is such that $\alpha=h(\ell_i)$ for some literal $\ell_i$ of $\tau$, such that for every $\beta\in\conf$, $\alpha\not\succ\beta$. 
Hence  $\conf\cup\{\no\beta\mid \beta\in P_\succ^\Dmc, \beta\notin\Bmc_\succ\}\models \aicbody(r_{\tau,i})$, so since $\Umc$ hence $\Umc'$ does not add any $P_\succ$ fact, $\Dmc\circ\Umc'\not\models r_{\tau,i}$.  

It follows that $\Umc$ is a minimal set of update actions such that $\Dmc\circ\Umc\models\eta^\constraints$. Hence $\Umc\in\ups{\Dmc,\eta^\constraints}$.

\item Let $A\in\Umc$, $A=-\alpha$ for some $\alpha\in\Dmc\setminus\Rmc$. 
By maximality of $\Rmc\in\deltareps{\Dmc,\constraints}$, there exists $\conf\in\deltaconflicts{\Dmc,\constraints}$ such that $\conf\subseteq\Rmc\cup\{\alpha\}$. 
Moreover, since $\Rmc\in\deltapreps{\Dmc,\constraints}$, there exists such $\conf$ such that for every $\beta\in\conf$, $\alpha\not\succ\beta$ (otherwise, $\Rmc'=\Rmc\cup\{\alpha\}\setminus\{\beta\mid\alpha\succ\beta\}$ would be a Pareto-improvement of $\Rmc$). 

Since $\conf\in\deltaconflicts{\Dmc,\constraints}$, by Lemma \ref{lem:data-indep-reduction} there exists an injective homomorphism $h: \terms(\tau) \rightarrow \domain{\Dmc}$ such that $\conf = h(\tau)$ for some $\tau\in \minnongr(\constraints)$. 
Let $\ell_i$ be the literal of $\tau$ such that $h(\ell_i)=\alpha$. 
Then $\conf\cup\{\no\beta\mid \beta\in P_\succ^\Dmc, \beta\notin\Bmc_\succ\}\models \aicbody(r_{\tau,i})$. 

Moreover, $-\ell_i$ is the update action of $r_{\tau,i}$ so $A=-\alpha=h(\ell_i)$ is the update action of the ground AIC $r_\conf\in gr_\Dmc(r_{\tau,i})$ whose body is $h(\aicbody(r_{\tau,i}))$. 

Finally, $\Dmc\circ\Umc\setminus\{A\}=\Rmc\cup\{\alpha\}$, so since $\conf\subseteq\Rmc\cup\{\alpha\}$ and $\Umc$ does not add any $P_\succ$ fact, then $\Dmc\circ\Umc\setminus\{A\}\not\models r_\conf$. 
Hence $\Umc\in\foundups{\Dmc,\eta^\constraints}$.\qedhere
\end{itemize}
\end{proof}

\subsection{Proofs for Section \ref{subsec:well-behaved}}

The following lemmas are useful to prove Propositions \ref{prop:founded-grounded-closed-min-faith}, \ref{prop:reduction-AICs-prio-binary} and \ref{prop:AICs-prio-non-binary}. 

\begin{lemma}\label{lem:update-repair-closed}
If $\Dmc\circ\Umc\models\eta$, then $ ne(\Dmc,\Dmc\circ\Umc)\cup\Umc$ is closed under $\eta$.
\end{lemma}
\begin{proof}
Let $r\in gr_\Dmc(\eta)$ be such that $ ne(\Dmc,\Dmc\circ\Umc)\cup\Umc$ satisfies all non-updatable literals of $r$. 
By definition of $ne(\Dmc,\Dmc\circ\Umc)$, this implies that $\Dmc\circ\Umc$ satisfies all non-updatable literals of $r$. 
Since $\Dmc\circ\Umc\models r$, it follows that $\Dmc\circ\Umc$ does not satisfy some updatable literal $\ell$ of $r$. 
 If $\ell$ is positive, $\ell=\beta$, then since $\ell$ is updatable $-\beta\in\aicup(r)$. Since $\Dmc\circ\Umc$ does not satisfy $\ell$, $\beta\notin\Dmc\circ\Umc$. If $\beta\notin\Dmc$, then $-\beta\in ne(\Dmc,\Dmc\circ\Umc)$. Otherwise $-\beta\in\Umc$. In both cases, $-\beta\in  ne(\Dmc,\Dmc\circ\Umc)\cup\Umc$.  If $\ell$ is negative, $\ell=\neg\beta$, then since $\ell$ is updatable $+\beta\in\aicup(r)$. Since $\Dmc\circ\Umc$ does not satisfy $\ell$, $\beta\in\Dmc\circ\Umc$. If $\beta\in\Dmc$, then $+\beta\in ne(\Dmc,\Dmc\circ\Umc)$. Otherwise $+\beta\in\Umc$. In both cases, $+\beta\in  ne(\Dmc,\Dmc\circ\Umc)\cup\Umc$. 
Hence $ ne(\Dmc,\Dmc\circ\Umc)\cup\Umc$ is closed under $\eta$.
\end{proof}

\begin{lemma}\label{lem:closed-under-res-min-conflicts}
If $\eta$ is closed under resolution, then for every database $\Dmc$, $\deltaconflicts{\Dmc,\constraints_\eta}=\{\aiclits(r)\mid r\in \mingr(\eta), \Dmc\not\models r\}$ where $\constraints_\eta$ is the set of universal constraints that correspond to AICs of $\eta$ and $\mingr(\eta)$ is the subset of $gr_\Dmc(\eta)$  that contains AICs whose bodies are subset-minimal.
\end{lemma}
\begin{proof}
Observe that $\Dmc\not\models r$ iff $\Dmc\models \aicbody(r)$ iff $\aiclits(r)\subseteq\litset$. 
Hence by Proposition~\ref{prop:defconflicts}, it is sufficient to show that $\{ \aiclits(r)\mid r\in \mingr(\eta)\} = \{\aiclits(r) \mid \aicbody(r)\text{ is a prime implicant of} \bigvee_{\varphi\rightarrow \bot\in gr_\Dmc(\constraints_\eta)} \varphi\}$. 

Since $\eta$ is closed under resolution, for every implicant $\psi$ of $\bigvee_{\varphi\rightarrow \bot\in gr_\Dmc(\constraints_\eta)} \varphi$, there is an AIC $r\in gr_\Dmc(\eta)$ such that $\aicbody(r)=\psi$. Since $\mingr(\eta)$ retains only the body minimal AICs in $gr_\Dmc(\eta)$, the result follows. 
\end{proof}

\begin{lemma}\label{lem:symdifminhittingset}
For every set of constraints $\constraints$, database $\Dmc$, and $\Bmc\in\deltareps{\Dmc,\constraints}$, $\Bmc\Delta\Dmc$ is a minimal hitting set of $\mhs{\Dmc,\constraints}$, where $\mhs{\Dmc,\constraints}$ is the set of all minimal hitting sets of $\{\Rmc\Delta\Dmc\mid \Rmc \in\deltareps{\Dmc,\constraints}\}$. 
\end{lemma}
\begin{proof}
Let $\Bmc\in\deltareps{\Dmc,\constraints}$. 
For every $\Hmc\in\mhs{\Dmc,\constraints}$, since $\Hmc$ is a hitting set of $\{\Rmc\Delta\Dmc\mid \Rmc \in\deltareps{\Dmc,\constraints}\}$, $(\Bmc\Delta\Dmc)\cap\Hmc\neq\emptyset$ so $\Bmc\Delta\Dmc$ is a  hitting set of $\mhs{\Dmc,\constraints}$. 

Let $\Bmc'$ be a database such that $\Bmc'\Delta\Dmc\subsetneq\Bmc\Delta\Dmc$. 
By Lemma \ref{lem:compl-delta-1}, $\comp{\Dmc}{\Bmc}\subsetneq\comp{\Dmc}{\Bmc'}$. 
Moreover, by Proposition~\ref{prop:characterizations-repairs-comp}, $\comp{\Dmc}{\Bmc}$ is a maximal subset of $\litset$ that does not include any conflicts of $\conflicts{\Dmc,\constraints}$. Hence there exists $\conf\in\conflicts{\Dmc,\constraints}$ such that $\conf\subseteq\comp{\Dmc}{\Bmc'}$. 
By Proposition~\ref{prop:defconflicts}, $\conflicts{\Dmc,\constraints}=
\{ \Hmc\cap\Dmc\cup\{\no\alpha \mid \alpha\in\Hmc\setminus\Dmc\} \mid \Hmc\in \mhs{\Dmc,\constraints} \}$ so there is some $\Hmc\in\mhs{\Dmc,\constraints}$ that corresponds to $\conf$. 
It is easy to check that $(\Bmc'\Delta\Dmc)\cap\Hmc=\emptyset$: otherwise, if $\alpha\in(\Bmc'\Delta\Dmc)\cap\Hmc$, then either $\alpha\in\Dmc$, $\alpha\notin\Bmc'$ and $\alpha\in\conf$, or  $\alpha\notin\Dmc$, $\alpha\in\Bmc'$ and $\no\alpha\in\conf$, and in both cases we would get $\Bmc'\not\models\conf$, contradicting $\conf\subseteq\comp{\Dmc}{\Bmc'}$. 
It follows that $\Bmc\Delta\Dmc$ is a minimal hitting set of $\mhs{\Dmc,\constraints}$. 
\end{proof}

\begin{lemma}\label{lem:founded-conf-closed-min-faith}
If $\eta$ is closed under resolution and preserves actions under resolution, then for every database $\Dmc$ and $\Umc\in\foundups{\Dmc,\eta}$, 
for every $A\in\Umc$, there exists $r_A\in gr_\Dmc(\eta)$ such that $A\in\aicup(r_A)$, $\Dmc\circ\Umc\setminus\{A\}\not\models r_A$ and $\Dmc\not\models r_A$.
\end{lemma}
\begin{proof}
We denote by $\constraints_\eta$ the set of universal constraints that corresponds to $\eta$. 
Let $\Umc\in\foundups{\Dmc,\eta}$. 
Let $A\in\Umc$ and let $\ell_A$ be the literal that is set to true by $A$.
\begin{itemize}
\item Since $\Umc$ is founded and $A\in\Umc$, there exists $r\in gr_\Dmc(\eta)$ such that $A\in\aicup(r)$ (hence $ \overline{\ell_A}$ is a literal of $r$) and $\Dmc\circ\Umc\setminus\{A\}\not\models r$. 

\item If $\Dmc\not\models r$, $r_A=r$ is as required. 
We next consider the case where $\Dmc\models r$. 

\item Since $\Dmc\models r$ while $\Dmc\circ\Umc\setminus\{A\}\not\models r$, then $\Dmc$ does not satisfy all literals of $r$  while $\Dmc\circ\Umc\setminus\{A\}$ does. 
Let $\ell_1,\dots, \ell_n$ be the literals of $r$ that are not satisfied by $\Dmc$. Since $\ell_1,\dots, \ell_n$ are satisfied by $\Dmc\circ\Umc\setminus\{A\}$, there are $B_1,\dots, B_n\in\Umc\setminus\{A\}$ that set $\ell_1,\dots, \ell_n$ to true respectively.

\item Since $\Umc\in\ups{\Dmc,\eta}$, $\Dmc\circ\Umc\in\deltareps{\Dmc,\constraints_\eta}$ so by Lemma \ref{lem:symdifminhittingset}, $(\Dmc\circ\Umc)\Delta\Dmc=\{\alpha\mid -\alpha\text{ or } +\alpha\in\Umc\}$ is a minimal hitting set of $\mhs{\Dmc,\constraints_\eta}$. 
Since all update actions in $\Umc$ modify $\Dmc$, $\{\alpha\mid-\alpha\in\Umc \}\cup\{\no\alpha\mid+\alpha\in\Umc\}=\{\lambda\mid\mi{fix}(\lambda)\in\Umc\}$ is thus a minimal hitting set of $\deltaconflicts{\Dmc,\constraints_\eta}=
\{ \Hmc\cap\Dmc\cup\{\no\alpha \mid \alpha\in\Hmc\setminus\Dmc\} \mid \Hmc\in \mhs{\Dmc,\constraints_\eta} \}$.

Hence, for every $1\leq i\leq n$, there exists $\conf_i\in\deltaconflicts{\Dmc,\constraints_\eta}$ such that $\Umc\cap\{\mi{fix}(\lambda)\mid \lambda\in\conf_i\}=\{B_i\}$ (otherwise $\{\lambda\mid\mi{fix}(\lambda)\in\Umc\setminus\{B_i\}\}$ would be a smaller hitting set of $\deltaconflicts{\Dmc,\constraints_\eta}$).

\item Since $\eta$ is closed under resolution, by Lemma~\ref{lem:closed-under-res-min-conflicts}, $\deltaconflicts{\Dmc,\constraints_\eta}=\{\aiclits(r)\mid r\in \mingr(\eta), \Dmc\not\models r\}$ where $\mingr(\eta)$ is the subset of $gr_\Dmc(\eta)$ that contains AICs whose bodies are subset-minimal. 
Hence, for every $1\leq i\leq n$, there exists $r_i\in \mingr(\eta)$ such that $\conf_i=\aiclits(r_i)$ and $\Dmc\not\models r_i$. 
Moreover, since $\Umc\cap\{\mi{fix}(\lambda)\mid \lambda\in\conf_i\}=\{B_i\}$ and $B_i$ sets $\ell_i$ to true, it follows that $\overline{\ell_i}$ is a literal of $r_i$ .

\item Since $\Dmc\not\models r_i$ for every $1\leq i\leq n$, $\Dmc$ satisfies all literals of $r_1,\dots,r_n$. Thus there is no literal $\ell$ such that both $\ell$ and $\overline{\ell}$ occur in some $r_i, r_j$, $1\leq i,j\leq n$. 

Moreover, by definition of $\ell_1,\dots,\ell_n$, all literals of $r$ that are not among $\ell_1,\dots,\ell_n$ are satisfied by $\Dmc$. Thus there is no literal $\ell$ different from $\ell_1,\dots,\ell_n$ that occurs in $r$ and such that $\overline{\ell}$ occurs in some $r_i$. 

Finally, since $\Umc\cap\{\mi{fix}(\lambda)\mid \lambda\in\conf_i\}=\{B_i\}$, $r_i$ does not contain any $\overline{\ell_j}$ with $1\leq j\neq i\leq n$ (otherwise $B_j\in \Umc\cap\{\mi{fix}(\lambda)\mid \lambda\in\conf_i\}=\{B_i\}$).

\item Since $\eta$ is closed under resolution, $gr_\Dmc(\eta)$ contains AICs $s_1,\dots, s_n$ with the following literals:
\begin{itemize}
\item  $\aiclits(s_1)=(\aiclits(r)\cup \aiclits(r_1))\setminus\{\ell_1,\overline{\ell_1}\}$
\item $\dots$
\item $\aiclits(s_n)=(\aiclits(r)\cup \aiclits(r_1)\cup\dots\cup \aiclits(r_n))\setminus\{\ell_1, \overline{\ell_1},\dots, \ell_n,\overline{\ell_n}\}$
\end{itemize} 
Moreover, since $\eta$ preserves actions under resolution and $A$ is an update action of $r$, then $A$ is an update action of each of the $s_i$ ($1\leq i\leq n$). 
Let $r_A=s_n$. 
 
\item Since $\Dmc$ satisfies all literals in $\aiclits(r)\setminus\{\ell_1,\dots, \ell_n\}$ and all literals in $\aiclits(r_1),\dots, \aiclits(r_n)$, it follows that $\Dmc\not\models r_A$.

\item Since $\Dmc\circ\Umc\setminus\{A\}\not\models r$, then $\Dmc\circ\Umc\setminus\{A\}$ satisfies all literals in $\aiclits(r)$.

Moreover, for every $1\leq i\leq n$, $\Dmc$ satisfies all literals in $\aiclits(r_i)$ and $\Umc\cap\{\mi{fix}(\lambda)\mid \lambda\in \aiclits(r_i)\}=\{B_i\}$, so that $\Umc$ does not change the value of the literals in $\aiclits(r_i)\setminus\overline{\ell_i}$. Hence $\Dmc\circ\Umc\setminus\{A\}$ satisfies all literals in $\aiclits(r_i)\setminus\{\overline{\ell_i}\}$.

It follows that $\Dmc\circ\Umc\setminus\{A\}$ satisfies all literals of $r_A$. Hence $\Dmc\circ\Umc\setminus\{A\}\not\models r_A$. 
\end{itemize}
 We conclude that there exists $r_A\in gr_\Dmc(\eta)$ such that $A$ is an update action of~$r_A$, $\Dmc\circ\Umc\setminus\{A\}\not\models r_A$, and $\Dmc\not\models r_A$. 
\end{proof}

\ClosedFaithfulCollapse*
\begin{proof}
\noindent$(\foundups{\Dmc,\eta}=\justifups{\Dmc,\eta})$ 
The inclusion $\justifups{\Dmc,\eta}\subseteq\foundups{\Dmc,\eta}$ is known in general. 
We show that when $\eta$ is closed under resolution  and preserves actions under resolution, $\foundups{\Dmc,\eta}\subseteq\justifups{\Dmc,\eta}$. 
Let $\Umc\in\foundups{\Dmc,\eta}$. 
\begin{itemize}
\item Since $\Umc$ is an r-update for $\Dmc$ \wrt $\eta$, by Lemma~\ref{lem:update-repair-closed}, $ ne(\Dmc,\Dmc\circ\Umc)\cup\Umc$ is closed under $\eta$.

\item Let $\Umc'\subsetneq\Umc$ and assume for a contradiction that $ ne(\Dmc,\Dmc\circ\Umc)\cup\Umc'$ is closed under $\eta$. 
\begin{itemize}
\item Let $A\in\Umc\setminus\Umc'$ and let $\ell_A$ be the literal that is fixed by $A$ (\ie $\ell_A=\alpha$ if $A=-\alpha$ and $\ell_A=\neg\alpha$ if $A=+\alpha$). 
\item Since $\Umc$ is founded and $\eta$ is closed under resolution  and preserves actions under resolution, by Lemma~\ref{lem:founded-conf-closed-min-faith}, there exists $r\in gr_\Dmc(\eta)$ such that $A$ is an update action of~$r$, $\Dmc\circ\Umc\setminus\{A\}\not\models r$ and $\Dmc\not\models r$. 
\item Let $\ell$ be a non-updatable literal of $r$. Since $\ell_A$ is an updatable literal of $r$, $\ell\neq \ell_A$, and since both $\Dmc$ and $\Dmc\circ\Umc\setminus\{A\}$ satisfies all literals of $r$, it follows that both $\Dmc$ and $\Dmc\circ\Umc$ satisfies $\ell$. 
It follows that both $\Dmc\circ\Umc$ and $\Dmc$ satisfy all non-updatable literals of $r$.

\item Hence $ne(\Dmc,\Dmc\circ\Umc)=\{+\alpha\mid \alpha\in\Dmc\cap(\Dmc\circ\Umc)\}\cup\{-\alpha \mid \alpha\notin\Dmc\cup (\Dmc\circ\Umc), \alpha\in\factset\}$, satisfies all non-updatable literals of $r$, and so does $ne(\Dmc,\Dmc\circ\Umc)\cup\Umc'$ (recall that $\Umc'\subseteq\Umc$ so that $ne(\Dmc,\Dmc\circ\Umc)\cup\Umc'$ is consistent). 
\item Since we assumed that $ ne(\Dmc,\Dmc\circ\Umc)\cup\Umc'$ is closed under $\eta$, then $ ne(\Dmc,\Dmc\circ\Umc)\cup\Umc'$ must contain an update action $B$ of $r$. Moreover, since $A\notin\Umc'$ and $A\notin ne(\Dmc,\Dmc\circ\Umc)$ (by minimality of the r-update $\Umc$), then $B\neq A$. 
\item Since $ne(\Dmc,\Dmc\circ\Umc)\cup\Umc'\subseteq ne(\Dmc,\Dmc\circ\Umc)\cup\Umc$, then $B\in ne(\Dmc,\Dmc\circ\Umc)\cup\Umc$ which contradicts the fact that $\Dmc\circ\Umc\setminus\{A\}\not\models r$. 
\end{itemize}
Thus $ ne(\Dmc,\Dmc\circ\Umc)\cup\Umc$ is a minimal set of update actions that is closed under $\eta$ and contains $ne(\Dmc,\Dmc\circ\Umc)$.
\end{itemize}
Hence $\Umc\in\justifups{\Dmc,\eta}$.
\smallskip

\noindent$(\foundups{\Dmc,\eta}=\groundups{\Dmc,\eta})$ 
The inclusion $\groundups{\Dmc,\eta}\subseteq\foundups{\Dmc,\eta}$ is known in general. 
We show that when $\eta$ is closed under resolution and preserves actions under resolution, $\foundups{\Dmc,\eta}\subseteq\groundups{\Dmc,\eta}$. 
Let $\Umc\in\foundups{\Dmc,\eta}$. 
By Proposition~\ref{prop:new-founded-grounded}, $\Umc$  is grounded if and only if  
it is an r-update for $\Dmc$ \wrt $\eta[\Umc]$ where $\eta[\Umc]$ is the set of AICs derived from $gr_\Dmc(\eta)$ by deleting update actions not occurring in $\Umc$ and AICs whose update actions have all been deleted. 
Assume for a contradiction that $\Umc$ is not grounded. 
\begin{itemize}
\item Since $\Dmc\circ\Umc\models \eta[\Umc]$, this means that there exists $\Umc'\subsetneq\Umc$ such that $\Dmc\circ\Umc'\models \eta[\Umc]$. 

\item Let $A\in \Umc\setminus\Umc'$. Since $\Umc$ is founded and $\eta$ is closed under resolution and preserves actions under resolution, by Lemma~\ref{lem:founded-conf-closed-min-faith}, there exists $r_A\in gr_\Dmc(\eta)$ such that $A$ is an update action of~$r_A$, $\Dmc\circ\Umc\setminus\{A\}\not\models r_A$ and $\Dmc\not\models r_A$. 
\item Since $A\in\Umc$, it follows that $r_A\in \eta[\Umc]$. 
\item Moreover, since $\Dmc$ satisfies every literal of $r_A$ and so does $\Dmc\circ\Umc\setminus\{A\}$, $A$ is the only update action of $\Umc$ that falsifies a literal of $r_A$. Hence $\Dmc\circ\Umc'\not\models r_A$. This contradicts  $\Dmc\circ\Umc'\models \eta[\Umc]$.
\end{itemize}
Hence $\Umc\in \groundups{\Dmc,\eta}$.
\smallskip

\noindent$(\foundups{\Dmc,\eta}\subseteq\wellfoundups{\Dmc,\eta})$ 
We show that when $\eta$ is closed under resolution and preserves actions under resolution, $\foundups{\Dmc,\eta}\subseteq\wellfoundups{\Dmc,\eta}$. 
Let $\Umc\in\foundups{\Dmc,\eta}$ and $\Umc=\{A_1,\dots,A_n\}$. 
\begin{itemize}
\item For every $1\leq i\leq n$, since $\Umc$ is founded and $\eta$ is closed under resolution and preserves actions under resolution, by Lemma~\ref{lem:founded-conf-closed-min-faith}, there exists $r_i\in gr_\Dmc(\eta)$ such that $A_i$ is an update action of~$r_i$, $\Dmc\circ\Umc\setminus\{A_i\}\not\models r_i$ and $\Dmc\not\models r_i$. 
\begin{itemize}
\item Since $\Dmc\not\models r_i$, $\Dmc$ satisfies every literal of $r_i$. 
\item Since $\Dmc\circ\Umc\setminus\{A_i\}\not\models r_i$, it follows that $\{A_1,\dots, A_{i-1}\}\subseteq\Umc\setminus\{A_i\}$ does not contain any update action that falsifies a literal of~$r_i$. 
\end{itemize}
Thus $\Dmc\circ\{A_1,\dots, A_{i-1}\}\not\models r_i$.
\end{itemize}
Hence $\Umc\in\wellfoundups{\Dmc,\eta}$.
\end{proof}

\StrongerMinBod*
\begin{proof}
Observe that $\ups{\Dmc, \minnongr(\eta)}=\ups{\Dmc,AN(\eta)}$: Since $ \minnongr(\eta)\subseteq AN(\eta)$, $\Dmc\circ\Umc\models AN(\eta)$ implies $\Dmc\circ\Umc\models  \minnongr(\eta)$, and for every $r\in AN(\eta)\setminus \minnongr(\eta)$, there exists $r'\in \minnongr(\eta)$ such that $\aiclits(r')\subseteq\aiclits(r)$, so $\Dmc\circ\Umc\models  \minnongr(\eta)$ implies $\Dmc\circ\Umc\models AN(\eta)$. 

We first show that for $X\in\{\mi{Found},\mi{WellFound},\mi{Ground}, \mi{Just}\}$,  
$\xups{\Dmc,AN(\eta)}=\xups{\Dmc,\minnongr(\eta)}$. 
\begin{itemize}
\item Let $\Umc\in\foundups{\Dmc,AN(\eta)}$. 
Let $A\in\Umc$. There exists $r\in AN(\eta)$ such that $A\in\aicup(r)$ and $\Dmc\circ\Umc\setminus\{A\}\not\models r$, hence for every $\ell\in\aiclits(r)$, $\Dmc\circ\Umc\setminus\{A\}\models\ell$. 
There exists $r'\in \minnongr(\eta)$ such that $\aiclits(r')\subseteq\aiclits(r)$, so that $\Dmc\circ\Umc\setminus\{A\}\models\ell$ for every $\ell\in\aiclits(r')$, \ie $\Dmc\circ\Umc\setminus\{A\}\not\models r'$. 
Since $\eta$ preserves actions under strengthening, $A\in\aicup(r')$. Hence $\Umc\in\foundups{\Dmc, \minnongr(\eta)}$. 

Let $\Umc\in\foundups{\Dmc, \minnongr(\eta)}$. 
Let $A\in\Umc$. There exists $r\in \minnongr(\eta)\subseteq AN(\eta)$ such that $A\in\aicup(r)$ and $\Dmc\circ\Umc\setminus\{A\}\not\models r$. Hence $\Umc\in\foundups{\Dmc,AN(\eta)}$.

\item Let $\Umc\in\wellfoundups{\Dmc,AN(\eta)}$. 
There exists a sequence of actions $A_1,\dots,A_n$ such that $\Umc=\{A_1,\dots,A_n\}$, and for every $1\leq i\leq n$, there exists $r_i\in AN(\eta)$ such that $A_i\in\aicup(r_i)$ and $\Dmc\circ\{A_1,\dots, A_{i-1}\}\not\models r_i$. 
For $1\leq i\leq n$, there exists $r_i'\in \minnongr(\eta)$ such that $\aiclits(r_i')\subseteq\aiclits(r_i)$, so that $\Dmc\circ\{A_1,\dots, A_{i-1}\}\not\models r_i'$. 
Since $\eta$ preserves actions under strengthening, $A_i\in\aicup(r_i')$. Hence $\Umc\in\wellfoundups{\Dmc, \minnongr(\eta)}$. 

Let $\Umc\in\wellfoundups{\Dmc, \minnongr(\eta)}$. 
There exists a sequence of actions $A_1,\dots,A_n$ such that $\Umc=\{A_1,\dots,A_n\}$, and for every $1\leq i\leq n$, there exists $r_i\in  \minnongr(\eta)\subseteq AN(\eta)$ such that $A_i\in\aicup(r_i)$ and $\Dmc\circ\{A_1,\dots, A_{i-1}\}\not\models r_i$. 
Hence $\Umc\in\wellfoundups{\Dmc,AN(\eta)}$. 

\item Let $\Umc\in\groundups{\Dmc,AN(\eta)}$. 
For every $\Vmc\subsetneq\Umc$, there exists $r_N\in N(\eta)$ such that $\Dmc\circ\Vmc\not\models r_N$ and the (only) update action $A$ of $r_N$ is in $\Umc\setminus\Vmc$, \ie  there exists $r\in AN(\eta)$ such that $\Dmc\circ\Vmc\not\models r$, $A\in\aicup(r)$ and $A\in\Umc\setminus\Vmc$. 
There exists $r'\in \minnongr(\eta)$ such that $\aiclits(r')\subseteq\aiclits(r)$, so that  $\Dmc\circ\Vmc\not\models r'$. Since $\eta$ preserves actions under strengthening, $A\in\aicup(r')$. 
It follows that there exists $r'_N\in N( \minnongr(\eta))$ such that $\Dmc\circ\Vmc\not\models r'_N$ and the (only) update action $A$ of $r'_N$ is in $\Umc\setminus\Vmc$. 
Hence $\Umc\in\groundups{\Dmc, \minnongr(\eta)}$. 

Let $\Umc\in\groundups{\Dmc, \minnongr(\eta)}$. 
For every $\Vmc\subsetneq\Umc$, there exists $r\in  \minnongr(\eta)\subseteq AN(\eta)$ such that $\Dmc\circ\Vmc\not\models r$ and an update action of $r$ is in $\Umc\setminus\Vmc$.
Hence $\Umc\in\groundups{\Dmc,AN(\eta)}$. 

\item Let $\Umc\in\justifups{\Dmc,AN(\eta)}$. 
Since $\Dmc\circ\Umc\models  \minnongr(\eta)$, by Lemma~\ref{lem:update-repair-closed}, $ne(\Dmc,\Dmc\circ\Umc)\cup\Umc$ is closed under $ \minnongr(\eta)$. 
Let $\Umc'\subsetneq\Umc$ and assume for a contradiction that $ne(\Dmc,\Dmc\circ\Umc)\cup\Umc'$ is closed under $ \minnongr(\eta)$. Let $r\in AN(\eta)$ such that $ne(\Dmc,\Dmc\circ\Umc)\cup\Umc'$ satisfy every non-updatable literal of $r$. Since $\eta$ preserves actions under strengthening, there exists $r'\in  \minnongr(\eta)$ such that $\aiclits(r')\subseteq\aiclits(r)$ and $\aicup(r)\subseteq\aicup(r')$, so that the non-updatable literals of $r'$ are also non-updatable literals of $r$. Since $ne(\Dmc,\Dmc\circ\Umc)\cup\Umc'$ is closed under $ \minnongr(\eta)$, it contains an update action $A$ of $r'$. If $A$ is not an update action of $r$, the literal $\ell_A$ set to false by $A$ is a non-updatable literal of $r$ not satisfied by $ne(\Dmc,\Dmc\circ\Umc)\cup\Umc'$, which contradicts the definition of $r$. Hence $A$ is an update action of $r$. It follows that $ne(\Dmc,\Dmc\circ\Umc)\cup\Umc'$ is closed under $AN(\eta)$. This contradicts the fact that $ ne(\Dmc,\Dmc\circ\Umc)\cup\Umc$ is a minimal set of update actions that is closed under $AN(\eta)$ and contains $ne(\Dmc,\Dmc\circ\Umc)$. 
We conclude that $ne(\Dmc,\Dmc\circ\Umc)\cup\Umc'$ is not closed under $ \minnongr(\eta)$. Hence $\Umc\in\justifups{\Dmc, \minnongr(\eta)}$. 

Let $\Umc\in\justifups{\Dmc, \minnongr(\eta)}$. 
Since $\Dmc\circ\Umc\models AN(\eta)$, by Lemma~\ref{lem:update-repair-closed}, $ne(\Dmc,\Dmc\circ\Umc)\cup\Umc$ is closed under $AN(\eta)$. 
Let $\Umc'\subsetneq\Umc$ and assume for a contradiction that $ne(\Dmc,\Dmc\circ\Umc)\cup\Umc'$ is closed under $AN(\eta)$. 
 Let $r\in  \minnongr(\eta)\subseteq AN(\eta)$ such that $ne(\Dmc,\Dmc\circ\Umc)\cup\Umc'$ satisfy every non-updatable literal of $r$. Since $ne(\Dmc,\Dmc\circ\Umc)\cup\Umc'$ is closed under $AN(\eta)$, it contains an update action of $r$. 
 Hence $ne(\Dmc,\Dmc\circ\Umc)\cup\Umc'$ is closed under $ \minnongr(\eta)$, which contradicts $\Umc\in\justifups{\Dmc, \minnongr(\eta)}$. 
 It follows that $ne(\Dmc,\Dmc\circ\Umc)\cup\Umc'$ is not closed under $AN(\eta)$. Hence $\Umc\in\justifups{\Dmc,AN(\eta)}$.
\end{itemize}

Finally, for $X\in\{\mi{Found},\mi{WellFound},\mi{Ground}\}$, since $\xups{\Dmc,\eta}=\xups{\Dmc,N(\eta)}$, $\xups{\Dmc,AN(\eta)}=\xups{\Dmc,N(AN(\eta))}$ and $N(AN(\eta))=N(\eta)$, then $\xups{\Dmc,\eta}=\xups{\Dmc,AN(\eta)}=\xups{\Dmc,\minnongr(\eta)}$. 
\end{proof}

\subsection{Proofs for Section \ref{subsec:aics-to-prio}}

\paragraph{Reduction from AICs to prioritized databases} 
Recall that given a set $\eta$ of AICs \emph{closed under resolution that preserves actions under resolution and under strengthening} and a database $\Dmc$, we take $\constraints_\eta=\{\tau_r \mid r \in \eta\}$ and define $\succ_\eta$ so that
$\lambda\succ_\eta\mu$ iff 
\begin{itemize}
\item there exists $r\in \mingr(\eta)$ such that $\Dmc\not\models r$, $\{\lambda,\mu\}\subseteq  \aiclits(r)$, and $\mi{fix}(\mu) \in \aicup(r)$; 
and
\item for every $r\in \mingr(\eta)$ such that $\Dmc\not\models r$ and $\{\lambda,\mu\}\subseteq \aiclits(r)$, $\mi{fix}(\lambda) \not \in \aicup(r)$,
\end{itemize} 
where $\mingr(\eta) = \{r \in gr_\Dmc(\eta) \mid \text{ there is no }   r'\in gr_\Dmc(\eta) $ $ \text{with } \aiclits(r') \subsetneq \aiclits(r)\}$.
As $\eta$ is closed under resolution, $\conflicts{\Dmc,\constraints_\eta}=\{\aiclits(r)\mid r\in \mingr(\eta), \Dmc\not\models r\}$.

We start with the proof of the inclusion that holds between founded repairs of $\Dmc$ \wrt $\eta$ and Pareto-optimal repairs of $\Dmc^{\constraints_\eta}_{\succ_\eta}$ in the general case (with potentially non-binary conflicts).

\ReductionAICsPrioGeneral*
\begin{proof}
By Proposition~\ref{prop:founded-grounded-closed-min-faith}, $\justifreps{\Dmc,\eta}=\groundreps{\Dmc,\eta}=\foundreps{\Dmc,\eta}$. We show that $\foundreps{\Dmc,\eta}\subseteq\deltapreps{\Dmc^{\constraints_\eta}_{\succ_\eta}}$. 
Let $\Umc\in\foundups{\Dmc,\eta}$ and $\Rmc=\Dmc\circ\Umc$. 

\begin{itemize}
\item Since $\Umc\in\ups{\Dmc,\eta}$ and ${\constraints_\eta}$ contains the universal constraints that correspond to the AICs in $\eta$, then $\Rmc\in\deltareps{\Dmc,{\constraints_\eta}}$.

\item Assume for a contradiction that $\Rmc\notin\deltapreps{\Dmc^{\constraints_\eta}_{\succ_\eta}}$: 
There exists $\Rmc'$ consistent \wrt ${\constraints_\eta}$ such that there is $\lambda\in\comp{\Dmc}{\Rmc'}\setminus\comp{\Dmc}{\Rmc}$ with $\lambda{\succ_\eta}\mu$ for every $\mu\in \comp{\Dmc}{\Rmc}\setminus\comp{\Dmc}{\Rmc'}$. 
\begin{itemize}
\item If $\lambda=\alpha\in\Dmc$, since $\alpha\notin\comp{\Dmc}{\Rmc}$, then $\alpha\notin\Rmc$ so $\mi{fix}(\lambda)=-\alpha$ is in $\Umc$. If $\lambda=\no\alpha$ for some $\alpha\notin\Dmc$, since $\no\alpha\notin\comp{\Dmc}{\Rmc}$, then $\alpha\in\Rmc$ so $\mi{fix}(\lambda)=+\alpha$ is in $\Umc$. Hence $\mi{fix}(\lambda)\in\Umc$.

\item By Lemma~\ref{lem:founded-conf-closed-min-faith}, since $\eta$ is closed under resolution and preserves actions under resolution, $\Umc\in \foundups{\Dmc,\eta}$ and $\mi{fix}(\lambda)\in\Umc$, then there exists $r\in gr_\Dmc(\eta)$ such that $\Dmc\not\models r$, $\Dmc\circ\Umc\setminus\{\mi{fix}(\lambda)\}\not\models r$ and $\mi{fix}(\lambda)\in\aicup(r)$.

\item Since $\eta$ preserves actions under strengthening, we can choose $r$ such that $r\in  \mingr(\eta)$: Indeed, if $r\notin  \mingr(\eta)$, there exists $r'\in  \mingr(\eta)$ such that $\aiclits(r')\subseteq\aiclits(r)$, so that $\Dmc\not\models r'$ and $\Dmc\circ\Umc\setminus\{\mi{fix}(\lambda)\}\not\models r'$, and by preservation of actions under strengthening, $\mi{fix}(\lambda)\in\aicup(r')$.

\item By Lemma~\ref{lem:closed-under-res-min-conflicts}, since $\eta$ is closed under resolution, $r\in \mingr(\eta)$, and $\Dmc\not\models r$, then $\aiclits(r)\in\deltaconflicts{\Dmc,{\constraints_\eta}}$.
Note that this implies that $\aiclits(r)\subseteq\litset$.  
  
\item Since $\Dmc\circ\Umc\setminus\{\mi{fix}(\lambda)\}\not\models r$ and $\aiclits(r)\subseteq\litset$, then $\aiclits(r)\subseteq\comp{\Dmc}{\Dmc\circ\Umc\setminus\{\mi{fix}(\lambda)\}}=\comp{\Dmc}{\Rmc}\cup\{\lambda\}$. 

\item Since $\mi{fix}(\lambda)\in\aicup(r)$, $r\in \mingr(\eta)$ and $\Dmc\not\models r$, then by construction of ${\succ_\eta}$, for every $\mu\in \aiclits(r)$, $\lambda\not{\succ_\eta}\mu$ (since $\lambda{\succ_\eta}\mu$ implies that for every $r\in \mingr(\eta)$ such that $\{\lambda,\mu\}\subseteq \aiclits(r)$ and $\Dmc\not\models r$, $\mi{fix}(\lambda)\notin\aicup(r)$). 
\item Since $ \comp{\Dmc}{\Rmc}\setminus\comp{\Dmc}{\Rmc'}\subseteq\{\mu\mid\lambda{\succ_\eta}\mu\}$, 
it follows that $\aiclits(r)\subseteq \comp{\Dmc}{\Rmc'}$.  Hence $\Rmc'\not\models r$, which contradicts $\Rmc'\models\Cmc_\eta$. 
\end{itemize}
It follows that $\Rmc\in\deltapreps{\Dmc^{\constraints_\eta}_{\succ_\eta}}$. \qedhere
\end{itemize}
\end{proof}

We show the inverse direction in the restricted case where the size of the conflicts is bounded by $2$.

\ReductionAICsPrioBinary*
\begin{proof}
By Proposition~\ref{prop:founded-grounded-closed-min-faith}, $\justifreps{\Dmc,\eta}=\groundreps{\Dmc,\eta}=\foundreps{\Dmc,\eta}\subseteq\wellfoundreps{\Dmc,\eta}$ and  
by Proposition~\ref{prop:AICs-prio-non-binary}, $\foundreps{\Dmc,\eta}\subseteq\deltapreps{\Dmc^{\constraints_\eta}_{\succ_\eta}}$. 
It remains to show that $\deltapreps{\Dmc^{\constraints_\eta}_{\succ_\eta}}\subseteq\foundreps{\Dmc,\eta}$. Let $\Rmc\in\deltapreps{\Dmc^{\constraints_\eta}_{\succ_\eta}}$ and $\Umc=\{-\alpha\mid \alpha\in\Dmc\setminus\Rmc\}\cup\{+\alpha\mid \alpha\in\Rmc\setminus\Dmc\}$ be the consistent set of update actions such that $\Dmc\circ\Umc=\Rmc$. 

\begin{itemize}
\item  Since $\Rmc\in\deltareps{\Dmc,{\constraints_\eta}}$ and ${\constraints_\eta}$ contains the universal constraints that correspond to the AICs in $\eta$, then $\Umc\in\ups{\Dmc,\eta}$.

\item Assume for a contradiction that $\Umc$ is not founded. There exists $A\in\Umc$ such that for every $r\in gr_\Dmc(\eta)$ such that $A\in\aicup(r)$, $\Dmc\circ\Umc\setminus\{A\}\models r$. Let $\ell_A$ denote the literal that $A$ falsifies.
\begin{itemize}
\item  Let $\conf\in\deltaconflicts{\Dmc,{\constraints_\eta}}$ such that $\ell_A\in\conf$ and $\ell_A\not{\succ_\eta}\lambda$ for every $\lambda\in\conf$. Since $\eta$ is closed under resolution, by Lemma~\ref{lem:closed-under-res-min-conflicts}, $\conf=\aiclits(r_\conf)$ for some $r_\conf\in  \mingr(\eta)$ such that $\Dmc\not\models r_\conf$. 
\begin{itemize}
\item If $A\in\aicup(r_\conf)$, then $\Dmc\circ\Umc\setminus\{A\}\models r_\conf$ by assumption on $A$. 
\item Otherwise, if $A\notin\aicup(r_\conf)$, let $\lambda\in\conf$ such that $\mi{fix}(\lambda)\in\aicup(r_\conf)$ (there must be at least one such $\lambda$ otherwise $r_\conf$ has no update actions).
Since $\ell_A\not{\succ_\eta}\lambda$, by construction of ${\succ_\eta}$:   
\begin{enumerate}
\item either for every $r\in  \mingr(\eta)$ such that $\{\ell_A,\lambda\}\subseteq \aiclits(r)$ and $\Dmc\not\models r$, 
$\mi{fix}(\lambda)\notin\aicup(r)$; 
\item or there is $r\in  \mingr(\eta)$ such that $\{\ell_A,\lambda\}\subseteq \aiclits(r)$, $\Dmc\not\models r$, and $A\in\aicup(r)$.
\end{enumerate}
Since $\{\ell_A,\lambda\}\subseteq \aiclits(r_\conf)$, $\Dmc\not\models r_\conf$, and $\mi{fix}(\lambda)\in\aicup(r_\conf$), we are not in case (1) so we are in case (2). 
Hence, there is $r_\lambda\in  \mingr(\eta)$ such that $\{\ell_A,\lambda\}\subseteq \aiclits(r_\lambda)$, $\Dmc\not\models r_\lambda$, and $A\in\aicup(r_\lambda)$. 

Since the size of the conflicts is bounded by $2$, $r_\lambda$ and $r_\conf$ have the same body: $\ell_A\wedge\lambda$. 
Hence, since $A\in\aicup(r_\lambda)$, $\Dmc\circ\Umc\setminus\{A\}\models r_\lambda$ by assumption on $A$. 
Hence $\Dmc\circ\Umc\setminus\{A\}\models r_\conf$. 
\end{itemize}
 It follows that $\conf=\aiclits(r_\conf)\not\subseteq \comp{\Dmc}{\Dmc\circ\Umc\setminus\{A\}}=\comp{\Dmc}{\Rmc}\cup\{\ell_A\}$. 

\item We have thus shown that for every $\conf\in\deltaconflicts{\Dmc,{\constraints_\eta}}$ such that $\ell_A\in\conf$, either $\conf$ contains some $\lambda$ such that $\ell_A{\succ_\eta}\lambda$, or $\conf\not\subseteq\comp{\Dmc}{\Rmc}\cup\{\ell_A\}$. 

\item Let $\Rmc'=\restr{\Dmc}{\comp{\Dmc}{\Rmc}\cup\{\ell_A\}\setminus\{\mu\mid\ell_A{\succ_\eta}\mu\}}$, so that $\comp{\Dmc}{\Rmc'}=\comp{\Dmc}{\Rmc}\cup\{\ell_A\}\setminus\{\mu\mid\ell_A{\succ_\eta}\mu\}$ by~Lemma~\ref{lem:compl-restr}. 
Since $\{\mu\mid\ell_A{\succ_\eta}\mu\}$ intersects every conflict $\conf\in\deltaconflicts{\Dmc,{\constraints_\eta}}$ such that $\ell_A\in\conf$ and $\conf\subseteq\comp{\Dmc}{\Rmc}\cup\{\ell_A\}$, then there is no conflict included in $\comp{\Dmc}{\Rmc'}$. 
 By Lemma~\ref{lem:pareto-improvement-if-one-beat-all}, it follows that $\Rmc\notin\deltapreps{\Dmc^{\constraints_\eta}_{\succ_\eta}}$. 
\end{itemize}
Hence $\Umc\in\foundups{\Dmc,\eta}$.\qedhere
\end{itemize}
\end{proof}

\end{document}